\documentclass{amsart}
\usepackage{fullpage,graphicx,subfigure,mathpazo,color}
\usepackage{amsmath,amscd,tikz,mathrsfs}
\usepackage[normalem]{ulem}
\usepackage{amsmath}
\usepackage{epstopdf}
\usepackage{tikz}
\usepackage{setspace}
\newcommand{\ii}{\mathrm{i}}
\newcommand{\ee}{\mathrm{e}}

\newcommand{\T}{\mathrm{T}}

\newcommand{\cn}{\mathrm{cn}}

\newtheorem{rhp}{Riemann-Hilbert Problem}
\newtheorem{theorem}{Theorem}
\newtheorem{lemma}{Lemma}

\newtheorem{definition}{Definition}

\usepackage{graphicx}      
\usepackage{titlesec}

\titleformat{\section}{\centering\LARGE\bfseries}{\thesection}{1em}{}
\titleformat{\subsection}{\Large\bfseries}{\thesubsection}{1em}{}
\begin{document}

\title{Large order breathers of the nonlinear Schr\"odinger equation}

\author{Liming Ling}
\address{School of Mathematics, South China University of Technology, Guangzhou, China 510641}
\email{linglm@scut.edu.cn}
\author{Xiaoen Zhang}
\address{School of Mathematics, South China University of Technology, Guangzhou, China 510641}
\email{zhangxiaoen@scut.edu.cn}

\begin{abstract}
Multi-soliton and high-order soliton solutions are two type of famous ones in the integrable focusing nonlinear Schr\"odinger equation. The dynamics of multi-soliton was well known to us since 70s of the last century by the determinant analysis. However, there is few progress on the high-order solitons. In this work, we would like to analyze the large order asymptotics for the high-order breathers, which are special cases of double high-order solitons with the same velocity to the nonlinear Schr\"odinger equation. To analyze the large order dynamics, we first convert the representation of Darboux transformation into a framework of Riemann-Hilbert problem. Then we show that there exist five distinct asymptotic regions by the Deift-Zhou nonlinear steepest descent method. It is very interesting that a novel genus-three asymptotic region is first found in the large order asymptotics to large high-order breathers, which enriches the dynamic behaviors in the field of large order solitons. All results to the asymptotic analysis are verified by the numerical method.


\end{abstract}

\date{\today}

\maketitle
\section{Introduction}
The well known one dimensional focusing nonlinear Schr\"{o}dinger equation
\begin{equation}\label{eq:nls-eq}
\ii q_{t}+\frac{1}{2}q_{xx}+\left|q\right|^2q=0,
\end{equation}
is a complete integrable model and admits lots of exact soliton solutions. For instance, the multi-soliton and high-order soliton. In the terminology of inverse scattering transform, the multi-solitons correspond to the scattering data with distinct pairs of simple pole; but the high-order ones correspond to the scattering data with complex-conjugated pairs of high-order poles.
The real and imaginary part of pole determine the velocity and amplitude of soliton respectively. For the multi-soliton with different velocity, their large time asymptotics can be performed by the determinant analysis, which shows that the $N$-soliton with distinct real poles can be decomposed into $N$ single solitons. If the multi-soliton has the same velocity, it will form a breather, which can be regarded as a special type of periodic solution. As the period tends to infinity, the breathers will form the second-order soliton (alias the double poles soliton). The second order soliton was firstly discovered in the original work of Zaharov and Shabat as early as 1972 \cite{shabat1972exact}, and who claimed that the dynamics of second-order soliton separated to each other with the distance $\ln(4\eta^2 t)$. A systematic study on the high-order solitons was carried out by Olmedilla \cite{OLMEDILLA1987330}. Based on the computer computations, he proposed a conjecture for the single arbitrary order soliton, which was proved by Schiebold after $30$ years ago \cite{schiebold2017asymptotics} based on the operator theoretic approach and determinant analysis. Another method based on the generalized Darboux transformation and determinant analysis was proposed by two authors of this work to the Hirota equation \cite{zhang2021asymptotic}.

A pioneered work for the infinite order soliton of KdV equation was proposed by Gesztesy, Karwowski and Zhao \cite{Fritz-1992} by using the spectral theory of operators. An alternative method had been proposed by Shabat through the Darboux transformation \cite{Shabat_1992}. Actually, the relative work on the infinite solitons can retrospect to the pioneering work by McKean and Trubowitz \cite{CPAM-1975}, who considered the infinitely many branch points of linear Schr\"odinger spectral problem (shrinking branch cuts into points in the hyperelliptic function theory can yield the solitons). For the nonlinear Schr\"odinger equation, the infinite many solitons were first obtained by Kamvissis \cite{JMP-1995}. Another independent works on the countable superposition of solitons to nonlinear Schr\"odinger equation was given by Schiebold \cite{Schiebold_2010}. In recent years, Grava et.al. studied the asymptotics of the KdV soliton and the nonlinear Schr\"{o}dinger soliton by using the Riemann-Hilbert method as $N\to\infty$ \cite{bertola2021soliton,girotti2021rigorous}. The results show that the soliton gas to these two equations can be asymptotically step-like oscillatory for $x\to\pm\infty$ respectively. While the difference is that the nonlinear Schr\"{o}dinger equation will exhibit soliton shielding effect, but the KdV equation will not.

For the nonlinear Schr\"{o}dinger equation, with some special initial data, the long time asymptotics will include some high genus sectors. Such as, through the Riemann-Hilbert method and the Deift-Zhou nonlinear steepest descent method \cite{deift1993steepest}, for a general nonzero background, the long time asymptotics of solution can be given by the Jacobi-Theta function \cite{biondini2017long}. For the nondecaying initial data $q(x,0)=A\ee^{-\ii |\mu|x}$, the solution will be asymptotic to the Theta function attached to a Riemann surface with genus two in the transition region \cite{buckingham2007long}. While for the step-like oscillating background, the asymptotics can be given by the hyperelliptic Theta functions attached to a Riemann surface with genus-three \cite{monvel2022focusing}. In addition, the semi-classical limit to the nonlinear Schr\"{o}dinger equation will also generate high genus region \cite{kamvissis2003semiclassical,lyng2007n,miller2008Riemann}.

Very recently, there are series of progresses on the infinite/large order rogue waves and solitons for the scalar or coupled nonlinear Schr\"odinger equation \cite{Bilman-JNS-2019,Bilman-JDE-2021,Bilman-Duke-2019,BILMAN2022133289,Bilman-arxiv-2021,Ling-arxiv-2021} by the far-field or near field limit and Riemann-Hilbert method. The first recent breakthrough in this aspect is infinite order rogue wave by the near-field limit, in which the infinite order rogue wave is related to the Painl\'eve-III hierarchy \cite{Bilman-Duke-2019}. The far-field asymptotics for single multiple-pole solitons was beautifully analyzed by the Riemann-Hilbert method \cite{Bilman-JNS-2019,Bilman-JDE-2021}. After that, the universal analyzing for the large order soliton or rogue wave was carried out in the literature \cite{BILMAN2022133289, Bilman-arxiv-2021}. The large/infinite asymptotics for the single high-order soliton in the coupled NLS equation was given in our previous work \cite{Ling-arxiv-2021}. In above-mentioned works, all the solutions are related to the single high-order poles under the zero or non-zero background in the framework of inverse scattering transform. Therefore, it is naturally to extend the single high-order poles to the general high-order solitons in large or infinite order asymptotics.

Another motivation for this work is from the study of integrable turbulence and soliton gas
\cite{PRL-2019}, in which they considered ensembles of $N$-soliton ($N$ is large enough) having a particular semiclassical distribution of IST eigenvalues and then they obtained the similar structures resembling the Akhmediev breathers. As $N\to\infty$, the cluster points will appear in a bounded region of spectral plane. The single high-order soliton can be considered as a single cluster point in the upper half plane. The double high-order soliton will correspond to two cluster points in the upper half plane.
Thus, in this work, we would like to analyze the asymptotics of a special double high-order solitons or a high-order breather under the far field limit. As we know that, as two solitons share the same velocity, this soliton will form a bound state, which exhibits the breathing dynamics with the evolution of time \cite{chen1976solitons,doktorov2007dressing,satsuma1974b,shabat1972exact,yang2010nonlinear}, whose dynamics is shown in Fig. \ref{fig:exact-3-d}(a). For the second order breather, as shown in Fig. \ref{fig:exact-3-d}(b), we can observe that two solitons admit the breathing dynamics in the center and another two solitons locating on either sides possess soliton behavior with the logarithm separating velocity. With the increasing of order, it is evidently seen that the dynamics of solitons in the center are differ from the ones on two sides (Fig. \ref{fig:exact-3-d-1}). By the plotting of lower order breathers, compared to the single high-order solitons, we can find that the dynamics of high-order breathers become more diversity and a new different asymptotic region in the center appears. Even though the high-order breather solutions have the determinant representation, it is difficult to analyze or plot by the large order determinant directly. Looking for a replacing way to analyze the dynamics systematically is one of our pursuit. The main tool in this work is using the Riemann-Hilbert representation. Return to the high-order breathers, we will show that the novel breathing solitons in the center part will correspond to the high genus case, which brings challenges for the asymptotic analysis and numeric verification.

To deal with the large order asymptotics of high-order breathers, we first establish the corresponding Riemann-Hilbert problem (RHP) by using the Darboux transformation. When $n$ is large, by using the far-field limit we can get five different asymptotic regions on the $(\chi,\tau)$-plane, which are the exponential-decay region, the algebraic-decay region, the genus-zero region, the genus-one region and the genus-three region respectively. For the exponential-decay region, we can get the asymptotics via the RHP directly. For the algebraic-decay region, the phase term appearing in the RHP is well suit for analyzing through the Deift-Zhou nonlinear steepest descent method directly. While for the left three asymptotic regions, we should introduce three different $g$ functions \cite{deift1997new} depending on different genus of algebraic curves to convert the controlling phase term into the new phase $h(\lambda)$ function.

The innovation of this work contains the following four points. (i) We extend the asymptotics of single high-order solitons to the high-order breathers.
Through the RHP and Deift-Zhou nonlinear steepest descent method, we establish the five different asymptotic regions, which enriches the dynamics of large order solitons further. (ii) We construct the corresponding RHP via the residue and kernel conditions of Darboux transformation directly. By the deformation of contour, we also establish the equivalence of RHP between normalization method \cite{Bilman-2019-CPAM} and direct way. By the RHP of direct way, we can analyze the asymptotic behavior for the exponential-decay region directly without deforming the contour. (iii) Due to the interaction of high-order breathers, it appears a new genus-three region. By introducing an algebraic curve with genus-three, we give the asymptotics expressing with the Riemann-Theta function. To our best knowledge, this is the first result about the asymptotic region of genus-three on the large order solitons. (iv) By choosing fixed $\chi$ and $\tau$, we numerically verify that the asymptotic solutions can match the determinantal ones very well.
\subsection{The Riemann-Hilbert representation of high-order breathers}

The Lax pair of the nonlinear Schr\"{o}dinger equation \eqref{eq:nls-eq} is given by the following linear system
\begin{equation}
\begin{aligned}
\pmb{\Phi}_{x}&=\mathbf{U}(\lambda; x, t)\pmb{\Phi},\quad \mathbf{U}(\lambda; x, t)=-\ii\lambda\sigma_3+\mathbf{Q},\quad \mathbf{Q}=\begin{bmatrix}0&q\\
-q^*&0
\end{bmatrix},\\
\pmb{\Phi}_{t}&=\mathbf{V}(\lambda; x, t)\pmb{\Phi},\quad\pmb{V}(\lambda; x, t)=-\ii\lambda^2\sigma_3+\lambda\mathbf{Q}+\frac{1}{2}\ii\sigma_3\left(\mathbf{Q}_x-\mathbf{Q}^2\right),
\end{aligned}
\end{equation}
where $\sigma_3$ is one of the Pauli matrices
\begin{equation}\label{eq:lax-pair}
\sigma_1=\begin{bmatrix}0&1\\
1&0
\end{bmatrix},\quad \sigma_2=\begin{bmatrix}0&-\ii\\
\ii&0
\end{bmatrix},\quad \sigma_3=\begin{bmatrix}1&0\\
0&-1
\end{bmatrix}.
\end{equation}
Under the zero background (i.e. $q=0$), the fundamental solution of Lax pair \eqref{eq:lax-pair} is
\begin{equation}
\pmb{\Phi}_{{\rm bg}}(\lambda; x, t)=\ee^{-\ii\lambda\left(x+\lambda t\right)\sigma_3}.
\end{equation}
In this paper, we would like to study the asymptotics of the large order breathers via the Riemann-Hilbert method and the Deift-Zhou nonlinear steepest descent method. Following the idea for analyzing the asymptotics of single large order soliton and large order rogue wave, this RHP can be given by using the Darboux transformation. Thus we give a brief review about the Darboux transformation about the nonlinear Schr\"{o}dinger equation. For the AKNS system with $su(2)$ symmetry, the loop group version of Darboux transformation $\mathbf{T}_1^{[1]}(\lambda; x, t)$ is given by Terng and Uhlebeck in \cite{terng2000backlund}:
\begin{equation}
\mathbf{T}_{1}^{[1]}(\lambda; x, t)=\mathbb{I}-\frac{\lambda_1-\lambda_1^*}{\lambda-\lambda_1^*}\mathbf{P}_{1}^{[1]}(x,t), \quad \mathbf{P}_{1}^{[1]}(x,t)=\frac{\pmb{\phi}_1(x,t)\pmb{\phi}_{1}^{\dagger}(x,t)}{\pmb{\phi}_{1}^{\dagger}(x,t)\pmb{\phi}_{1}(x,t)},\quad \pmb{\phi}_{1}(x,t)=\pmb{\Phi}_{{\rm bg}}(\lambda; x, t)\left(c_1, c_2\right)^{\T},
\end{equation}
where $(c_1, c_2)$ are arbitrary complex constants. Moreover, if we want to derive the high-order solitons at the same spectral parameter $\lambda_1$, the classical Darboux transformation can not work any more, we should introduce the generalized Darboux transformation \cite{guo2012nonlinear}. Another recent version of Darboux transformation to construct the high-order solution was proposed by Bilman and Miller in \cite{Bilman-2019-CPAM} based on the robust inverse scattering transform.

Set the second order Darboux transformation $\mathbf{T}_{1}^{[2]}(\lambda; x, t)$ as
\begin{multline}
\mathbf{T}^{[2]}_{1}(\lambda; x, t)=\mathbb{I}-\frac{\lambda_1-\lambda_1^*}{\lambda-\lambda_1^*}\mathbf{P}_{1}^{[2]}(x,t),\quad \mathbf{P}_{1}^{[2]}(x,t)=\frac{\pmb{\phi}_1^{[1]}(x,t)\left(\pmb{\phi}_1^{[1]}(x,t)\right)^{\dagger}}{\left(\pmb{\phi}_1^{[1]}(x,t)\right)^{\dagger}\pmb{\phi}_1^{[1]}(x,t)}, \\ \pmb{\phi}_1^{[1]}(x,t)=\frac{d}{d\lambda}\left(\mathbf{T}_{1}^{[1]}(\lambda; x, t)\pmb{\Phi}_{\rm bg}(\lambda; x, t)\right)\Big{|}_{\lambda=\lambda_1}\left(c_1, c_2\right)^{\T}.
\end{multline}
Continuing the above iteration, we can get the $n$-th order Darboux transformation about single pole case as
\begin{multline}
\mathbf{T}^{[n]}_{1}(\lambda; x, t)=\mathbb{I}-\frac{\lambda_1-\lambda_1^*}{\lambda-\lambda_1^*}\mathbf{P}_{1}^{[n]}(x,t),\quad \mathbf{P}_{1}^{[n]}(x,t)=\frac{\pmb{\phi}_{1}^{[n-1]}(x,t)\left(\pmb{\phi}_{1}^{[n-1]}(x,t)\right)^{\dagger}}{\left(\pmb{\phi}_{1}^{[n-1]}(x,t)\right)^{\dagger}\pmb{\phi}_{1}^{[n-1]}(x,t)},\\
\pmb{\phi}_{1}^{[n-1]}(x,t)=\frac{1}{(n-1)!}\frac{d^{n-1}\left(\mathbf{T}^{[n-1]}_{1}(\lambda; x, t)\mathbf{T}^{[n-2]}_{1}(\lambda; x, t)\cdots\mathbf{T}^{[1]}_{1}(\lambda; x, t)\pmb{\Phi}_{\rm bg}(\lambda; x, t)\right)}{d\lambda^{n-1}}\Bigg{|}_{\lambda=\lambda_1}\left(c_1, c_2\right)^{\T}.
\end{multline}
In this paper, our major aim is studying the large order asymptotics about the high-order breathers, which is a particular one of double high-order solitons with the same velocity. By combining the classical Darboux transformation with $N$ distinct spectral parameters and the generalized Darboux transformation for high-order solitons at the same spectral parameters, we can set the generalized Darboux transformation with two spectral parameters $\lambda_1, \lambda_2$ as
\begin{equation}\label{eq:DT}
\mathbf{T}^{[2n]}(\lambda; x, t)=\mathbf{T}_{2}^{[n]}\mathbf{T}_{2}^{[n-1]}\cdots\mathbf{T}_{2}^{[1]}\mathbf{T}_{1}^{[n]}\mathbf{T}_{1}^{[n-1]}\cdots\mathbf{T}_{1}^{[1]},
\end{equation}
where
\begin{equation}
\mathbf{T}_{i}^{[j]}\equiv\mathbf{T}_{i}^{[j]}(\lambda; x, t), \quad \mathbf{T}_{i}^{[j]}:=\mathbb{I}-\frac{\lambda_i-\lambda_i^*}{\lambda-\lambda_i^*}\mathbf{P}_{i}^{[j]}, \left(i=1,2, j=1,2,\cdots, n\right),
\end{equation}
and
\begin{equation}
\begin{aligned}
\mathbf{P}_{i}^{[j]}&=\frac{\pmb{\phi}_{i}^{[j-1]}\left(\pmb{\phi}_{i}^{[j-1]}\right)^{\dagger}}{\left(\pmb{\phi}_{i}^{[j-1]}\right)^{\dagger}\pmb{\phi}_{i}^{[j-1]}},
\quad \pmb{\phi}_{1}^{[j]}=\frac{1}{j!}\frac{d^{j}\left(\mathbf{T}_{1}^{[j]}\mathbf{T}_{1}^{[j-1]}\cdots\mathbf{T}_{1}^{[1]}\pmb{\Phi}_{\rm bg}(\lambda; x, t)\right)}{d\lambda^{j}}\Bigg{|}_{\lambda=\lambda_1}\left(c_1, c_2\right)^{\T},\\
\pmb{\phi}_{2}^{[j]}&=\frac{1}{j!}\frac{d^{j}\left(\mathbf{T}_{2}^{[j]}\mathbf{T}_{2}^{[j-1]}\cdots\mathbf{T}_{2}^{[1]}\mathbf{T}_{1}^{[n]}\mathbf{T}_{1}^{[n-1]}\cdots\mathbf{T}_{1}^{[1]}\pmb{\Phi}_{\rm bg}(\lambda; x, t)\right)}{d\lambda^j}\Bigg{|}_{\lambda=\lambda_2}\left(c_1, c_2\right)^{\T}.
\end{aligned}
\end{equation}
Following the idea in \cite{ling2016multi}, we can rewrite the Darboux transformation Eq.\eqref{eq:DT} as another equivalent formula, that is
\begin{equation}\label{eq:DT-1}
\mathbf{T}^{[2n]}(\lambda; x, t)=\mathbb{I}+\mathbf{Y}_{2n}\mathbf{M}^{-1}\mathbf{D}\mathbf{Y}_{2n}^{\dagger},\quad \mathbf{M}=\mathbf{X}^{\dagger}\mathbf{S}\mathbf{X},
\end{equation}
where
\begin{equation}
\begin{aligned}
\mathbf{Y}_{2n} &=\left[\pmb{\varphi}_1^{[0]},\pmb{\varphi}_1^{[1]},\cdots,\pmb{\varphi}_1^{[n-1]}, \pmb{\varphi}_2^{[0]},\pmb{\varphi}_2^{[1]},\cdots,\pmb{\varphi}_2^{[n-1]}\right],\\
\mathbf{D}&=\begin{bmatrix}
\mathbf{D}_1 &0\\
0 &\mathbf{D}_2
\end{bmatrix}, \quad
\mathbf{X}=\begin{bmatrix}
\mathbf{X}_1 &0\\
0 &\mathbf{X}_2
\end{bmatrix},\quad
\mathbf{S}=\begin{bmatrix}
\mathbf{S}_{11} & \mathbf{S}_{12}\\
\mathbf{S}_{21} &\mathbf{S}_{22}
\end{bmatrix},\\\mathbf{D}_i&=\begin{bmatrix}
 \frac{1}{\lambda-\lambda_i^*}&0 &\cdots & 0 \\
\frac{1}{(\lambda-\lambda_i^*)^2}&\frac{1}{\lambda-\lambda_i^*} &\cdots & 0 \\
\vdots &\vdots & &\vdots \\
\frac{1}{(\lambda-\lambda_i^*)^{n}}&\frac{1}{(\lambda-\lambda_i^*)^{n-1}} &\cdots &\frac{1}{\lambda-\lambda_i^*} \\
\end{bmatrix},\quad \mathbf{X}_i=\begin{bmatrix}
\pmb{\varphi}_i^{[0]}&\pmb{\varphi}_i^{[1]}&\cdots&\pmb{\varphi}_i^{[n-1]}\\
0&\pmb{\varphi}_i^{[0]}&\cdots&\pmb{\varphi}_i^{[n-2]} \\
\vdots&\vdots&\ddots&\vdots\\
0&0&\cdots& \pmb{\varphi}_i^{[0]}\\
\end{bmatrix},\quad (i=1,2),\\
\mathbf{S}_{i,j}&=\begin{bmatrix}
\binom{0}{0}\frac{\mathbb{I}_2}{\lambda_i^*-\lambda_j}&\binom{1}{0} \frac{\mathbb{I}_2}{(\lambda_i^*-\lambda_j)^2} &\cdots&\binom{n-1}{0}\frac{\mathbb{I}_2}{(\lambda_i^*-\lambda_j)^{n}} \\
\binom{1}{1}\frac{(-1)\mathbb{I}_2}{(\lambda_i^*-\lambda_j)^2}& \binom{2}{1}\frac{(-1)\mathbb{I}_2}{(\lambda_i^*-\lambda_j)^3} &\cdots&\binom{n}{1}\frac{(-1)\mathbb{I}_2}{(\lambda_i^*-\lambda_j)^{n+1}} \\\vdots&\vdots&\ddots&\vdots\\
\binom{n-1}{n-1}\frac{(-1)^{n-1}\mathbb{I}_2}{(\lambda_i^*-\lambda_j)^{n}}& \binom{n}{n-1}\frac{(-1)^{n-1}\mathbb{I}_2}{\lambda_i^*-\lambda_j)^{n+1}} &\cdots&\binom{2n-2}{n-1}\frac{(-1)^{n-1}\mathbb{I}_2}{(\lambda_i^*-\lambda_j)^{2n-1}}\\
\end{bmatrix},\quad (i,j=1,2)
\end{aligned}
\end{equation}
and $\pmb{\varphi}_i^{[k]}=\frac{1}{k!}\left(\frac{\rm d}{{\rm d}\lambda}\right)^{k}\pmb{\varphi}_i(\lambda)\Big|_{\lambda=\lambda_i}$, $\pmb{\varphi}_i(\lambda_i)$ is a special solution of the Lax pair with $\lambda=\lambda_i$, that is $\pmb{\varphi}_i(\lambda_i)=\pmb{\Phi}_{\rm bg}(\lambda_i; x, t)\left(c_1, c_2\right)^{\T}$. Then the high-order breather $q^{[n]}(x,t)$ can be recovered by a limit calculation:
\begin{equation}\label{eq:qn-exact}
q^{[n]}(x,t)=2\ii\lim\limits_{\lambda\to\infty}\lambda\mathbf{T}^{[2n]}_{12}(\lambda; x, t)=2\ii \mathbf{Y}_{2n,2}\mathbf{M}^{-1}\mathbf{Y}_{2n,1}^{\dagger},
\end{equation}
where the subscript $_{2n,1}$ and $_{2n,2}$ denote the first row and second row of $\mathbf{Y}_{2n}$ respectively. In \cite{zhang2020study}, the authors proved the regularities of high order solitons by the inverse scattering transform. For the formula \eqref{eq:qn-exact}, its regularity is evidently because the matrix $\mathbf{M}$ is positive definite.

Without loss of generality, we set $\lambda_1=\ii, \lambda_2=k\ii \,(k>1)$. For $k<1$ case, the solutions can be obtained by the scale invariance of nonlinear Schr\"{o}dinger equation
\begin{equation}
\widetilde{q}(x,t)=\ee^{\epsilon}q\left(x\ee^{\epsilon},t\ee^{2\epsilon}\right).
\end{equation}
When $n=1$,  by the determinant formula \eqref{eq:qn-exact}, the first order breather can be given as
\begin{equation}\label{eq:single-breather}
q^{[1]}=\frac{4c_1c_2^*(k^2-1)\left(k\ee^{2\ii k^2t}\cosh\left(2x+\log\left|\frac{c_1}{c_2}\right|\right)-\ee^{2\ii t}\cosh\left(2kx+\log\left|\frac{c_1}{c_2}\right|\right)\right)}{\left|c_1c_2\right|\left((k-1)^2\cosh\left(2(k+1)x+2\log\left|\frac{c_1}{c_2}\right|\right)+(k+1)^2\cosh\left(2(k-1)x\right)-4k\cos\left(2(k^2-1)t\right)\right)}.
\end{equation}
From the above exact solution, we can see that it is localized in the $x$-direction and periodic in the $t$-direction with the frequency $\left|k^2-1\right|/\pi$, which has the similar dynamics as the Kuznetsov-Ma breather in the plane wave background \cite{kuznetsov1977solitons,ma1979perturbed}. As shown in the expression \eqref{eq:single-breather}, the parameter $c_1/c_2$ will  not only affect the location of breather but also the dynamics. If the value of $c_1/c_2$ is large enough, the solutions will exhibit two parallel propagating solitons with distinct amplitudes and very weak breathing effect. Thus in the following, we just analyze the special case $c_1=c_2=1$, in which the interaction between two solitons are strongest to each other periodically. For the high-order breathers, their expressions are the rational functions of exponential and polynomial functions and very complicated. Thus we just show their dynamics by the computer graphics in the following subsection.
\subsection{Qualitative properties of high-order breathers}
In this subsection, we continue to display some features of high-order breathers. From the exact solutions in Eq.\eqref{eq:qn-exact}, by choosing different order $n$, we plot some high-order breathers in Fig. \ref{fig:exact-3-d}.
\begin{figure}[ht]
\centering
\includegraphics[width=0.8\textwidth,height=0.6\textwidth]{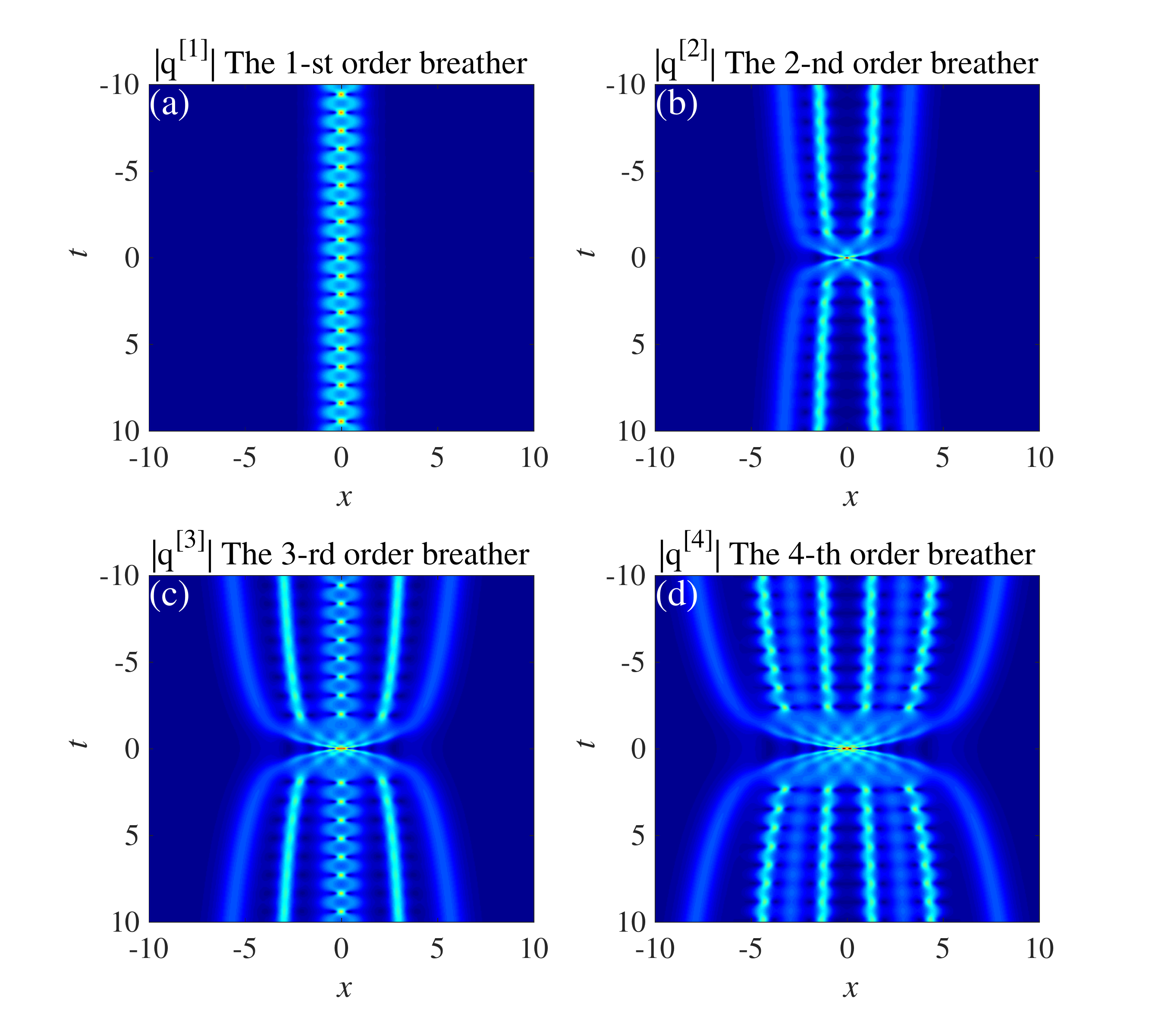}
\caption{The modulus of $\left|q^{[n]}(x,t)\right|$ for $n=1$, $n=2$, $n=3$ and $n=4$ by choosing $k=2$ and $c_1=c_2=1.$}
\label{fig:exact-3-d}
\end{figure}

When $n=1$, there is only one single breather, and as the order increases, there appear some supernatural regions describing the interaction of high-order breathers. Especially, due to the effect of periodicity, there is a high frequency oscillatory region, which is not existent in the single high-order solitons. To see their dynamics clearly, we also give some high-order cases in Fig. \ref{fig:exact-3-d-1}, which show that the additional region is observable evidently.
\begin{figure}[ht]
\centering
\includegraphics[width=0.8\textwidth,height=0.3\textwidth]{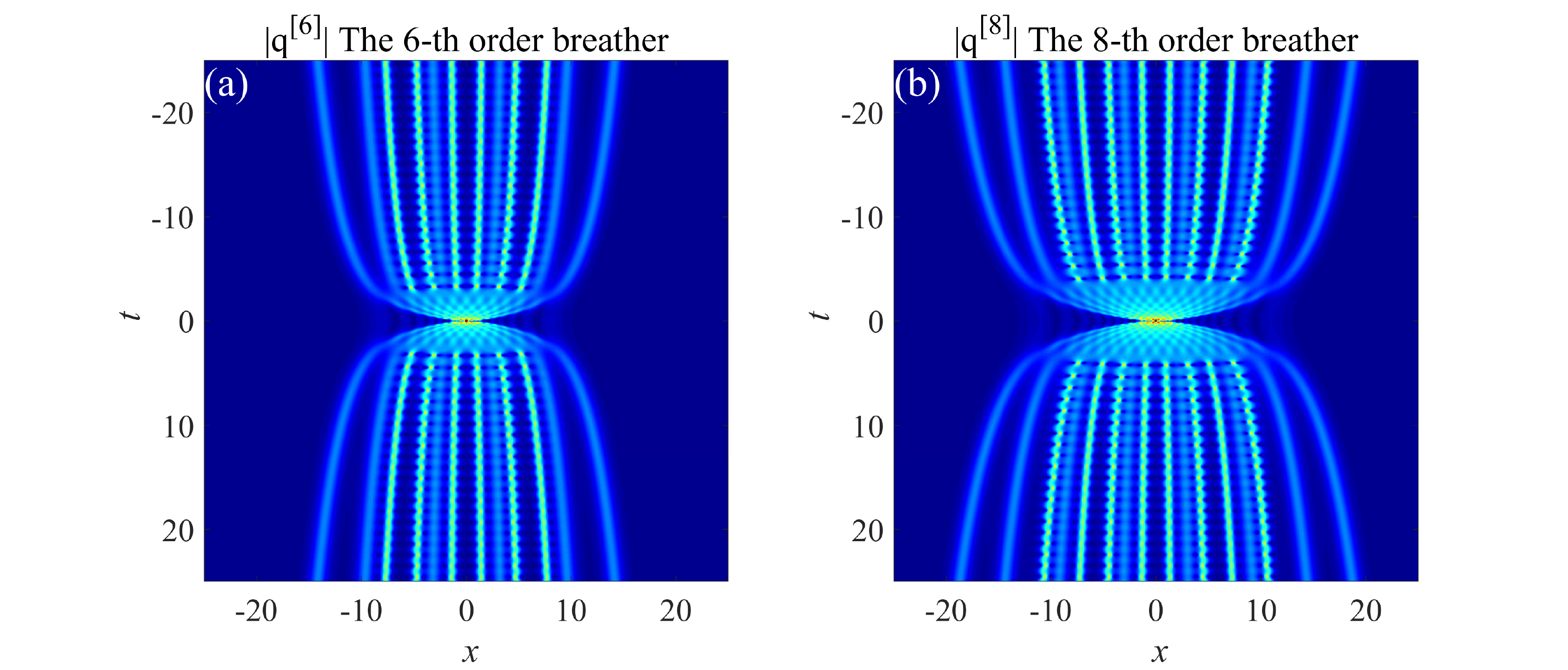}
\caption{The modulus of $\left|q^{[n]}(x,t)\right|$ for $n=6$ and $n=8$ with the parameter $k=2$ respectively.}
\label{fig:exact-3-d-1}
\end{figure}

From these figures, we can find that the asymptotics about this additional interaction region is out of the previous analyzing framework, thus we need to improve the corresponding methods to handle them. If we try to use the exact solution in Eq.\eqref{eq:qn-exact} for this analysis, we can see that the Darboux transformation $\mathbf{T}^{[2n]}(\lambda; x, t)$ is very complicated, which is hard to be used to analyze the large order asymptotics. Following the idea in \cite{Bilman-JDE-2021}, we would like to use the Riemann-Hilbert method and the Deift-Zhou nonlinear steepest descent method to analyze the asymptotics.
\subsection{The Riemann-Hilbert representation of high-order breather}
In this subsection, we will construct the RHP about the high-order breathers via the Darboux transformation. Before establishing the RHP, we first give a lemma about the properties of Darboux transformation.
\begin{lemma}
Set $c_1=c_2=1$, then $\mathbf{T}^{[2n]}(\lambda; x, t)\begin{bmatrix}\Delta(\lambda)^2\\\Delta(\lambda)^2\ee^{2\ii\lambda(x+\lambda t)}
\end{bmatrix}$
is analytic at the singularities $\lambda=\lambda_1, \lambda=\lambda_2$ and $\mathbf{T}^{[2n]}(\lambda; x, t)\begin{bmatrix}-\ee^{-2\ii\lambda(x+\lambda t)}\\1
\end{bmatrix}$ is analytic at $\lambda=\lambda_1^*, \lambda=\lambda_2^*$, where $\Delta(\lambda)=\left(\frac{\lambda-\lambda_1}{\lambda-\lambda_1^*}\frac{\lambda-\lambda_2}{\lambda-\lambda_2^*}\right)^{-n/2}$.
\end{lemma}
\begin{proof}
From the kernel conditions of the Darboux transformation, in the neighborhood of $\lambda=\lambda_1$ and $\lambda=\lambda_2$, we know that the meromorphic vector $\mathbf{T}^{[2n]}(\lambda; x, t)\begin{bmatrix}1\\\ee^{2\ii\lambda(x+\lambda t)}
\end{bmatrix}$ can be expanded into the form of $\mathcal{O}\left((\lambda-\lambda_1)^{n}\right)$ and $\mathcal{O}\left((\lambda-\lambda_2)^{n}\right)$ respectively. Thus we get the conclusion that if this vector multiplies the factor $\Delta(\lambda)^2$, it is analytic in the neighborhood of $\lambda=\lambda_1$ and $\lambda=\lambda_2$. Similarly, in the neighborhood of $\lambda=\lambda_1^*$ and $\lambda=\lambda_2^*$, with the residue conditions of Darboux transformation, we know that $\mathbf{T}^{[2n]}(\lambda; x, t)\begin{bmatrix}-\ee^{-2\ii\lambda(x+\lambda t)}\\1
\end{bmatrix}$ is analytic at $\lambda=\lambda_1^*, \lambda=\lambda_2^*$, it completes the proof.
\end{proof}
From the properties of Darboux transformation, we can construct three sectional analytic matrices as follows,
\begin{equation}\label{eq:M}
\mathbf{M}(\lambda; x, t):=\left\{\begin{aligned}&\mathbf{T}^{[2n]}(\lambda; x, t)\begin{bmatrix}\Delta(\lambda)^2&0\\0&1\end{bmatrix},\quad\lambda\in\mathbb{C}\setminus\left(D_{0}^{+}\cup D_{0}^{-}\right),\\
&\mathbf{T}^{[2n]}(\lambda; x, t)\begin{bmatrix}\Delta(\lambda)^2&0\\
\Delta(\lambda)^2\ee^{2\ii\lambda(x+\lambda t)}&1
\end{bmatrix},\quad \lambda\in D_{0}^{+},\\
&\mathbf{T}^{[2n]}(\lambda; x, t)\begin{bmatrix}\Delta(\lambda)^2&-\ee^{-2\ii\lambda(x+\lambda t)}\\
0&1
\end{bmatrix},\quad \lambda\in D_{0}^{-},
\end{aligned}\right.
\end{equation}
where $D_{0}^{+}$ and $D_{0}^{-}$ are two upper/lower semicircles involving the singularities $\lambda=\lambda_1, \lambda=\lambda_2$ and $\lambda=\lambda_1^*, \lambda=\lambda_2^*$ respectively, which can be seen from Fig.\ref{fig:jump-contour}.
\begin{figure}[ht]
\centering
\includegraphics[width=0.4\textwidth]{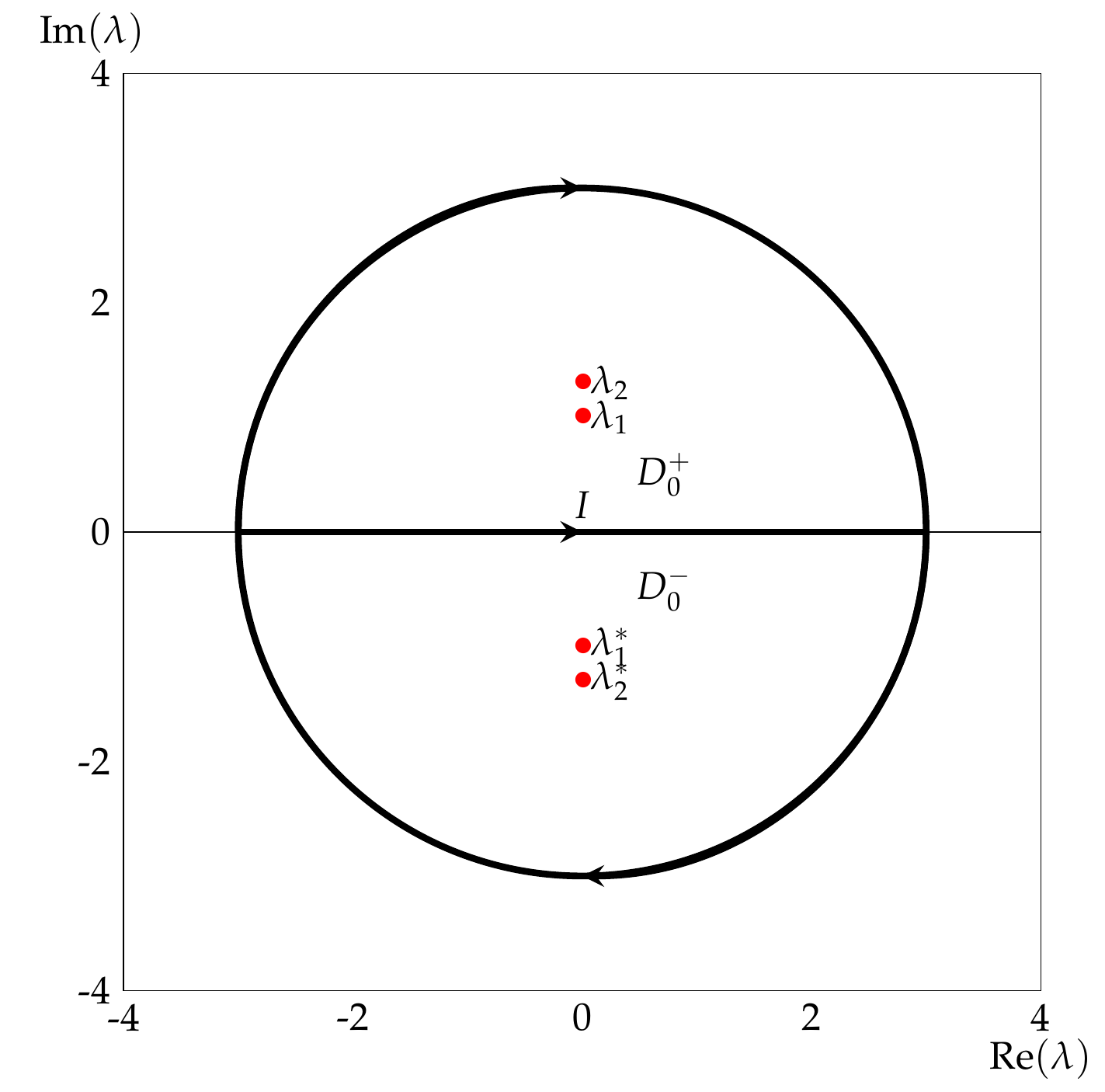}
\caption{The jump contour about $\mathbf{M}(\lambda; x, t)$}
\label{fig:jump-contour}
\end{figure}
Then the jump condition about $\mathbf{M}(\lambda; x, t)$ changes into
\begin{equation}\label{eq:jump-M}
\mathbf{M}_{+}(\lambda; x, t)=\mathbf{M}_{-}(\lambda; x, t)\left\{\begin{aligned}&\begin{bmatrix}1&0\\
-\ee^{2\ii\theta(\lambda; x, t)}&1
\end{bmatrix},\quad\lambda\in\partial D_{0}^{+},\\
&\begin{bmatrix}1&\ee^{-2\ii\theta(\lambda; x, t)}\\
0&1
\end{bmatrix},\quad \lambda\in\partial D_{0}^{-},\\
&\begin{bmatrix}2&\ee^{-2\ii\theta(\lambda; x, t)}\\
\ee^{2\ii\theta(\lambda; x, t)}&1
\end{bmatrix},\quad \lambda\in I,
\end{aligned}\right.
\end{equation}
where $\theta(\lambda; x, t)=\lambda(x+\lambda t)-\ii\log\Delta(\lambda)$. Next, we begin to study the asymptotics under the framework of this RHP. We can see that the phase term $\theta(\lambda; x, t)$ appearing in $\mathbf{M}(\lambda; x, t)$ is well suited for analyzing the large order asymptotics when the variables $x$ and $t$ are proportional to the order $n$.
Indeed, we give a rescale transformation as \cite{Bilman-JDE-2021}:
\begin{equation}
x=n \chi,\quad t=n\tau,
\end{equation}
then the jump conditions on $\partial D_{0}^{\pm}$ change into
\begin{equation}\label{eq:jump-M-1}
\mathbf{M}_{+}(\lambda; n\chi, n\tau)=\mathbf{M}_{-}(\lambda; n\chi, n\tau)\left\{\begin{aligned}&\begin{bmatrix}1&0\\
-\ee^{2\ii n\vartheta(\lambda; \chi, \tau)}&1
\end{bmatrix},\quad\lambda\in\partial D_{0}^{+},\\
&\begin{bmatrix}1&\ee^{-2\ii n\vartheta(\lambda; \chi, \tau)}\\
0&1
\end{bmatrix},\quad \lambda\in\partial D_{0}^{-},
\end{aligned}\right.
\end{equation}
where
\begin{equation}\label{eq:phase}
\vartheta(\lambda; \chi, \tau):=\lambda\chi+\lambda^2\tau+\frac{\ii}{2}\log\left(\frac{\lambda-\lambda_1}{\lambda-\lambda_1^*}\right)+\frac{\ii}{2}\log\left(\frac{\lambda-\lambda_2}{\lambda-\lambda_2^*}\right).
\end{equation}
The $n$-th order breather $q^{[n]}(n\chi, n\tau)$ can be recovered by the asymptotic expansion,
\begin{equation}\label{eq:qn-M}
q^{[n]}(n\chi, n\tau)=2\ii\lim\limits_{\lambda\to\infty}\lambda \mathbf{M}(\lambda; n\chi, n\tau).
\end{equation}
It is easy to see that when $n$ is large, the exponential term $\ee^{2\ii n\vartheta(\lambda; \chi, \tau)}$ in the jump matrix will decay into zero in the region of $\Im(\vartheta(\lambda; \chi, \tau))>0$, and the factor $\ee^{-2\ii n\vartheta(\lambda; \chi, \tau)}$ will decay to zero in the region of $\Im(\vartheta(\lambda; \chi, \tau))<0.$ By choosing one fixed $\chi$ and $\tau$, we give the contour plot of $\Im(\vartheta(\lambda; \chi, \tau))$ in the following figure (Fig. \ref{fig:exp-decay}).
\begin{figure}[ht]
\centering
\includegraphics[width=0.45\textwidth]{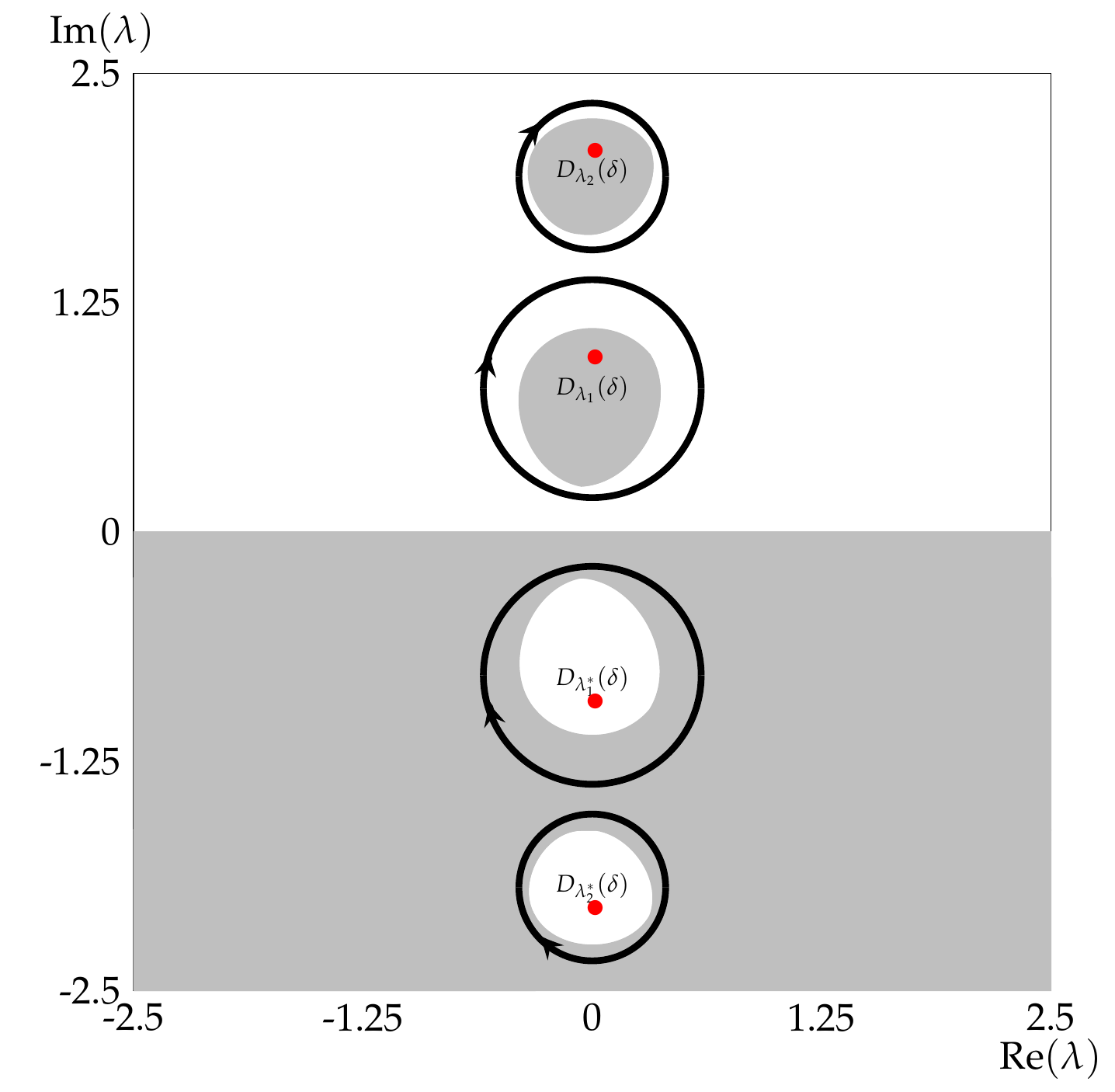}
\caption{The contour plot of $\Im\left(\vartheta\left(\lambda; 38/25, 0\right)\right)$ in the exponential-decay region by choosing $k=2$, where $\Im\left(\vartheta\left(\lambda; 38/25, 0\right)\right)<0$(shaded) and $\Im\left(\vartheta\left(\lambda; 38/25, 0\right)\right)>0$(unshaded), the red dots are the singularities at $\lambda=\pm\ii$ and $\lambda=\pm2\ii$.}
\label{fig:exp-decay}
\end{figure}

From the jump condition about $\mathbf{M}(\lambda; n\chi, n\tau)$ in Eq.\eqref{eq:jump-M-1}, we see that if the jump contour on $\partial D_{0}^{+}$ deforms into the boundary of $\partial D_{\lambda_2}(\delta)$ and $\partial D_{\lambda_1}(\delta)$, and the jump contour on $\partial D_{0}^{-}$ deforms into the boundary of $\partial D_{\lambda_2^*}(\delta)$ and $\partial D_{\lambda_1^*}(\delta)$, which are shown in Fig. \ref{fig:exp-decay}.
Under this condition, the jump matrices in Eq.\eqref{eq:jump-M} satisfy
\begin{equation}
\|\mathbf{V}_{\mathbf{M}}(\lambda; n\chi, n\tau)-\mathbb{I}\|=\mathcal{O}(\ee^{-\mu n}), \quad \mu>0,
\end{equation}
where $\mathbf{V}_{\mathbf{M}}(\lambda; n\chi, n\tau)$ indicates the jump matrix of $\mathbf{M}(\lambda; n\chi, n\tau)$. According to Eq.\eqref{eq:qn-M}, we can get the asymptotics in the exponential-decay region, that is
\begin{equation}\label{eq:exp-decay}
q^{[n]}(n\chi, n\tau)=\mathcal{O}\left(\ee^{-\mu n}\right).
\end{equation}
To get the asymptotics in other regions, we should deform $\mathbf{M}(\lambda; n\chi, n\tau)$ in a proper way. To realize it, we first need to eliminate the jump contour on $\lambda=I$. With the Deift-Zhou nonlinear steepest descent method, we take a decomposition to the jump matrix on $\lambda=I$,
\begin{equation}
\begin{bmatrix}2&\ee^{-2\ii n\vartheta(\lambda; \chi, \tau)}\\
\ee^{2\ii n\vartheta(\lambda; \chi, \tau)}&1
\end{bmatrix}=\ee^{-\ii n\vartheta(\lambda; \chi, \tau)\sigma_3}\begin{bmatrix}\sqrt{2}&0\\
\frac{\sqrt{2}}{2}&\frac{\sqrt{2}}{2}
\end{bmatrix}\begin{bmatrix}\sqrt{2}&\frac{\sqrt{2}}{2}\\
0&\frac{\sqrt{2}}{2}
\end{bmatrix}\ee^{\ii n\vartheta(\lambda; \chi, \tau)\sigma_3}.
\end{equation}
Based on this decomposition, we can introduce another matrix $\mathbf{R}(\lambda; \chi, \tau)$ defined as
\begin{equation}
\mathbf{R}(\lambda; \chi, \tau):=\left\{\begin{aligned}&\mathbf{M}(\lambda; n\chi, n\tau)\ee^{-\ii n\vartheta(\lambda; \chi, \tau)\sigma_3}\begin{bmatrix}\frac{\sqrt{2}}{2}&-\frac{\sqrt{2}}{2}\\
0&\sqrt{2}
\end{bmatrix}\ee^{\ii n\vartheta(\lambda; \chi, \tau)\sigma_3},\quad\lambda\in D_{0}^{+},\\
&\mathbf{M}(\lambda; n\chi, n\tau)\ee^{-\ii n\vartheta(\lambda; \chi, \tau)\sigma_3}\begin{bmatrix}\sqrt{2}&0\\
\frac{\sqrt{2}}{2}&\frac{\sqrt{2}}{2}
\end{bmatrix}\ee^{\ii n\vartheta(\lambda; \chi, \tau)\sigma_3},\quad\lambda\in D_{0}^{-},\\
&\mathbf{M}(\lambda; n\chi, n\tau),\quad\lambda\in\mathbb{C}\setminus\left(D_{0}^{+}\cup D_{0}^{-}\right).
\end{aligned}\right.
\end{equation}
Then $\mathbf{R}(\lambda; \chi, \tau)$ satisfies the following RHP.

\begin{rhp}\label{rhp-M1}
($n$-th order breather) Let $(\chi,\tau)\in\mathbb{R}^2$ be arbitrary parameters. Then we can find a $2\times 2$ matrix $\mathbf{R}(\lambda; \chi, \tau)$ with the following properties.
\begin{itemize}
\item {\bf Analyticity:} $\mathbf{R}(\lambda; \chi, \tau)$ is analytic for $\lambda\in\mathbb{C}\setminus \partial D_0$ and it takes continuous boundary value from the interior and the exterior of $\partial D_0$.
\item {\bf Jump condition:} When $\lambda\in\partial D_0$, $\mathbf{R}_{\pm}(\lambda; \chi, \tau)$ are related by the following jump condition,
\begin{equation}\label{eq:jump-R}
\mathbf{R}_{+}(\lambda; \chi, \tau)=\mathbf{R}_{-}(\lambda; \chi, \tau)\ee^{-\ii n\vartheta(\lambda; \chi, \tau)\sigma_3}\mathbf{Q}_{c}\ee^{\ii n\vartheta(\lambda; \chi, \tau)\sigma_3}.
\end{equation}
\item {\bf Normalization:} $\mathbf{R}(\lambda; \chi, \tau)\to\mathbb{I}$ as $\lambda\to\infty$,
\end{itemize}
where $D_{0}=D_{0}^+\cup D_{0}^-$, $\mathbf{Q}_{c}=\frac{1}{\sqrt{2}}
\begin{bmatrix}1&1\\
-1&1
\end{bmatrix}$.
\end{rhp}
In the references \cite{Bilman-JDE-2021,Bilman-2019-CPAM,Bilman-arxiv-2021}, the authors give the corresponding RHP through the normalized method. While we derive the RHP by using the properties of Darboux transformation directly, we can further check that both two ways of Riemann-Hilbert problems are consistent.
From this RHP \ref{rhp-M1}, the solution $q^{[n]}(n\chi, n\tau)$ can be given with the following limit process,
\begin{equation}\label{eq:qn-al}
q^{[n]}(n\chi, n\tau)=2\ii\lim\limits_{\lambda\to\infty} \lambda R_{12}(\lambda; \chi, \tau).
\end{equation}
Following the theory of Deift-Zhou nonlinear steepest descent method, the controlling phase term $\vartheta(\lambda; \chi, \tau)$ and the decomposition of $\mathbf{Q}_{c}$ are important for the asymptotic analysis. So we show the following notations about the decomposition of this constant matrix $\mathbf{Q}_{c}$,
\begin{equation}\label{remark:decom}
\begin{aligned}
\mathbf{Q}_{c}&=\begin{bmatrix}\frac{\sqrt{2}}{2}&0\\0&\sqrt{2}
\end{bmatrix}\begin{bmatrix}1&0\\-\frac{1}{2}&1
\end{bmatrix}\begin{bmatrix}1&1\\0&1
\end{bmatrix}:=\mathbf{Q}_{L}^{[1]}\mathbf{Q}_C^{[1]}\mathbf{Q}_{R}^{[1]},&\quad \left(``{\rm DLU}"\right),\\
\mathbf{Q}_{c}&=\begin{bmatrix}\sqrt{2}&0\\0&\frac{\sqrt{2}}{2}
\end{bmatrix}\begin{bmatrix}1&\frac{1}{2}\\0&1
\end{bmatrix}\begin{bmatrix}1&0\\-1&1
\end{bmatrix}:=\mathbf{Q}_{L}^{[2]}\mathbf{Q}_C^{[2]}\mathbf{Q}_{R}^{[2]},&\quad \left(``{\rm DUL}"\right),\\
\mathbf{Q}_{c}&=\mathbf{Q}_{L}^{[2]}\begin{bmatrix}1&-\frac{1}{2}\\0&1
\end{bmatrix}\begin{bmatrix}0&1\\-1&0
\end{bmatrix}\begin{bmatrix}1&-1\\0&1
\end{bmatrix}:=\mathbf{Q}_{L}^{[2]}\mathbf{Q}_{L}^{[3]}\mathbf{Q}_C^{[3]}\mathbf{Q}_{R}^{[3]},&\quad \left(``\rm{DUTU}"\right),\\
\mathbf{Q}_{c}&=\mathbf{Q}_{L}^{[1]}\begin{bmatrix}1&0\\\frac{1}{2}&1
\end{bmatrix}\begin{bmatrix}0&1\\-1&0
\end{bmatrix}\begin{bmatrix}1&0\\1&1
\end{bmatrix}:=\mathbf{Q}_{L}^{[1]}\mathbf{Q}_{L}^{[4]}\mathbf{Q}_C^{[4]}\mathbf{Q}_{R}^{[4]}&\quad \left(``\rm{DLTL}"\right),\\
\mathbf{Q}_{c}&=\begin{bmatrix}1&1\\0&1
\end{bmatrix}\begin{bmatrix}\sqrt{2}&0\\0&\frac{\sqrt{2}}{2}
\end{bmatrix}\begin{bmatrix}1&0\\-1&1
\end{bmatrix}:=\mathbf{Q}_{L}^{[5]}\mathbf{Q}_C^{[5]}\mathbf{Q}_{R}^{[5]}&\quad \left(``\rm{UDL}"\right),\\
\mathbf{Q}_{c}&=\begin{bmatrix}1&0\\-1&1
\end{bmatrix}\begin{bmatrix}\frac{\sqrt{2}}{2}&0\\0&\sqrt{2}
\end{bmatrix}\begin{bmatrix}1&1\\0&1
\end{bmatrix}:=\mathbf{Q}_{L}^{[6]}\mathbf{Q}_C^{[6]}\mathbf{Q}_{R}^{[6]}&\quad \left(``\rm{LDU}"\right).
\end{aligned}
\end{equation}
As to this RHP \ref{rhp-M1}, when $\lambda_1=\ii, \lambda_2=k\ii\,(k>1)$, we plot the $15$-th order breather in Fig.\ref{fig:3-d} and give the boundary curves (yellow solid lines) to these five different regions marked with $A, E$, $g=0, g=1$ and $g=3$. Now, we give a detailed description about these boundary lines. From the definition of $\vartheta(\lambda; \chi, \tau)$ in Eq.\eqref{eq:phase}, the critical points of $\vartheta(\lambda; \chi, \tau)$ are given by the roots of following algebraic equation:
\begin{equation}\label{eq:cri-poi}
\left(\chi+2\tau\lambda\right)\left(\lambda^2+1\right)\left(\lambda^2+k^2\right)-(k+1)\left(\lambda^2+k\right)=0.
\end{equation}
The discriminant of Eq.\eqref{eq:cri-poi} about $\lambda$ is a function with respect to $\chi, \tau$ and $k$, and the exact expression is tedious, so we give a special example for $k=2$,
\begin{multline}\label{eq:dis}
288\Big[73728\tau^8+\left(222560-229632\chi+46080\chi^2\right)\tau^6+\left(9504\chi^4-50832\chi^3+99168\chi^2-54360\chi-14094\right)\tau^4\\
+\left(720\chi^6-3744\chi^5+6936\chi^4-8685\chi^3+8721\chi^2-5103\chi+2187\right)\tau^2\\+18\chi^8-51\chi^7+80\chi^6-90\chi^5+54\chi^4-27
\chi^3\Big].
\end{multline}
If this discriminant is greater than $0$, the quintic polynomial Eq.\eqref{eq:cri-poi} has at least three real critical points, otherwise, it has only one real critical point. We can show that Eq.\eqref{eq:dis}$>0$ will correspond to the algebraic-decay region. On the exterior of this algebraic-decay region, there are three kinds of curves separating the different regions, one is the boundary line between $A$ and $g=1$, one is the boundary line between $g=1$ and $g=3$, the rest one is the boundary between $g=0$ and $g=1$ and $g=3$. Similar to the large order soliton and large order rogue wave asymptotic analysis in \cite{Bilman-arxiv-2021}, this curve between $A$ and $g=1$ can be determined by the condition
\begin{equation}
\ell_{sol}:\Re\left(\int_{\Gamma}\ii\vartheta'\lambda; \chi, \tau)d\lambda\right)=0,
\end{equation}
where $\Gamma$ is any Schwarz-symmetric contour avoiding these four singularities $\lambda=\pm\ii, \lambda=\pm k\ii$. As $\vartheta(\lambda; \chi, \tau)$ has five different critical points in this region, one is a real root and the others are two pairs of conjugate complex roots, thus there exist two different curves, one is separating the $g=1$ region and the $E$ region and the other one is an approximate boundary between the $g=3$ region and the $g=1$ region, which should be revised by the discriminant of the modified controlling phase term $h_3(\lambda; \chi, \tau)$ function. This modified phase term is given in the asymptotic analysis of genus-three region. The rest of curves (Fig. \ref{fig:3-d}) are the boundaries between the $g=0$ region and the $g=1$ and $g=3$ region, which are determined by the algebraic curve of genus-zero. In the genus-zero region, the original phase term $\vartheta(\lambda; \chi, \tau)$ is no longer used to analyze the asymptotics, we need to introduce the $g$-function defined as
\begin{equation}
g'(\lambda):=\frac{R(\lambda)}{2}\left(\frac{\ii}{\left(\lambda-\ii\right)R(\ii)}-\frac{\ii}{(\lambda+\ii)R(-\ii)}+\frac{\ii}{(\lambda-k\ii)R(k\ii)}-\frac{\ii}{(\lambda+k\ii)R(-k\ii)}+4\tau\right)-\vartheta'(\lambda; \chi, \tau),
\end{equation}
where
$$g(\lambda)\equiv g(\lambda; \chi, \tau),\quad R(\lambda)\equiv R(\lambda; \chi, \tau)=\sqrt{(\lambda-a_1(\chi, \tau))(\lambda-a_1^*(\chi, \tau))}.$$
Then the controlling phase term $h(\lambda)\equiv h(\lambda; \chi, \tau)$ can be defined as
\begin{equation}
h'(\lambda):=g'(\lambda)+\vartheta'(\lambda; \chi, \tau)=\frac{R(\lambda)}{2}\left(\frac{\ii}{\left(\lambda{-}\ii\right)R(\ii)}{-}\frac{\ii}{(\lambda{+}\ii)R({-}\ii)}{+}\frac{\ii}{(\lambda{-}k\ii)R(k\ii)}{-}\frac{\ii}{(\lambda{+}k\ii)R({-}k\ii)}{+}4\tau\right).
\end{equation}
In general, in the genus-zero region, we require that $h'(\lambda)$ has at least two real roots such that there exists a closed curve involving at least one pair of conjugate singularities. But when these two real roots coincide into one doule root, $(\chi, \tau)$ will transfer into the higher genus regions. Therefore, we can get the  boundary between the genus-zero region and higher genus regions, which can be seen from Fig.\ref{fig:3-d}.

\subsection{Outline of this work}
The outline of the following work is given as follows: In section \ref{sec:result}, we present the main results about the asymptotics for these different five regions, and we also give the numerical verification between the exact solutions and the asymptotic expressions. Then the next four sections are the detailed calculations to these asymptotic regions. In section \ref{sec:al-decay}, we first give the asymptotics in the algebraic-decay region, and the local analysis can be given by the solution of parabolic cylinder function, the leading order term is $\mathcal{O}(n^{-1/2})$. Afterwards, we give the asymptotic analysis to the genus-zero region in section \ref{sec:genus-zero}. In this region, we use the parabolic cylinder function and the Airy function to construct the inner parametrix, and we also give a detailed calculation about the solution of Airy function in the appendix \ref{app:Airy}. In last two sections, we give the asymptotics about the genus-one region and the genus-three region respectively, both of them are expressed by the Riemann-Theta function, especially, to the genus-one region, the asymptotics can also be simplified as the Jacobi elliptic function.
\section{The main results and the numeric verification}
\label{sec:result}
In this section, we give our main results about the asymptotic analysis of high-order breathers to nonlinear Schr\"{o}dinger equation. By choosing distinct $\chi$ and $\tau$, we can get five different asymptotic regions, which can also be seen from the exact solution in Fig. \ref{fig:3-d} clearly.
\begin{figure}[ht]
\centering
\includegraphics[width=0.45\textwidth]{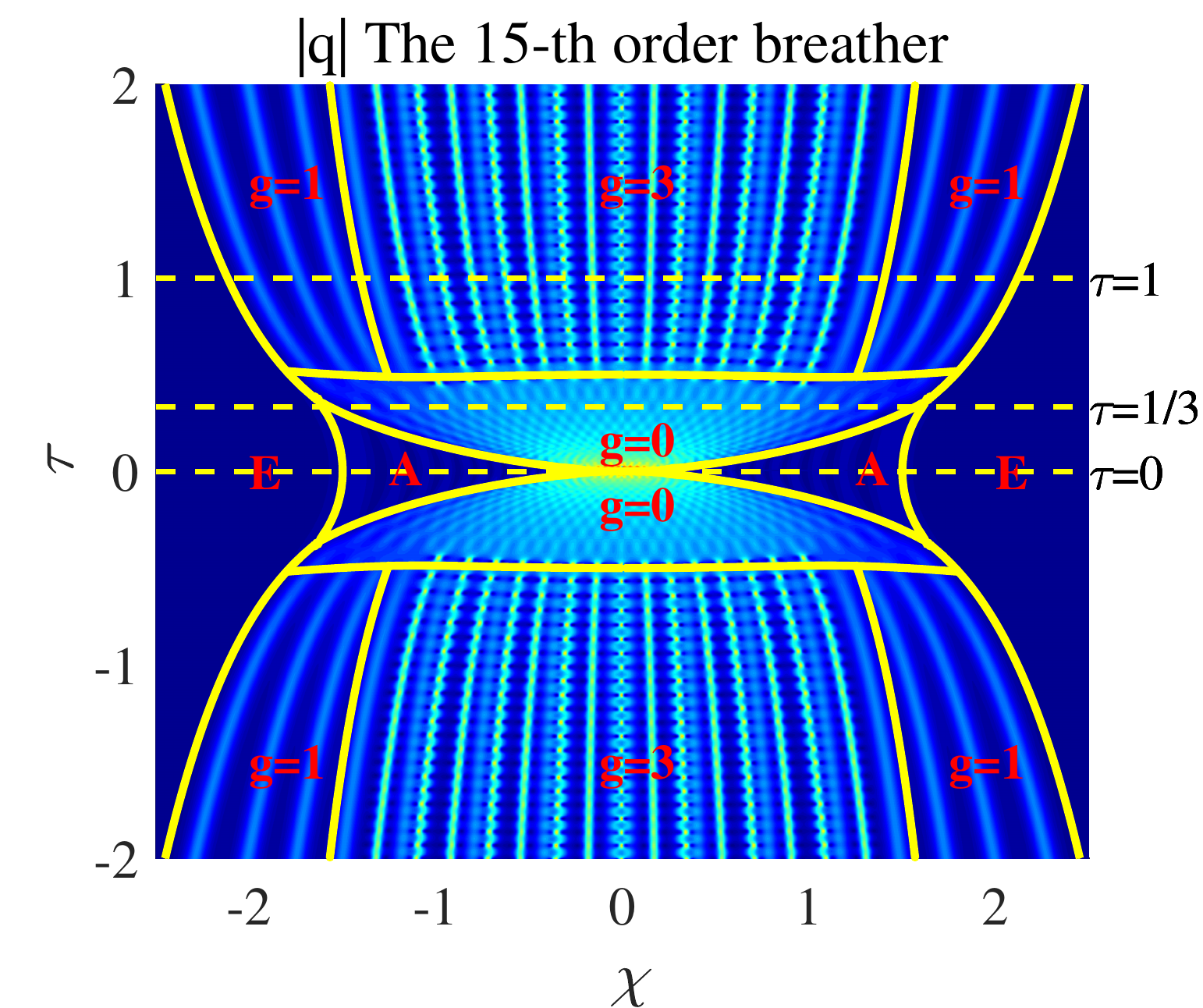}
\caption{The $15$-th order solution for nonlinear Schr\"{o}dinger equation by choosing the parameters $\lambda_1=\ii$ and $\lambda_2=2\ii$. The yellow lines are the boundary between the different regions. The alphabets $A$ and $E$ denote the algebraic-decay and the exponential-decay region, $g=0, g=1, g=3$ indicate the regions of genus-zero, genus-one and genus-three respectively. We give three different yellow dashed lines about $\tau=0$, $\tau=1/3$ and $\tau=1$, which correspond the algebraic-decay region, the genus-zero region and genus-one region as well as the genus-three region respectively. In later analysis, we choose these three different $\tau$ to compare the exact solutions and the asymptotic results in Fig. \ref{fig:al}, Fig. \ref{fig:no} and Fig. \ref{fig:os}.}
\label{fig:3-d}
\end{figure}

For these five different regions, we give the asymptotic expressions by the following five theorems.
\begin{theorem}\label{theo-exp}
(The exponential-decay region) When $\left(\chi, \tau\right)$ is in the exponential-decay region, the asymptotics of $q^{[n]}(n\chi, n\tau)$ is
\begin{equation}
q^{[n]}(n\chi, n\tau)=\ee^{-\mu n}, (\mu>0).
\end{equation}
\end{theorem}
This theorem is proven in the last section.
Afterwards, we give the asymptotic result about the algebraic-decay region.
\begin{theorem}\label{theo:al}
(The algebraic-decay region) When $\left(\chi, \tau\right)$ is located in the algebraic-decay region, the asymptotics of $q^{[n]}(n\chi, n\tau)$ is
\begin{multline}\label{eq:al}
q^{[n]}(n\chi, n\tau){=}\frac{1}{n^{1/2}}\sqrt{\frac{\ln(2)}{\pi}}\left[\frac{\ee^{{-}2\ii n\vartheta(\alpha_1, \chi, \tau)}\left({-}\vartheta''(\alpha_1; \chi, \tau)\right)^{{-}\ii p}}{\sqrt{{-}\vartheta''(\alpha_1; \chi, \tau)}}\ee^{\ii\phi(\chi; \tau)}{+}\frac{\ee^{{-}2\ii n\vartheta(\beta_1, \chi, \tau)}\vartheta''(\beta_1; \chi, \tau)^{\ii p}}{\sqrt{\vartheta''(\beta_1; \chi, \tau)}}\ee^{{-}\ii\phi(\chi; \tau)}\right]\\+\mathcal{O}(n^{-3/2}),
\end{multline}
where
$$\phi=-\frac{\ln(2)}{2\pi}\ln(n)-\frac{\ln(2)}{\pi}\ln\left(\beta_1-\alpha_1\right)-\frac{\ln(2)^2}{2\pi}-\frac{1}{4}\pi+\arg\left(\Gamma\left(\frac{\ii\ln(2)}{2\pi}\right)\right).$$
\end{theorem}
In this case, we choose $\tau=0$ and give the comparison between the exact solution and the asymptotic solution, which is shown in Fig. \ref{fig:al}.
\begin{figure}[ht]
\centering
\includegraphics[width=0.8\textwidth,height=0.6\textwidth]{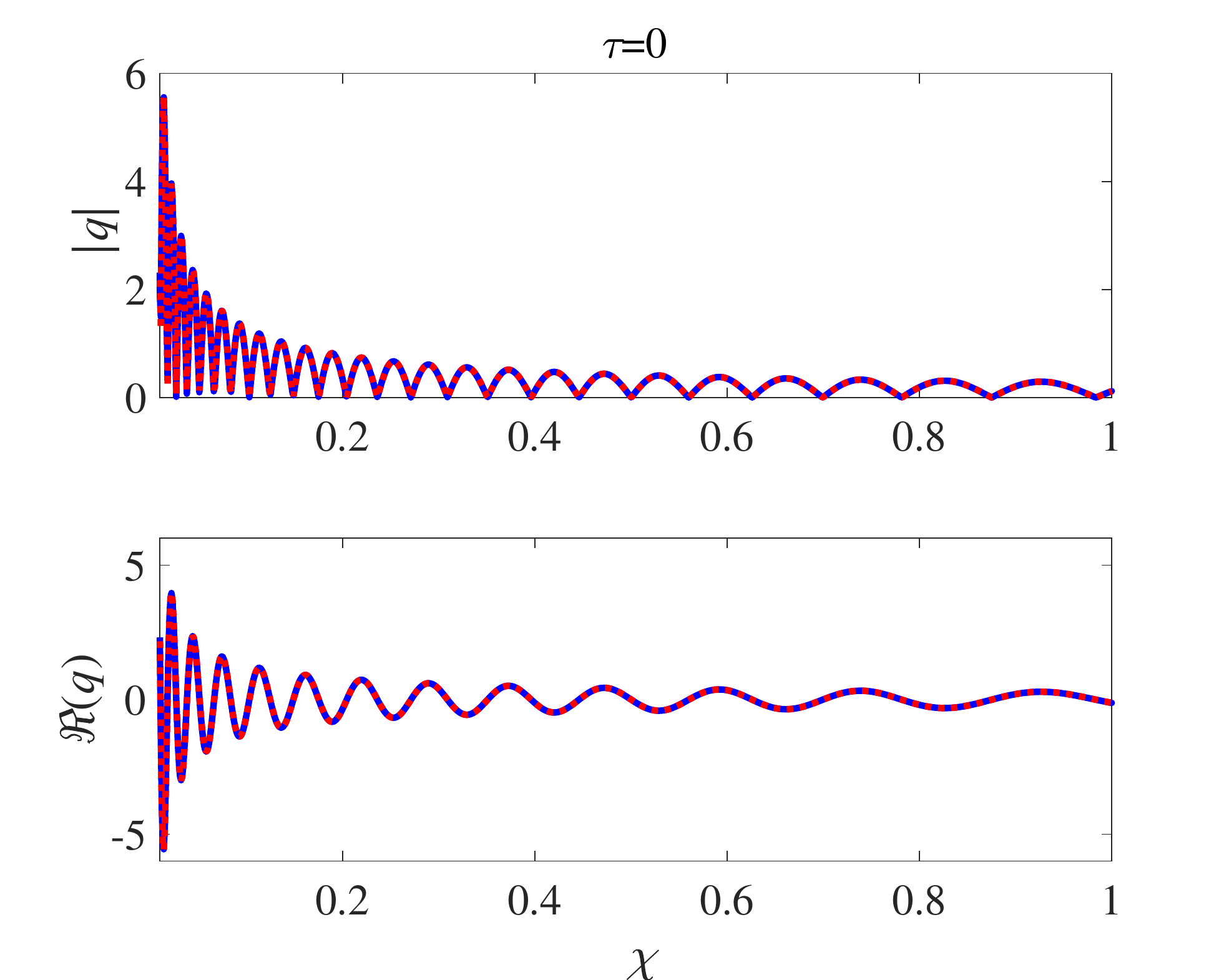}
\caption{The comparison between the exact solution ($15$-th order breather) and its asymptotic solution in the algebraic-decay region by choosing $\tau=0$ (as shown by the yellow dashed line in Fig. \ref{fig:3-d}). The blue solid line is the exact solution and the red dotted line is the asymptotic solution. The upper one is the modulus of this solution and the below one is the real part of it.}
\label{fig:al}
\end{figure}

When $\tau$ increases, the asymptotic region will transfer into the genus-zero region for small $\chi$. In this region, the original phase term $\vartheta(\lambda; \chi, \tau)$ is no longer viable, we should construct a $g$ function depending on an algebraic curve with the genus-zero to modify the phase term into a $h$ function. When $n$ is large, the asymptotics is shown in the following theorem \ref{theo:genus-zero},
\begin{theorem}\label{theo:genus-zero} (The genus-zero region) If $(\chi,\tau)$ is located in the genus-zero region, the large order asymptotics of $q^{[n]}(n\chi,n\tau)$ is given by
\begin{multline}\label{eq:q-gen-zero}
q^{[n]}(n\chi, n\tau)=\ee^{2\ii k_{2}(\infty)-\ii n\kappa_2}\Bigg[\frac{\sqrt{2p}}{n^{1/2}\sqrt{-h''_{2,-}(\alpha_2)}}\left(m_{-}^{\alpha_2}\ee^{\ii\phi_{\alpha_2}}-m_{+}^{\alpha_2}\ee^{-\ii\phi_{\alpha_2}}\right)\\
+\frac{\sqrt{2p}}{n^{1/2}\sqrt{h''_{2}(\beta_2)}}\left(m_{+}^{\beta_2}\ee^{\ii\phi_{\beta_2}}-m_{-}^{\beta_2}\ee^{-\ii\phi_{\beta_2}}\right)-\ii {\Im (a_2)}\Bigg]+\mathcal{O}(n^{-1}),
\end{multline}
where the parameters $m_{\pm}^{\alpha_2, \beta_2}$ are defined in Eq.\eqref{eq:para-genus-zero}, $ k_{2}(\infty)$ is defined in Eq.\eqref{eq:mu-genus-zero}, $h_2(\lambda; \chi, \tau)$ is defined in Eq.\eqref{eq:h2} and $\alpha_2, \beta_2$ are two real roots of $h_2(\lambda; \chi, \tau)$, the constant $\kappa_2$ is an integration constant defined in RHP\ref{rhp-genus-zero}. Similarly, we choose one fixed $\tau$ and give the comparison between the exact solution and the asymptotic solution in this region, which is shown in Fig.\ref{fig:no}.
\end{theorem}
\begin{figure}[ht]
\centering
\includegraphics[width=0.3\textwidth]{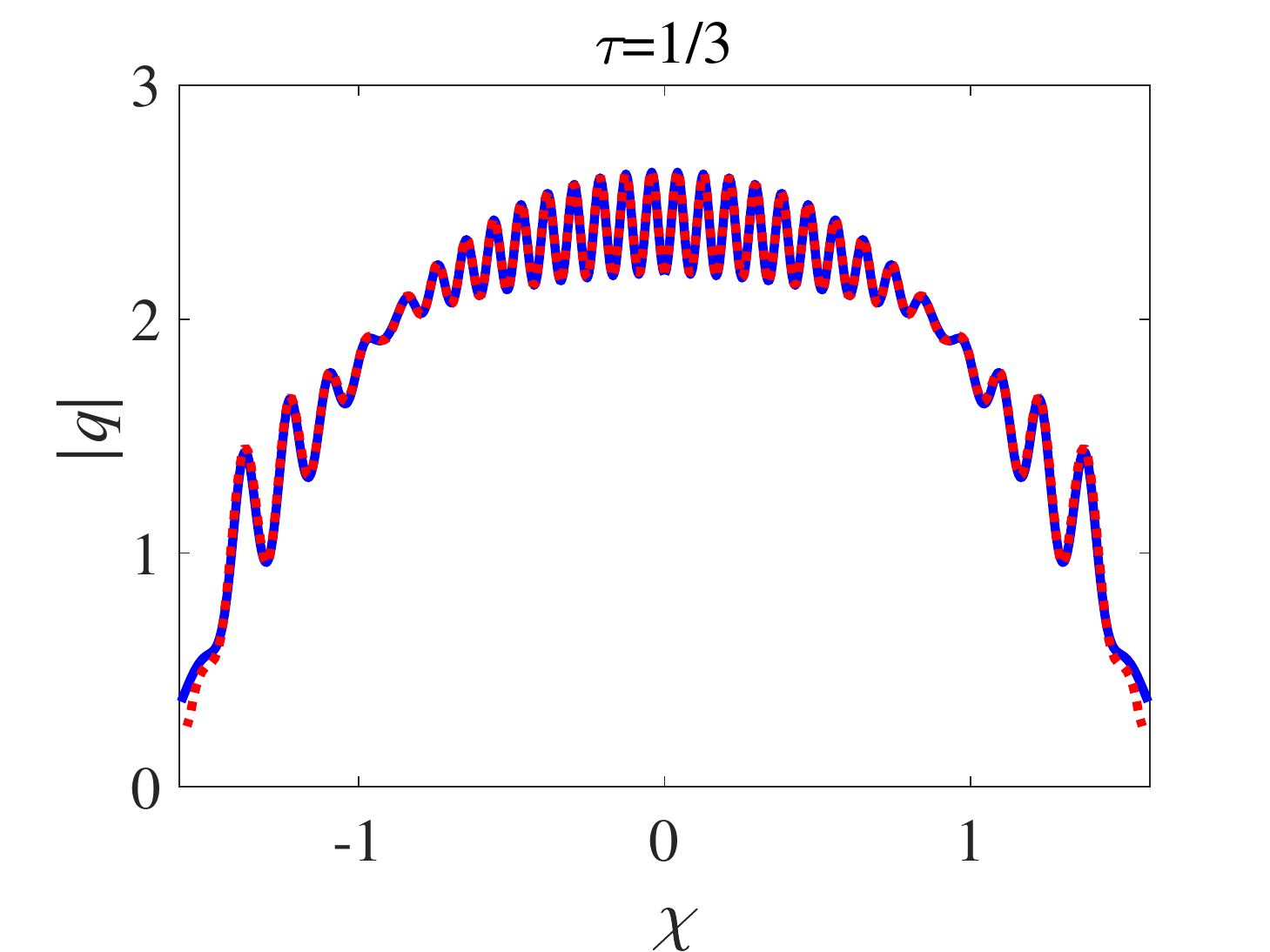}
\centering
\includegraphics[width=0.3\textwidth]{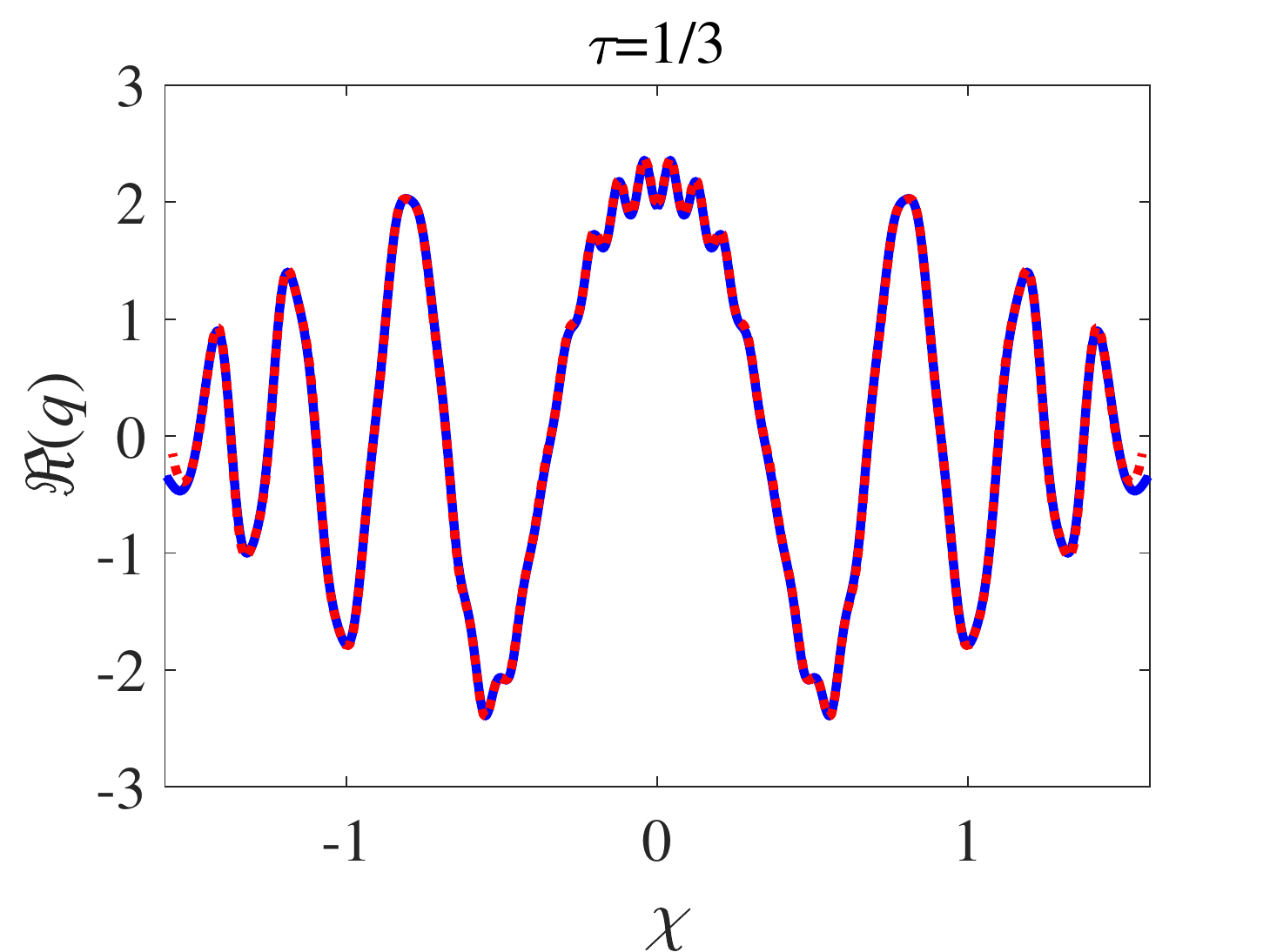}
\centering
\includegraphics[width=0.3\textwidth]{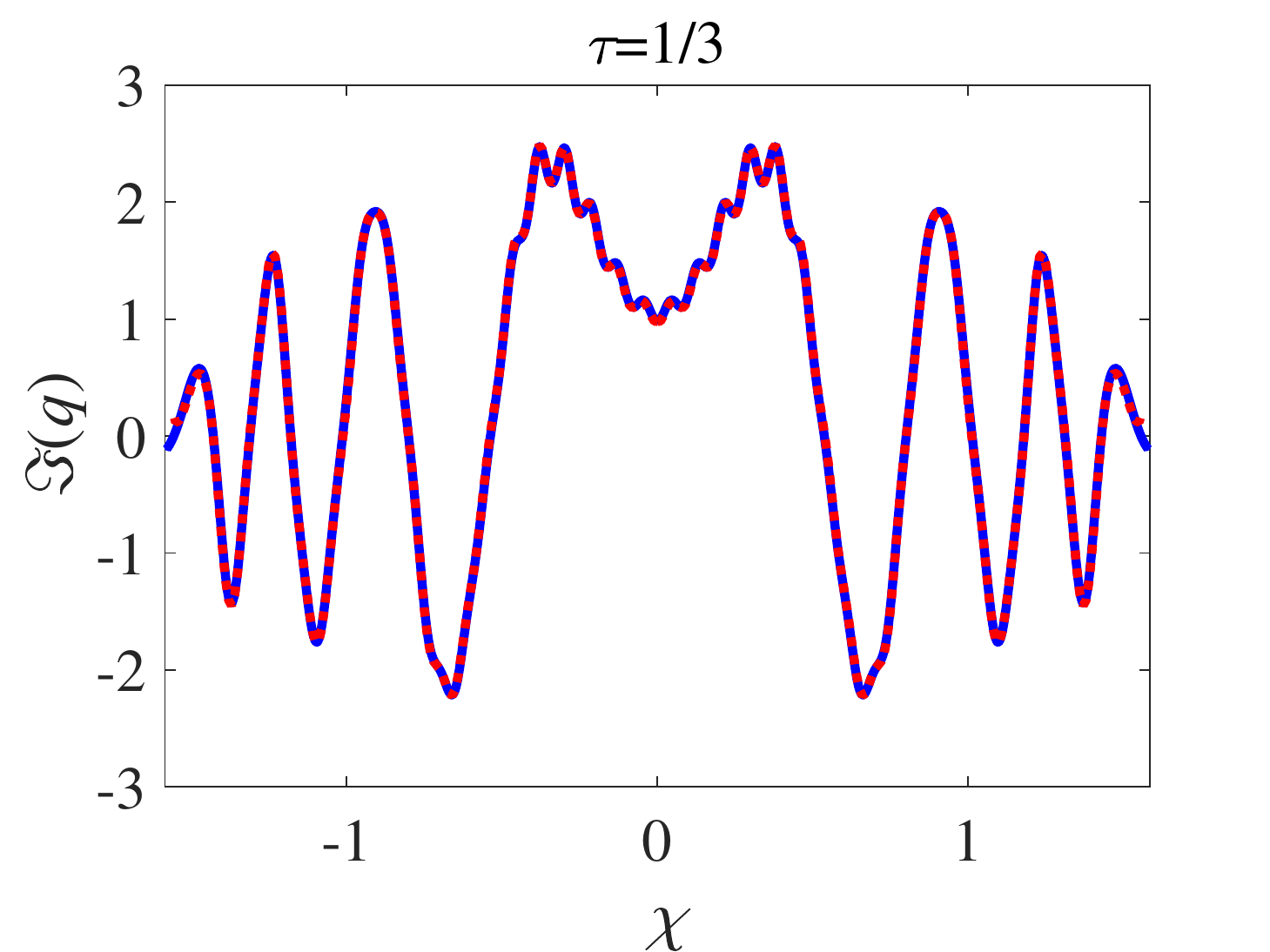}
\caption{The comparison between the exact solution ($15$-th order breather) and its asymptotic solution in the genus-zero region with $\tau=\frac{1}{3}$ (as shown by the yellow dashed line in Fig. \ref{fig:3-d}). The blue solid line is the exact solution and the red dotted line is the asymptotic solution. The left one is the modulus of this solution, the middle one is the real part and the right one is the imaginary part.}
\label{fig:no}
\end{figure}

Moreover, when $h_{2}(\lambda; \chi, \tau)$ defined in the genus-zero region has double root, the algebraic curve given in this region is no longer true, which should be modified a little. As a result, we can get genus-one region and genus-three region. Firstly, we give the asymptotic analysis to the genus-three region.
\begin{theorem}\label{theo:genus-three}
(The asymptotics in the genus-three region) When $\left(\chi, \tau\right)$ is in the genus-three region, we should introduce an algebraic curve with the genus-three, and the asymptotic expression can be given with the Riemann-Theta function,
\begin{multline}\label{eq:qn-genus-3-1}
q^{[n]}(n\chi, n\tau)=\frac{\Theta\left(\mathbf{A}(\infty)+\mathbf{d}\right)}{\Theta\left(\mathbf{A}(\infty)+\mathbf{d}-\pmb{\mathcal{U}}F_{41}-\pmb{\mathcal{V}}F_{42}-\pmb{\mathcal{W}}F_{43}\right)}
\frac{\Theta\left(\mathbf{A}(\infty)-\mathbf{d}+\pmb{\mathcal{U}}F_{41}+\pmb{\mathcal{V}}F_{42}+\pmb{\mathcal{W}}F_{43}\right)}{\Theta\left(\mathbf{A}(\infty)-\mathbf{d}\right)}\\
\times \ii\left(\Im(a_4)-\Im(b_4)+\Im(c_4)-\Im(d_4)\right)\ee^{2F_{41}J_{41}+F_{42}J_{42}+F_{43}J_{43}-2F_{40}}+\mathcal{O}(n^{-1}),
\end{multline}
where $\mathbf{A}(\infty)$ is the Abel mapping defined in Eq.\eqref{eq:abel-map}, and $\mathbf{U}, \mathbf{V}, \mathbf{W}, J_{41}, J_{42}, J_{43}, F_{40}, F_{41}, F_{42}, F_{43}$ are some functions with $\chi$ and $\tau$  given in Eq.\eqref{eq:UVW}, Eq.\eqref{eq:J-cons} and Eq.\eqref{eq:F4342}, $\mathbf{d}$ is related to the Abel mapping and the Riemann-Theta constant defined in Eq.\eqref{eq:d}, $a_{4}, b_{4}, c_{4}, d_{4}$ as well as their conjugates are eight branch points about this algebraic curve.
\end{theorem}
In this region, for fixed $\tau$, as $\chi$ increases, two branch points will approach to each other until they coincide, then the genus-three region will transfer into the genus-one region. Now we give the genus-one region asymptotics in the following theorem.
\begin{theorem}\label{theo:genus-one}
(The asymptotics in the genus-one region) For given $(\chi, \tau)$ in the genus-one region, the asymptotics about $q^{[n]}(n\chi, n\tau)$ can be expressed with the Riemann-Theta function of genus-one. Especially, its modulus can be written as the Jacobi elliptic function,
\begin{multline}\label{eq:qn-genus-1-4}
q^{[n]}(n\chi, n\tau)=\frac{\Theta\left(A(\infty)+A(Q)+\ii\pi+\frac{B}{2}\right)}{\Theta\left(A(\infty)+A(Q)+\ii\pi+\frac{B}{2}-F_{31}U\right)}\frac{\Theta\left(A(\infty)-A(Q)-\ii\pi-\frac{B}{2}+F_{31}U\right)}{\Theta\left(A(\infty)-A(Q)-\ii\pi-\frac{B}{2}\right)}\\
\times \ii\left(\Im(a_3)-\Im(b_3)\right)\ee^{2F_{31}J_{31}-2F_{30}}+\mathcal{O}(n^{-1}),
\end{multline}
where $A(\infty), B$ are the Abel integrals defined in Eq.\eqref{eq:Abel-int}, $F_{31}, F_{30}$ and $U$ are some functions with $\chi$ and $\tau$ given in Eq.\eqref{eq:F3130} and Eq.\eqref{eq:U-genus-one}, $Q=\frac{\Re(a_3)\Im(b_3)-\Re(b_3)\Im(a_3)}{\Im(b_3)-\Im(a_3)}$, $a_3, b_3$ are the branch points of the genus-one algebraic curve. Especially, the modulus of $q^{[n]}(n\chi, n\tau)$ can be rewritten as
\begin{equation}\label{eq:qn-genus-one-1}
\left|q^{[n]}(n\chi, n\tau)\right|^2=\left[\left(\Im(a_3)-\Im(b_3)\right)^2-\left|a_3-b_3\right|^2\cn^2\left(u+K(m),m\right)\right]\ee^{\ii n\varsigma d_3+4\left|F_{31}J_{31}\right|^2}+\mathcal{O}(n^{-1}),
\end{equation}
where $m=\frac{\theta_{2}^4(0)}{\theta_{3}^4(0)},$ $ \varsigma$ is given in Eq.\eqref{eq:varsigma} and $K(m)$ is the first kind of complete elliptic integral defined as
\begin{equation}
K(m):=\int_{0}^{\frac{\pi}{2}}\frac{d\xi}{\sqrt{1-m\sin^2(\xi)}}.
\end{equation}
\end{theorem}
By choosing $\tau=1$, we give the comparison between the exact determinant solution and the asymptotic analysis for the genus-three and genus-one region together, which is shown in Fig.\ref{fig:os}.
\begin{figure}[ht]
\centering
\includegraphics[width=0.3\textwidth]{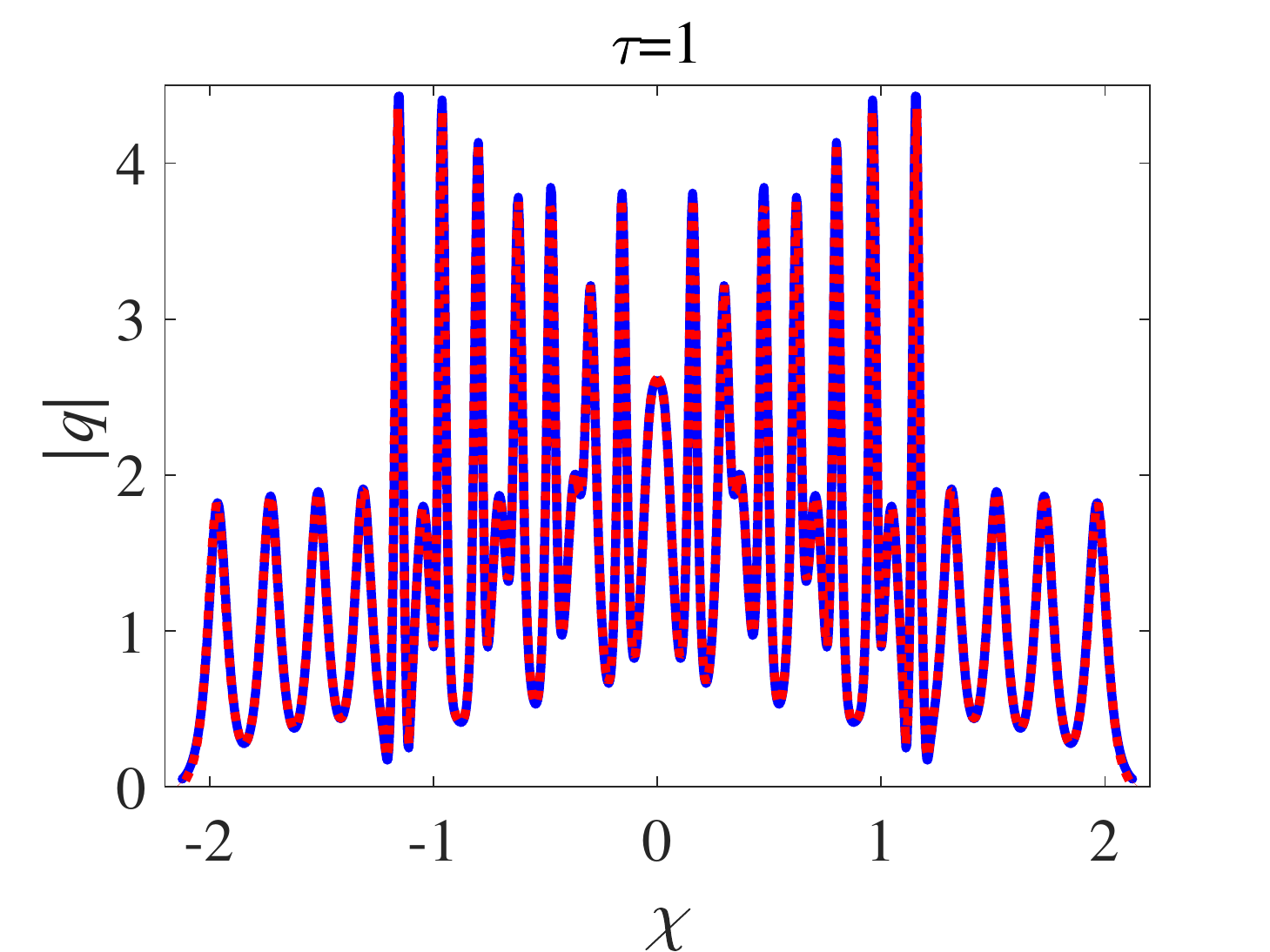}
\centering
\includegraphics[width=0.3\textwidth]{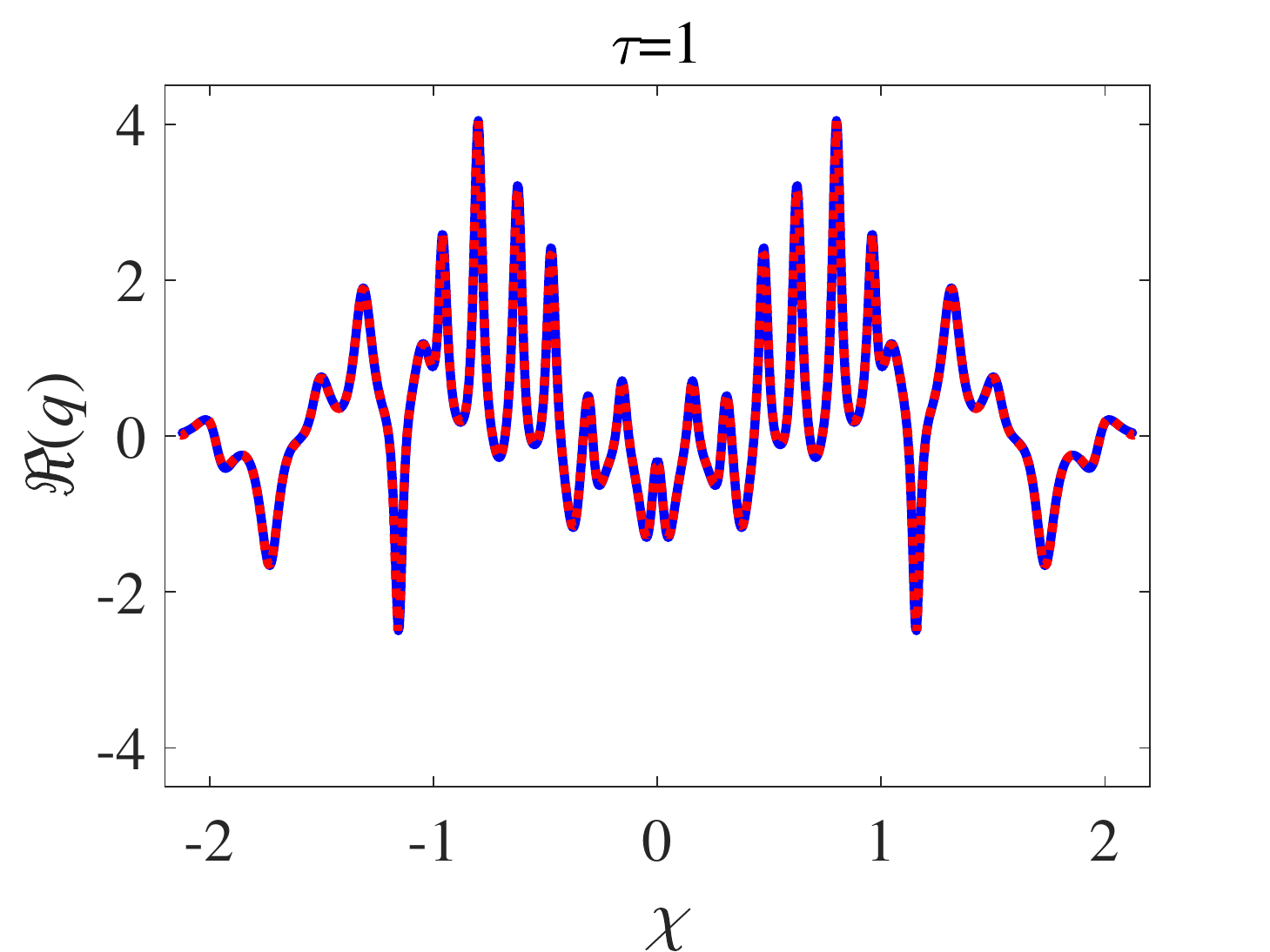}
\centering
\includegraphics[width=0.3\textwidth]{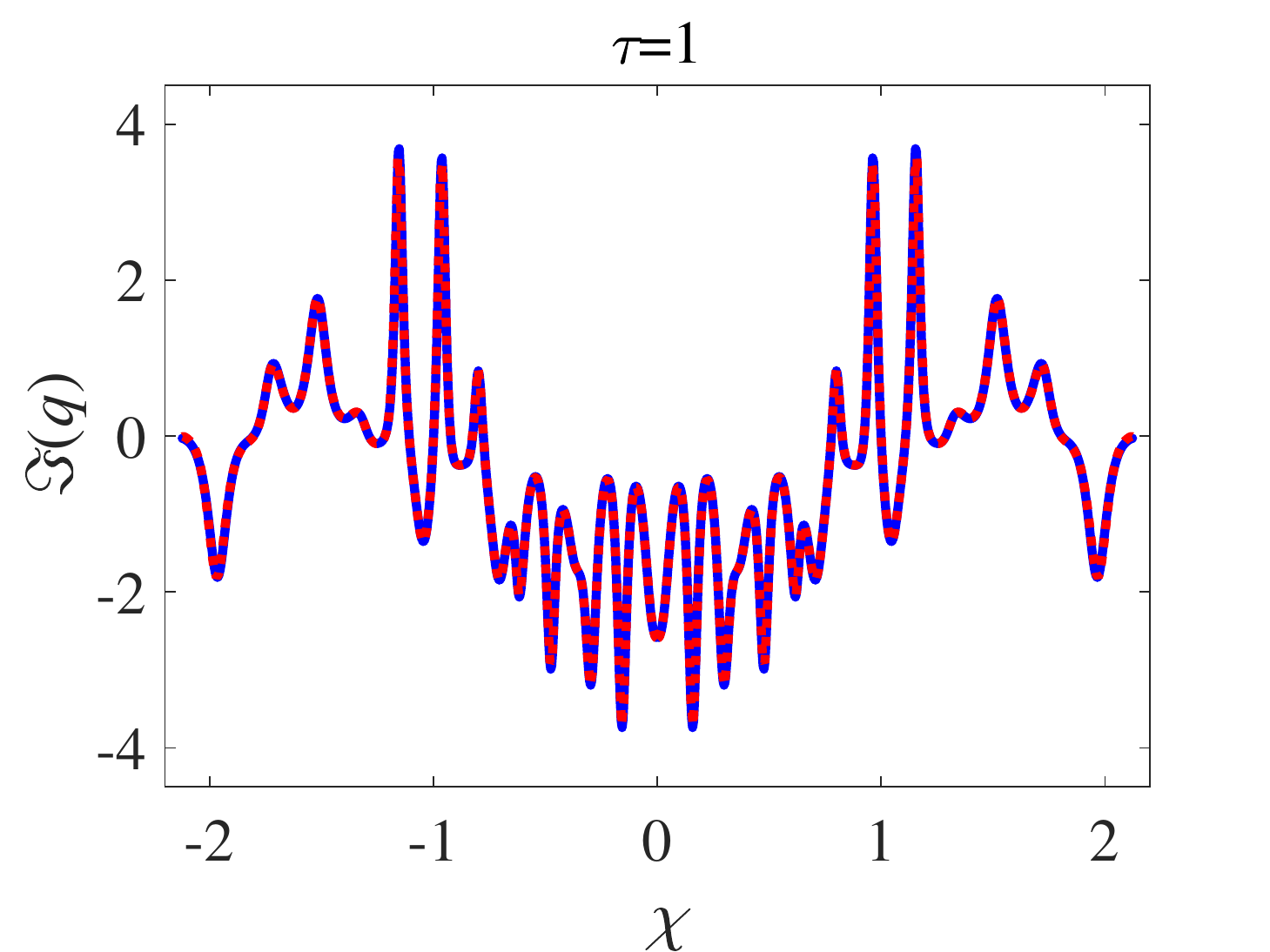}
\caption{The comparison between the exact solution ($15$-th order breather) and its asymptotic solution in the genus-one and genus-three region with $\tau=1$ (as shown by the yellow dashed line in Fig. \ref{fig:3-d}). The blue solid line is the exact solution and the red dotted line is the asymptotic solution. The left panel is the modulus of this solution, the middle one is the real part and the right one is the imaginary part.}
\label{fig:os}
\end{figure}

Next, we give a detailed computation about these four different asympototics.
\section{The algebraic-decay region}
\label{sec:al-decay}
Now, we begin to study the asymptotics in the algebraic-decay region. For this region, the controlling phase term $\vartheta(\lambda; \chi, \tau)$ has at least three real critical points, and we can analyze it through the Deift-Zhou nonlinear steepest descent method directly. By choosing one suit $\chi$ and $\tau$, we give the contour plot of $\Im(\vartheta(\lambda; \chi, \tau))$ in Fig.\ref{fig:al-decay}.
\begin{figure}[ht]
\centering
\includegraphics[width=0.45\textwidth]{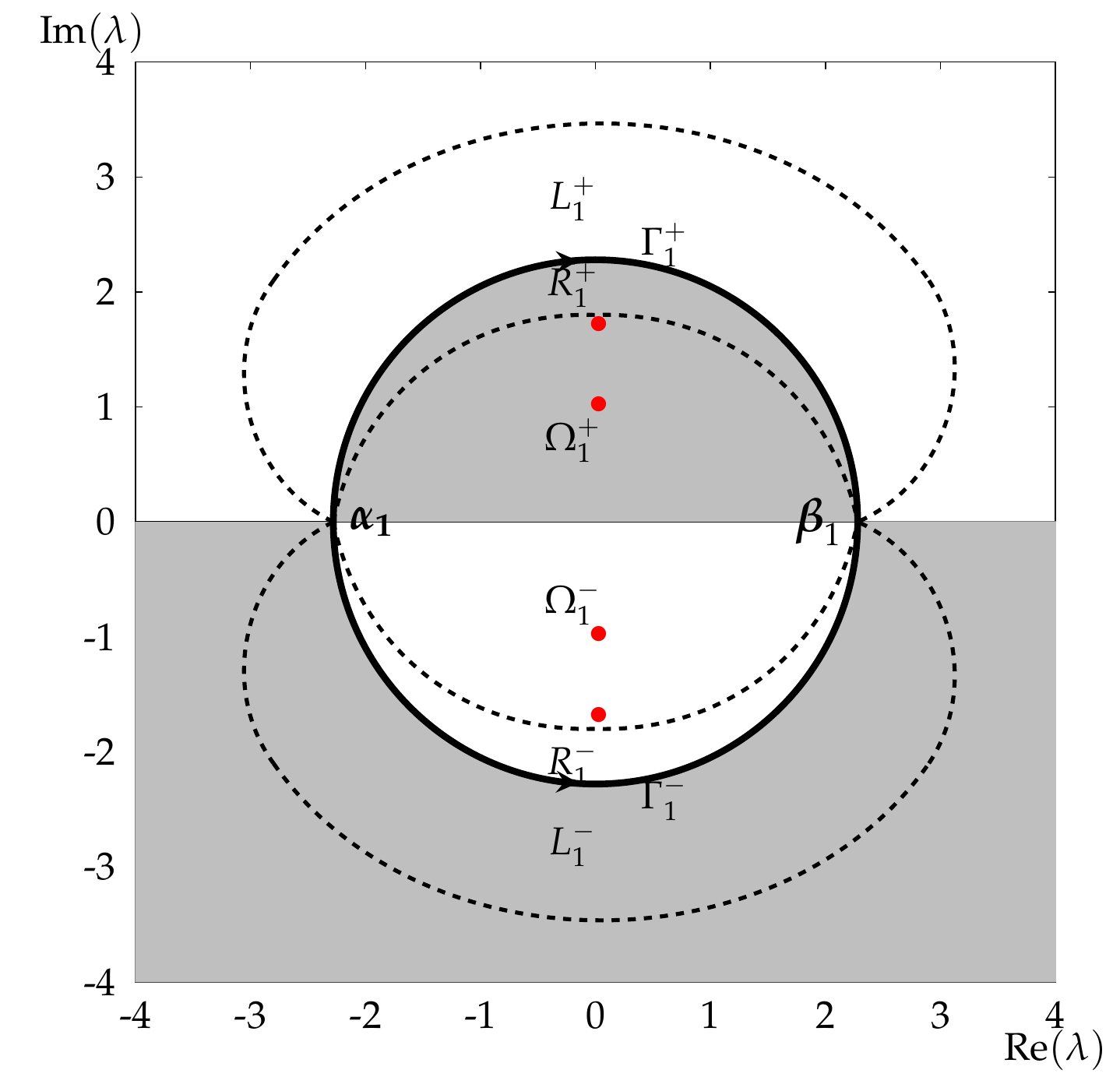}
\centering
\includegraphics[width=0.45\textwidth]{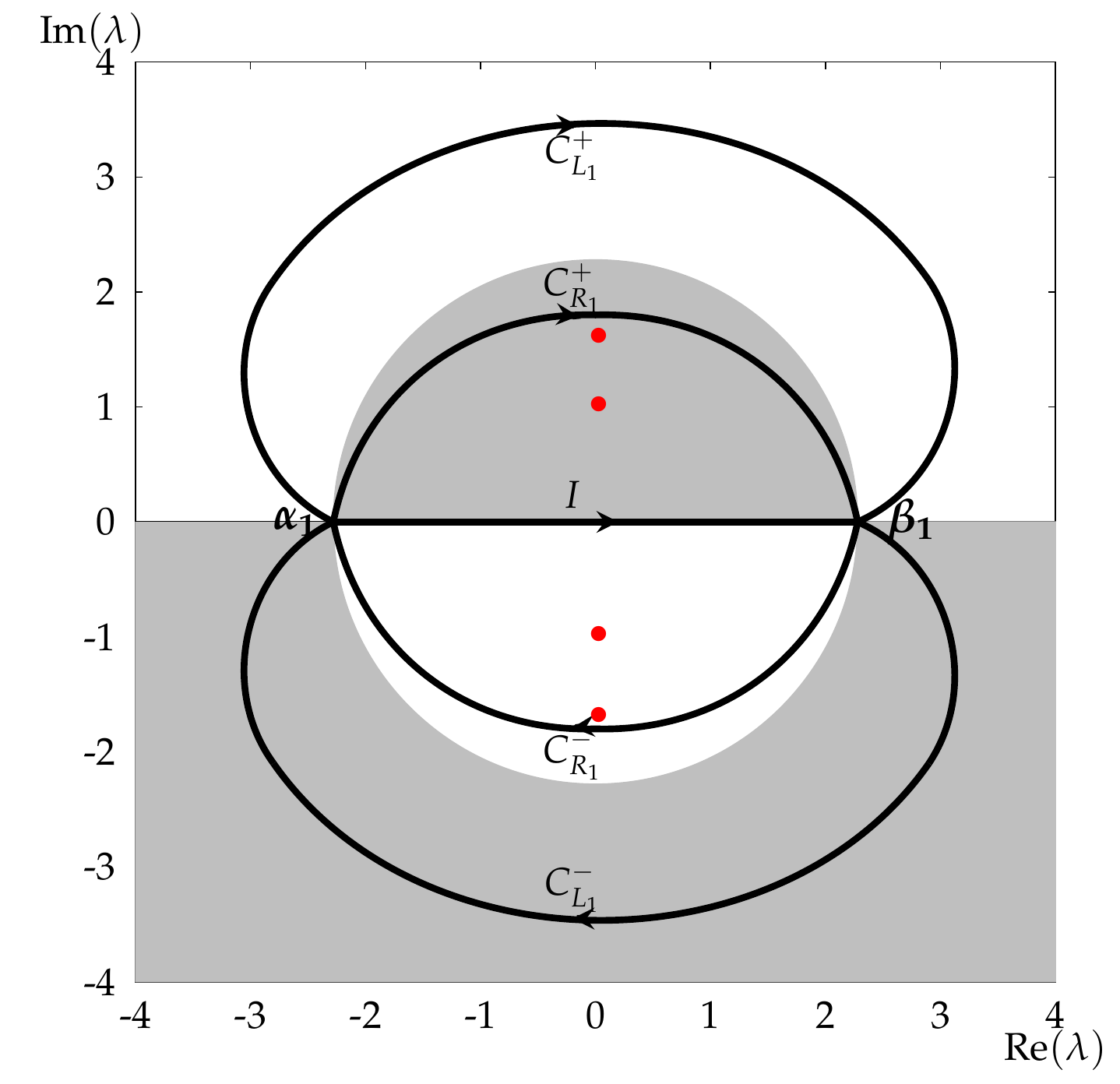}
\caption{The contour of ${\Im}\left(\vartheta\left(\lambda; \frac{1}{2}, 0\right)\right)$ in the algebraic-decay region by choosing $k=2$, where ${\Im}\left(\vartheta\left(\lambda; \frac{1}{2}, 0\right)\right)<0$ (shaded) and ${\Im}\left(\vartheta\left(\lambda; \frac{1}{2}, 0\right)\right)>0$ (unshaded), the red dots are the singularities at $\lambda=\pm\ii, \lambda=\pm2\ii$. The left one gives the original contour and the right one is the corresponding contour deformation. }
\label{fig:al-decay}
\end{figure}

With the sign of $\Im(\vartheta(\lambda; \chi, \tau))$, we can deform $\mathbf{R}(\lambda; \chi, \tau)$ through the theory of Deift-Zhou nonlinear steepest descent method.

Set
\begin{equation}
\mathbf{S}_{1}(\lambda; \chi, \tau):=\left\{\begin{aligned}&\mathbf{R}(\lambda; \chi, \tau)\ee^{-\ii n\vartheta(\lambda; \chi, \tau)\sigma_3}\mathbf{Q}_{c}\ee^{\ii n\vartheta(\lambda; \chi, \tau)\sigma_3},\quad\lambda\in D_0\cap\left(D_{1}^{+}\cup D_{1}^{-}\right)^{c},\\
&\mathbf{R}(\lambda; \chi, \tau),\quad\quad\quad\quad\quad\quad\quad\quad\quad\quad\quad\quad\quad\quad\quad\quad\quad {\rm otherwise},\end{aligned}\right.
\end{equation}
where $D_{1}^{+}=\Omega_1^+\cup R_1^{+}, D_1^{-}=\Omega_1^{-}\cup R_1^-$. Then the jump contour about $\mathbf{S}_{1}(\lambda; \chi, \tau)$ transfers into the boundary of $D_{1}^{\pm}$, which can be seen from the left panel of Fig.\ref{fig:al-decay}. Based on the decomposition in equation \eqref{remark:decom}, we can define a new matrix $\mathbf{T}_{1}(\lambda; \chi, \tau)$ related to $\mathbf{S}_{1}(\lambda; \chi, \tau)$ by introducing the ``lens" domains $L_1^{\pm}, R_1^{\pm}$ on either side of $\partial D_{1}^{\pm}$, see the left panel again. Set
\begin{equation}
\begin{aligned}
\mathbf{T}_{1}(\lambda; \chi, \tau):&=\mathbf{S}_{1}(\lambda; \chi, \tau)\ee^{-\ii n\vartheta(\lambda; \chi, \tau)\sigma_3}
\left(\mathbf{Q}_{R}^{[2]}\right)^{-1}\ee^{\ii n\vartheta(\lambda; \chi, \tau)\sigma_3},\quad &\lambda\in L_1^{+},\\
\mathbf{T}_{1}(\lambda; \chi, \tau):&=\mathbf{S}_{1}(\lambda; \chi, \tau)\mathbf{Q}_{L}^{[2]}\ee^{-\ii n\vartheta(\lambda; \chi, \tau)\sigma_3}
\mathbf{Q}_{C}^{[2]}\ee^{\ii n\vartheta(\lambda; \chi, \tau)\sigma_3},\quad &\lambda\in R_{1}^{+},\\
\mathbf{T}_{1}(\lambda; \chi, \tau):&=\mathbf{S}_{1}(\lambda; \chi, \tau)\mathbf{Q}_{L}^{[2]},\quad &\lambda\in \Omega^{+}_{1},\\
\mathbf{T}_{1}(\lambda; \chi, \tau):&=\mathbf{S}_{1}(\lambda; \chi, \tau)\ee^{-\ii n\vartheta(\lambda; \chi, \tau)\sigma_3}
\left(\mathbf{Q}_{R}^{[1]}\right)^{-1}\ee^{\ii n\vartheta(\lambda; \chi, \tau)\sigma_3},\quad &\lambda\in L_{1}^{-},\\
\mathbf{T}_{1}(\lambda; \chi, \tau):&=\mathbf{S}_{1}(\lambda; \chi, \tau)\mathbf{Q}_{L}^{[1]}\ee^{-\ii n\vartheta(\lambda; \chi, \tau)\sigma_3}
\mathbf{Q}_{C}^{[1]}
\ee^{\ii n\vartheta(\lambda; \chi, \tau)\sigma_3},\quad &\lambda\in R_{1}^{-},\\
\mathbf{T}_{1}(\lambda; \chi, \tau):&=\mathbf{S}_{1}(\lambda; \chi, \tau)\mathbf{Q}_{L}^{[1]},\quad &\lambda\in \Omega_{1}^{-},
\end{aligned}
\end{equation}
in other domains, we set $\mathbf{T}_{1}(\lambda; \chi, \tau)=\mathbf{S}_{1}(\lambda; \chi, \tau)$, then the jump conditions about $\mathbf{T}_{1}(\lambda; \chi, \tau)$ on the arcs of jump contour will be changed into the following form,
\begin{equation}
\begin{aligned}
\mathbf{T}_{1,+}(\lambda; \chi, \tau)&=\mathbf{T}_{1,-}(\lambda; \chi, \tau)\ee^{-\ii n\vartheta(\lambda; \chi, \tau)\sigma_3}\mathbf{Q}_{R}^{[2]}\ee^{\ii n\vartheta(\lambda; \chi, \tau)\sigma_3}, \quad&\lambda\in C_{L_1}^{+},\\
\mathbf{T}_{1,+}(\lambda; \chi, \tau)&=\mathbf{T}_{1,-}(\lambda; \chi, \tau)\ee^{-\ii n\vartheta(\lambda; \chi, \tau)\sigma_3}\mathbf{Q}_{C}^{[2]}\ee^{\ii n\vartheta(\lambda; \chi, \tau)\sigma_3}, \quad&\lambda\in C_{R_1}^{+},\\
\mathbf{T}_{1,+}(\lambda; \chi, \tau)&=\mathbf{T}_{1,-}(\lambda; \chi, \tau)\ee^{-\ii n\vartheta(\lambda; \chi, \tau)\sigma_3}\mathbf{Q}_{R}^{[1]}\ee^{\ii n\vartheta(\lambda; \chi, \tau)\sigma_3}, \quad&\lambda\in C_{L_1}^{-},\\
\mathbf{T}_{1,+}(\lambda; \chi, \tau)&=\mathbf{T}_{1,-}(\lambda; \chi, \tau)\ee^{-\ii n\vartheta(\lambda; \chi, \tau)\sigma_3}\mathbf{Q}_{C}^{[1]}\ee^{\ii n\vartheta(\lambda; \chi, \tau)\sigma_3}, \quad&\lambda\in C_{R_1}^{-},\\
\mathbf{T}_{1,+}(\lambda; \chi, \tau)&=\mathbf{T}_{1,-}(\lambda; \chi, \tau)2^{\sigma_3},\quad&\lambda\in I.
\end{aligned}
\end{equation}
From the sign chart in Fig.\ref{fig:al-decay}, we can see that when $n$ is large, the jump matrices about $\mathbf{T}_{1}(\lambda; \chi, \tau)$ will decay into the identity exponentially except for the interval $I=\left[\alpha_1, \beta_1\right]$. As for this jump condition, we can construct the parametrix to deal with it.
\subsection{Parametrix construction} With this constant jump condition on $I=\left[\alpha_1, \beta_1\right]$, we first construct the outer parametrix $\dot{\mathbf{T}}_{1}^{\rm out}(\lambda; \chi, \tau)$ by the Plemelj formula,
\begin{equation}\label{eq:T1out}
\dot{\mathbf{T}}_{1}^{\rm out}(\lambda; \chi, \tau)=\left(\frac{\lambda-\alpha_1}{\lambda-\beta_1}\right)^{\ii p\sigma_3}, \quad p:=\frac{\log(2)}{2\pi}, \quad \lambda\in\mathbb{C}\setminus I.
\end{equation}
It is easy to check that $\dot{\mathbf{T}}_{1}^{\rm out}(\lambda; \chi, \tau)$ can match $\mathbf{T}_{1}(\lambda; \chi, \tau)$ very well on $I$ except for the endpoints $\alpha_1, \beta_1$, where $\dot{\mathbf{T}}_{1}^{\rm out}(\lambda; \chi, \tau)$ has singularities. Thus we need to construct the inner parametrix in the neighborhood of $\alpha_1$ and $\beta_1$ separately. In the small neighborhood of $\alpha_1$ and $\beta_1$, we hope that the inner parametrix should satisfy all the jump conditions with $\mathbf{T}_{1}(\lambda; \chi, \tau)$. Denote $D_{\alpha_1}(\delta), D_{\beta_1}(\delta)$ as the small disks centered at $\alpha_1, \beta_1$ with the small radius $\delta$ respectively. Firstly, set a conformal map $f_{\alpha_1}(\lambda; \chi, \tau), f_{\beta_1}(\lambda; \chi, \tau)$ as
\begin{equation}\label{eq:con-map}
f_{\alpha_1}(\lambda; \chi, \tau)^2=2\left[\vartheta(\alpha_1, \chi, \tau)-\vartheta(\lambda; \chi, \tau)\right],\quad f_{\beta_1}(\lambda; \chi, \tau)^2=2\left[\vartheta(\lambda; \chi, \tau)-\vartheta(\beta_1; \chi, \tau)\right].
\end{equation}
It is easy to see that $f_{\alpha_1}(\alpha_1; \chi, \tau)=0, f_{\beta_1}(\beta_1; \chi, \tau)=0$, and $(f'_{\alpha_1}(\alpha_1; \chi, \tau))^2=-\vartheta''(\alpha_1; \chi, \tau), (f'_{\beta_1}(\beta_1; \chi, \tau))^2=\vartheta''(\beta_1; \chi, \tau)$, we take the analytic square root such that $f'_{\alpha_1}(\alpha_1; \chi, \tau)=-\sqrt{-\vartheta''(\alpha_1; \chi, \tau)}<0, f'_{\beta_1}(\beta_1; \chi, \tau)=\sqrt{\vartheta''(\beta_1; \chi, \tau)}>0.$ Based on this new defined $f_{\alpha_1}(\lambda; \chi, \tau), f_{\beta_1}(\lambda; \chi, \tau)$, we can construct a holomorphic function in $D_{\alpha_1}(\delta), D_{\beta_1}(\delta)$,
\begin{equation}
\mathbf{H}_{\alpha_1}(\lambda; \chi, \tau):=\left(\frac{\alpha_1-\lambda}{f_{\alpha_1}(\lambda; \chi, \tau)}\right)^{\ii p\sigma_3}(\beta_1-\lambda)^{-\ii p\sigma_3}\left(\ii\sigma_2\right),\quad \mathbf{H}_{\beta_1}(\lambda; \chi, \tau):=(\lambda-\alpha_1)^{\ii p\sigma_3}\left(\frac{f_{\beta_1}(\lambda; \chi, \tau)}{\lambda-\beta_1}\right)^{\ii p\sigma_3}.
\end{equation}
With the L'Hospital's rule, at the points $\alpha_1, \beta_1$, we have
\begin{equation}\label{eq:Hab}
\mathbf{H}_{\alpha_1}(\alpha_1; \chi, \tau)=\left(\frac{-1}{f'_{\alpha_1}(\alpha_1; \chi, \tau)}\right)^{\ii p\sigma_3}\left(\beta_1-\alpha_1\right)^{-\ii p\sigma_3}(\ii\sigma_2), \quad\mathbf{H}_{\beta_1}(\beta_1; \chi, \tau)=\left(\beta_1-\alpha_1\right)^{\ii p\sigma_3}\left(f_{\beta_1}'(\beta_1; \chi, \tau)\right)^{\ii p\sigma_3}.
\end{equation}
In $D_{\alpha_1}(\delta), D_{\beta_1}(\delta)$, we introduce two new variables $\zeta_{\alpha_1}=n^{1/2}f_{\alpha_1}(\lambda; \chi, \tau), \zeta_{\beta_1}=n^{1/2}f_{\beta_1}(\lambda; \chi, \tau)$ and two matrices $\mathbf{U}_{\alpha_1}(\lambda; \chi, \tau), \mathbf{U}_{\beta_1}(\lambda; \chi, \tau)$ defined as
\begin{equation}
\begin{aligned}
\mathbf{U}_{\alpha_1}(\lambda; \chi, \tau):&=\mathbf{T}_{1}(\lambda; \chi, \tau)\ee^{-\ii n\vartheta(\alpha_1; \chi, \tau)\sigma_3}\ii\sigma_2,&\quad \lambda\in D_{\alpha_1}(\delta),\\
\mathbf{U}_{\beta_1}(\lambda; \chi, \tau):&=\mathbf{T}_{1}(\lambda; \chi, \tau)\ee^{-\ii n\vartheta(\beta_1; \chi, \tau)\sigma_3},&\quad \lambda\in D_{\beta_1}(\delta).
\end{aligned}
\end{equation}
Then the jump conditions satisfied by $\mathbf{U}_{\alpha_1}(\lambda; \chi, \tau), \mathbf{U}_{\beta_1}(\lambda; \chi, \tau)$ are shown in Fig.\ref{fig:rays-para} with the variables $\zeta=\zeta_{\alpha_1}, \zeta=\zeta_{\beta_1}$ respectively.
\begin{figure}[ht]
\centering
\includegraphics[width=0.45\textwidth]{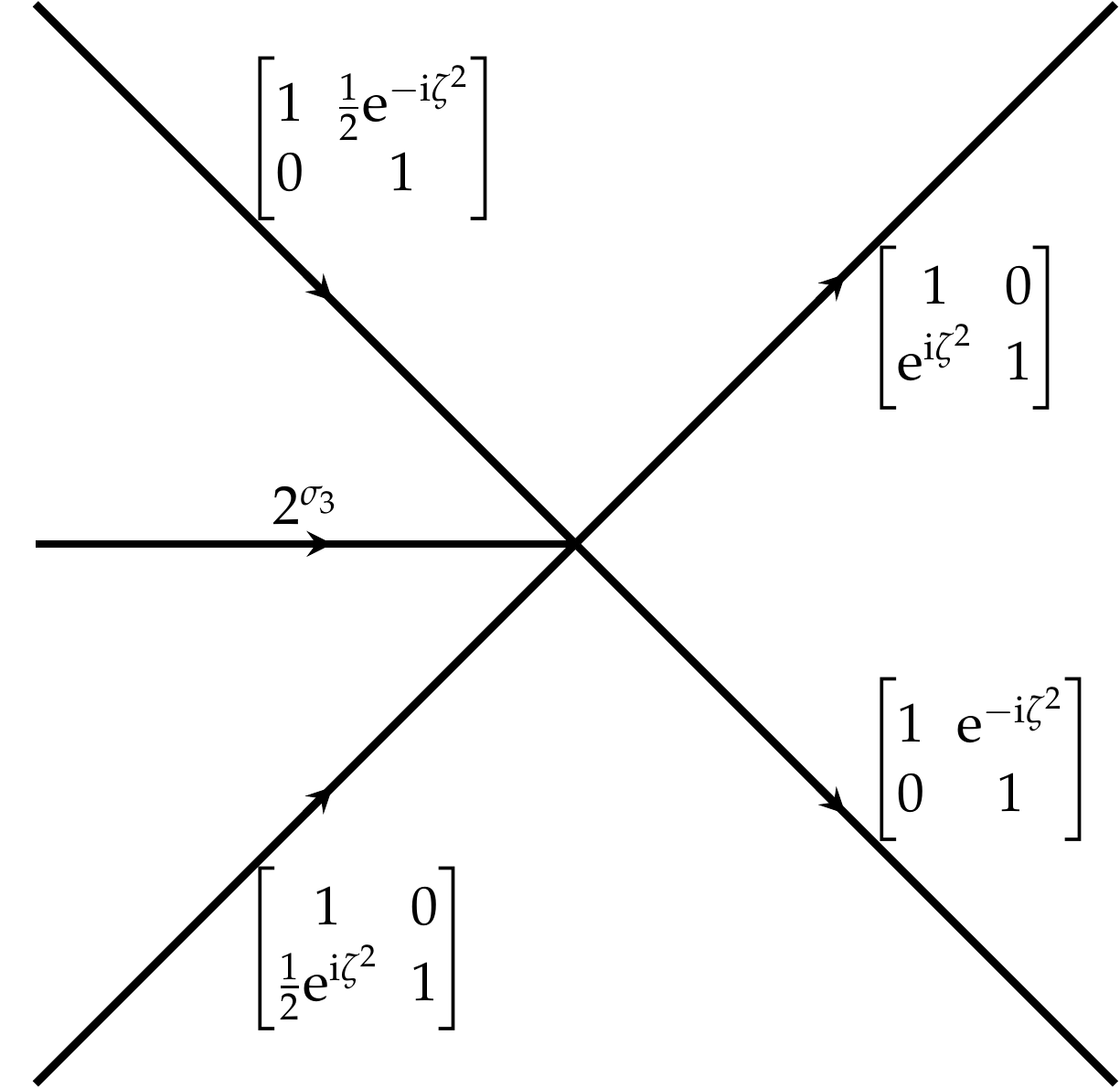}
\caption{The jump conditions about $\mathbf{U}=\mathbf{U}_{\alpha_1}(\lambda; \chi, \tau)$ and $\mathbf{U}=\mathbf{U}_{\beta_1}(\lambda; \chi, \tau)$ with the new conformal variables $\zeta_{\alpha_1}, \zeta_{\beta_1}$.}
\label{fig:rays-para}
\end{figure}

In \cite{Bilman-Duke-2019}, the authors showed that the solutions $\mathbf{U}_{\alpha_1}(\lambda; \chi, \tau), \mathbf{U}_{\beta_1}(\lambda; \chi, \tau)$ can be given by the standard parabolic cylinder function. We now give a RHP about the matrix $\mathbf{U}(\zeta)$.
\begin{rhp}
(Parabolic cylinder parametrix)
There exists a $2\times 2$ matrix $\mathbf{U}(\zeta)$ satisfying the following conditions.
\begin{itemize}
\item {\bf Analyticity:} $\mathbf{U}(\zeta)$ is analytic except for the five rays shown in Fig. \ref{fig:rays-para}.
\item {\bf Jump condition:} When $\lambda$ is on these five rays, $\mathbf{U}_{+}(\zeta)=\mathbf{U}_{-}(\zeta)\mathbf{V}^{\rm PC}(\zeta), $ where $\mathbf{V}^{\rm PC}(\zeta)$ is the jump matrices given in Fig. \ref{fig:rays-para}.
\item {\bf Normalization:} $\mathbf{U}(\zeta)\zeta^{\ii p\sigma_3}\to\mathbb{I}$ as $\lambda\to\infty$.
\end{itemize}
\end{rhp}
The solution of $\mathbf{U}(\zeta)$ can be expressed by the parabolic cylinder function. What we need is the asymptotic expansion when $\zeta\to\infty$. Following the result in \cite{Bilman-Duke-2019}, when $\zeta\to\infty$, the asymptotic expansion about $\mathbf{U}(\zeta)\zeta^{\ii p\sigma_3}$ has the following formula,
\begin{equation}
\mathbf{U}(\zeta)\zeta^{\ii p\sigma_3}=\mathbb{I}+\frac{1}{2\ii\zeta}\begin{bmatrix}0&\alpha\\
-\beta&0
\end{bmatrix}+\begin{bmatrix}\mathcal{O}(\zeta^{-2})&\mathcal{O}(\zeta^{-3})\\
\mathcal{O}(\zeta^{-3})&\mathcal{O}(\zeta^{-2})
\end{bmatrix},\quad\zeta\to\infty,
\end{equation}
where
\begin{equation}\label{eq:alphabeta}
\alpha=2^{\frac{3}{4}}\sqrt{2\pi}\Gamma\left(\frac{\ii\ln(2)}{2\pi}\right)^{-1}\ee^{\ii\pi/4}\ee^{\ii(\ln(2))^2/(2\pi)},\quad \beta=-\alpha^*.
\end{equation}
By using $\mathbf{U}(\zeta)$, we can define the inner parametrix in $D_{\alpha_1}(\delta)$ and $D_{\beta_1}(\delta)$,
\begin{equation}
\begin{aligned}
\dot{\mathbf{T}}_{1}^{\alpha_1}(\lambda; \chi, \tau):&=n^{-\ii p\sigma_3/2}\ee^{-\ii n\vartheta(\alpha_1; \chi, \tau)\sigma_3}\mathbf{H}_{\alpha_1}(\lambda; \chi, \tau)\mathbf{U}_{\alpha_1}(\zeta_{\alpha_1})(-\ii\sigma_2)\ee^{\ii n\vartheta(\alpha_1; \chi, \tau)\sigma_3}, &\quad\lambda\in D_{\alpha_1}(\delta),\\
\dot{\mathbf{T}}_{1}^{\beta_1}(\lambda; \chi, \tau):&=n^{\ii p\sigma_3/2}\ee^{-\ii n\vartheta(\beta_1; \chi, \tau)\sigma_3}\mathbf{H}_{\beta_1}(\lambda; \chi, \tau)\mathbf{U}_{\beta_1}(\zeta_{\beta_1})\ee^{\ii n\vartheta(\beta_1; \chi, \tau)\sigma_3},&\quad \lambda\in D_{\beta_1}(\delta).
\end{aligned}
\end{equation}
Then the whole parametrix about $\mathbf{T}_{1}(\lambda; \chi, \tau)$ is
\begin{equation}
\dot{\mathbf{T}}_{1}(\lambda; \chi, \tau):=\left\{\begin{aligned}&\dot{\mathbf{T}}_{1}^{\alpha_1}(\lambda; \chi, \tau),&\quad\lambda\in D_{\alpha_1}(\delta),\\
&\dot{\mathbf{T}}_{1}^{\beta_1}(\lambda; \chi, \tau),&\quad\lambda\in D_{\beta_1}(\delta),\\
&\dot{\mathbf{T}}_{1}^{\rm out}(\lambda; \chi, \tau),&\quad\lambda\in \mathbb{C}\setminus\left(I\cup\overline{D_{\alpha_1}(\delta)}\cup \overline{D_{\beta_1}(\delta)}\right).\\
\end{aligned}\right.
\end{equation}
Next we will analyze the error between the parametrix $\dot{\mathbf{T}}_{1}(\lambda; \chi, \tau)$ and $\mathbf{T}_{1}(\lambda; \chi, \tau)$.
\subsection{Error analysis}
Set the error function between $\mathbf{T}_{1}(\lambda; \chi, \tau)$ and its parametrix $\dot{\mathbf{T}}_{1}(\lambda; \chi, \tau)$ as $\mathcal{E}_1(\lambda; \chi, \tau)$, that is
\begin{equation}
\mathcal{E}_1(\lambda; \chi, \tau):=\mathbf{T}_{1}(\lambda; \chi, \tau)\left(\dot{\mathbf{T}}_{1}(\lambda; \chi, \tau)\right)^{-1},
\end{equation}
the corresponding discontinuous contours about $\mathcal{E}_1(\lambda; \chi, \tau)$ are set as $\Sigma_{\mathcal{E}_1}$ and the jump matrix about $\mathcal{E}_1(\lambda; \chi, \tau)$ is set as $\mathbf{V}_{\mathcal{E}_1}(\lambda; \chi, \tau)$, that is $\mathcal{E}_{1,+}(\lambda; \chi, \tau)=\mathcal{E}_{1,-}(\lambda; \chi, \tau)\mathbf{V}_{\mathcal{E}_1}(\lambda; \chi, \tau)$. From the definition of $\dot{\mathbf{T}}_{1}(\lambda; \chi, \tau)$, we know that the jump arcs about $\mathcal{E}_{1}(\lambda; \chi, \tau)$ are shown in Fig. \ref{fig:jump-E-al}.
\begin{figure}[ht]
\centering
\includegraphics[width=0.45\textwidth]{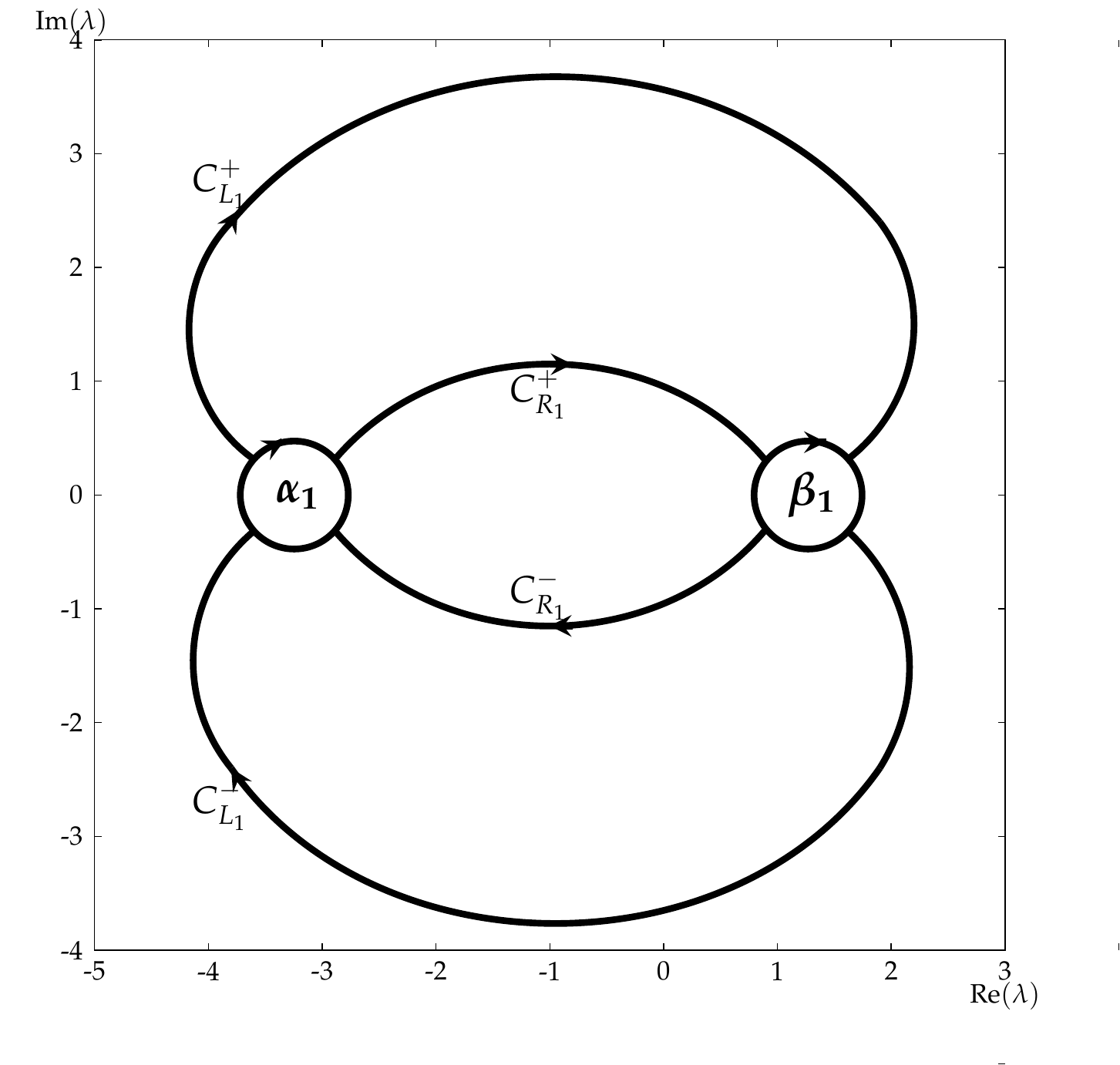}
\caption{The jump contours about $\mathcal{E}_{1}(\lambda; \chi, \tau)$.}
\label{fig:jump-E-al}
\end{figure}

When $\lambda\in \left(C_{L_1}^{\pm}\cup C_{R_1}^{\pm}\right)\cap \Sigma_{\mathcal{E}_1}$, the jump condition $\mathbf{V}_{\mathcal{E}_1}(\lambda; \chi, \tau)$ is
\begin{equation}
\mathbf{V}_{\mathcal{E}_1}(\lambda; \chi, \tau)=\dot{\mathbf{T}}^{\rm out}_{1}(\lambda; \chi, \tau)\mathbf{T}_{1,-}(\lambda; \chi, \tau)^{-1}\mathbf{T}_{1,+}(\lambda; \chi, \tau)\dot{\mathbf{T}}_{1}^{\rm out}(\lambda; \chi, \tau)^{-1},
\end{equation}
which will exponentially approach to the identity as $n$ is large, thus the errors in this arcs are
\begin{equation}
\|\mathbf{V}_{\mathcal{E}_1}(\lambda; \chi, \tau)-\mathbb{I}\|=\mathcal{O}(\ee^{-\mu_1n}),
\end{equation}
where $\mu_1$ is a positive constant. When $\lambda\in \partial D_{\alpha_1}(\delta)$ and $\lambda\in \partial D_{\beta_1}(\delta)$, the jump conditions between $\mathcal{E}_{1,+}(\lambda; \chi, \tau)$ and $\mathcal{E}_{1,-}(\lambda; \chi, \tau)$ change into
\begin{multline}\label{eq:jump-E-alpha1}
\mathcal{E}_{1,+}(\lambda; \chi, \tau)=\mathcal{E}_{1,-}(\lambda; \chi, \tau)n^{-\frac{1}{2}\ii p\sigma_3}\ee^{-\ii n\vartheta(\alpha_1; \chi, \tau)\sigma_3}\mathbf{H}_{\alpha_1}(\lambda; \chi, \tau)\\\cdot\left(\mathbb{I}+\frac{1}{2\ii n^{1/2}f_{\alpha_1}(\lambda; \chi, \tau)}\begin{bmatrix}0&\alpha\\
-\beta&0
\end{bmatrix}+\begin{bmatrix}\mathcal{O}(n^{-1})&\mathcal{O}(n^{-\frac{3}{2}})\\
\mathcal{O}(n^{-\frac{3}{2}})&\mathcal{O}(n^{-1})
\end{bmatrix}\right)\mathbf{H}_{\alpha_1}(\lambda; \chi, \tau)^{-1}\ee^{\ii n\vartheta(\alpha_1; \chi, \tau)\sigma_3}n^{\frac{1}{2}\ii p\sigma_3},\,\lambda\in\partial D_{\alpha_1}(\delta),
\end{multline}
and
\begin{multline}\label{eq:jump-E-beta1}
\mathcal{E}_{1,+}(\lambda; \chi, \tau)=\mathcal{E}_{1,-}(\lambda; \chi, \tau)n^{\frac{1}{2}\ii p\sigma_3}\ee^{-\ii n\vartheta(\beta_1; \chi, \tau)\sigma_3}\mathbf{H}_{\beta_1}(\lambda; \chi, \tau)\\\cdot\left(\mathbb{I}+\frac{1}{2\ii n^{1/2}f_{\beta_1}(\lambda; \chi, \tau)}\begin{bmatrix}0&\alpha\\
-\beta&0
\end{bmatrix}+\begin{bmatrix}\mathcal{O}(n^{-1})&\mathcal{O}(n^{-\frac{3}{2}})\\
\mathcal{O}(n^{-\frac{3}{2}})&\mathcal{O}(n^{-1})
\end{bmatrix}\right)\mathbf{H}_{\beta_1}(\lambda; \chi, \tau)^{-1}\ee^{\ii n\vartheta(\beta_1; \chi, \tau)\sigma_3}n^{-\frac{1}{2}\ii p\sigma_3},\,\lambda\in\partial D_{\beta_1}(\delta).
\end{multline}
With the jump conditions, we know that the errors about $\mathbf{V}_{\mathcal{E}_{1}}(\lambda; \chi, \tau)$ satisfy
\begin{equation}
\|\mathbf{V}_{\mathcal{E}_1}(\lambda; \chi, \tau)-\mathbb{I}\|=\mathcal{O}(n^{-1/2}).
\end{equation}
\subsection{The asymptotic expression of $q^{[n]}(n\chi, n\tau)$ in the algebraic-decay region}
From the beginning expression $q^{[n]}(n\chi, n\tau)$ in Eq.\eqref{eq:qn-al}, we can recover the solution $q^{[n]}(n\chi, n\tau)$ from $\mathbf{T}_{1}(\lambda; \chi, \tau)$ by the following formula,
\begin{equation}\label{eq:qn-al-1}
\begin{aligned}
q^{[n]}(n\chi, n\tau)&=2\ii\lim\limits_{\lambda\to\infty}\lambda \mathbf{T}_{1}(\lambda; \chi, \tau)_{12}\\
&=2\ii\lim\limits_{\lambda\to\infty}\lambda\left(\mathcal{E}_{1}(\lambda; \chi, \tau)\dot{\mathbf{T}}^{\rm out}_{1}(\lambda; \chi, \tau)\right)_{12}\\
&=2\ii\lim\limits_{\lambda\to\infty}\lambda\left(\mathcal{E}_{1,11}(\lambda; \chi, \tau)\dot{\mathbf{T}}_{1,12}^{\rm out}(\lambda; \chi, \tau)+\mathcal{E}_{1,12}(\lambda; \chi, \tau)\dot{\mathbf{T}}_{1,22}^{\rm out}(\lambda; \chi, \tau)\right).
\end{aligned}
\end{equation}
By Eq.\eqref{eq:T1out}, we know that $\dot{\mathbf{T}}_{1}^{\rm out}(\lambda; \chi, \tau)$ is a diagonal matrix, thus Eq.\eqref{eq:qn-al-1} changes into
\begin{equation}
q^{[n]}(n\chi, n\tau)=2\ii\lim\limits_{\lambda\to\infty}\lambda\mathcal{E}_{1,12}(\lambda; \chi, \tau).
\end{equation}
Therefore, the key point is studying the asymptotic behavior of $\mathcal{E}_{1,12}$ when $\lambda\to\infty$. From the above analysis, we see that when $\lambda\in\left(C_{L_1}^{\pm}\cup C_{R_1}^{\pm}\right)\cap\Sigma_{\mathcal{E}_1}$, $\|\mathbf{V}_{\mathcal{E}_1}(\lambda; \chi, \tau)-\mathbb{I}\|=\mathcal{O}(\ee^{-\mu_1n})$, while for $\lambda\in\left(\partial D_{\alpha_1}(\delta)\cup \partial D_{\beta_1}(\delta)\right)$, $\|\mathbf{V}_{\mathcal{E}_1}(\lambda; \chi, \tau)-\mathbb{I}\|=\mathcal{O}(n^{-1/2})$. Obviously, we only need to consider the leading order error term for $\lambda\in\left(\partial D_{\alpha_1}(\delta)\cup \partial D_{\beta_1}(\delta)\right)$, others are the higher order error terms and can be omitted. Through the Plemelj formula, we have
\begin{equation}
\mathcal{E}_{1}(\lambda; \chi, \tau)=\mathbb{I}+\frac{1}{2\pi\ii}\int_{\partial D_{\alpha_1}(\delta)\cup \partial D_{\beta_1}(\delta)}\frac{\mathcal{E}_{1,-}(\xi; \chi, \tau)\left(\mathbf{V}_{\mathcal{E}_{1}}(\xi; \chi, \tau)-\mathbb{I}\right)}{\xi-\lambda}d\xi+\mathcal{O}(\ee^{-\mu_1n}),
\end{equation}
then the solution $q^{[n]}(n\chi, n\tau)$ can be given by
\begin{multline}
q^{[n]}(n\chi, n\tau)=-\frac{1}{\pi}\int_{\partial D_{\alpha_1}(\delta)\cup \partial D_{\beta_1}(\delta)}\left[\mathcal{E}_{1,11,-}(\xi; \chi, \tau)\mathbf{V}_{\mathcal{E}_1, 12}(\xi; \chi, \tau)+\mathcal{E}_{1,12,-}(\xi; \chi, \tau)\left(\mathbf{V}_{\mathcal{E}_{1},22}(\xi; \chi, \tau)-1\right)\right]d\xi\\+\mathcal{O}(\ee^{-\mu_1n}).
\end{multline}
From the jump conditions in Eq.\eqref{eq:jump-E-alpha1} and Eq.\eqref{eq:jump-E-beta1}, we know that $\mathbf{V}_{\mathcal{E}_1, 22}(\lambda; \chi, \tau)-1=\mathcal{O}(n^{-1}), \mathcal{E}_{1,12}(\lambda; \chi, \tau)=\mathcal{O}(n^{-1/2})$ in $L^{\infty}\left(\partial D_{\alpha_1}(\delta)\cup\partial D_{\beta_1}(\delta)\right)$, thus we have
\begin{equation}
q^{[n]}(n\chi, n\tau)=-\frac{1}{\pi}\int_{\partial D_{\alpha_1}(\delta)\cup\partial D_{\beta_1}(\delta)}\mathbf{V}_{\mathcal{E}_1,12}(\xi; \chi, \tau)d\xi+\mathcal{O}(n^{-3/2}).
\end{equation}
Again, from the jump conditions in Eq.\eqref{eq:jump-E-alpha1} and Eq.\eqref{eq:jump-E-beta1}, we can calculate $\mathbf{V}_{\mathcal{E}_{1}, 12}(\xi; \chi, \tau)$ as
\begin{equation}
\begin{aligned}
\mathbf{V}_{\mathcal{E}_1, 12}(\xi; \chi, \tau)&=\frac{n^{-\ii p}\ee^{-2\ii n\vartheta(\alpha_1; \chi, \tau)}}{2\ii n^{1/2}f_{\alpha_1}(\xi; \chi, \tau)}\beta H_{\alpha_1, 12}(\xi; \chi, \tau)^2+\mathcal{O}(n^{-3/2}),\quad \xi\in\partial D_{\alpha_1}(\delta),\\
\mathbf{V}_{\mathcal{E}_1, 12}(\xi; \chi, \tau)&=\frac{n^{\ii p}\ee^{-2\ii n\vartheta(\beta_1; \chi, \tau)}}{2\ii n^{1/2}f_{\beta_1}(\xi; \chi, \tau)}\alpha H_{\beta_1, 11}(\xi; \chi, \tau)^2+\mathcal{O}(n^{-3/2}),\quad \xi\in\partial D_{\beta_1}(\delta).
\end{aligned}
\end{equation}
With the residue theorem at $\xi=\alpha_1, \xi=\beta_1$, the potential function $q^{[n]}(n\chi, n\tau)$ in the algebraic-decay region can be given as
\begin{equation}\label{eq:qn-al-2}
q^{[n]}(n\chi, n\tau){=}\frac{1}{n^{1/2}}\Bigg[n^{-\ii p}\ee^{-2\ii n\vartheta(\alpha_1; \chi, \tau)}\frac{\beta H_{\alpha_1, 12}(\alpha_1; \chi, \tau)^2}{f'_{\alpha_1}(\alpha_1; \chi, \tau)}{+}n^{\ii p}\ee^{-2\ii n\vartheta(\beta_1; \chi, \tau)}\frac{\alpha H_{\beta_1, 11}(\beta_1; \chi, \tau)^2}{f'_{\beta_1}(\beta_1; \chi, \tau)}\Bigg]{+}\mathcal{O}(n^{-3/2}).
\end{equation}
Substituting $H_{\alpha_1,12}(\alpha_1; \chi, \tau), H_{\beta_1, 11}(\beta_1; \chi, \tau), f'_{\alpha_1}(\alpha_1; \chi, \tau), f'_{\beta_1}(\beta_1; \chi, \tau), \alpha $ and $\beta$ into Eq.\eqref{eq:qn-al-2}, we can get the expression in Eq.\eqref{eq:al}, it completes the asymptotic analysis for the algebraic-decay region (Theorem \ref{theo:al}).
\section{The genus-zero region}
\label{sec:genus-zero}
Last section, we give the asymptotic analysis in the algebraic-decay region. In this section, we continue to analyze this asymptotics in the genus-zero region. Compared to the algebraic-decay region, the genus-zero region becomes a little more complicated, the reason is that the original phase term $\vartheta(\lambda; \chi, \tau)$ can not be used for analyzing in this region, we must introduce a $g$ function relying on an algebraic curve with the genus-zero. We begin by giving a scalar RHP about $g_{2}(\lambda)\equiv g_{2}(\lambda; \chi, \tau)$ function.
\begin{rhp}\label{rhp-genus-zero}
(The $g_{2}(\lambda)$ function in the genus-zero region) Seek $g_{2}(\lambda)$ function with the following properties.
\begin{itemize}
\item{\bf Analyticity:} $g_{2}(\lambda)$ is analytic for $\lambda\in\mathbb{C}\setminus \Sigma_{g_2}^{\pm}$, where $\Sigma_{g_2}^{\pm}$ is the branch cut to be determined, where $g_2(\lambda)$ takes the continuous boundary values from the left and right sides.
\item {\bf Jump condition:} When $\lambda\in\Sigma_{g_2}^{\pm}$, $g_{2, \pm}(\lambda)$ is related by the jump conditions,
\begin{equation}\label{eq:g2}
g_{2,+}(\lambda)+g_{2,-}(\lambda)+2\vartheta(\lambda; \chi, \tau)=\kappa_2,
\end{equation}
where $\kappa_2$ is a real integration constant.
\item {\bf Normalization:} $g_2(\lambda)\to\mathcal{O}(\lambda^{-1})$ as $\lambda\to\infty$.
\item {\bf Symmetry:} $g_2(\lambda)$ has the symmetry condition
\begin{equation}
g_2(\lambda)=g_2(\lambda^*)^*.
\end{equation}
\end{itemize}
\end{rhp}
To solve $g_2(\lambda)$ in RHP \ref{rhp-genus-zero}, we introduce the algebraic curve $\mathcal{R}_{2}(\lambda)\equiv\mathcal{R}_{2}(\lambda; \chi, \tau)$ as
\begin{equation}\label{eq:R2}
\mathcal{R}_{2}(\lambda)=\sqrt{(\lambda-a_{2})(\lambda-a_2^*)}.
\end{equation}
As $\vartheta(\lambda; \chi, \tau)$ contains the logarithmic function and $\kappa_2$ is an unknown integration constant, we take the derivative with respect to $\lambda$ for Eq.\eqref{eq:g2} and get
\begin{equation}
g'_{2,+}(\lambda)+g'_{2,-}(\lambda)=-2\left[\chi+2\lambda\tau+\frac{\ii}{2}\left(\frac{1}{\lambda-\ii}-\frac{1}{\lambda+\ii}+\frac{1}{\lambda-k\ii}-\frac{1}{\lambda+k\ii}\right)\right],\quad \lambda\in\Sigma_{g_2}^{+}\cup \Sigma_{g_2}^{-}=\left[a_2^*, a_2\right].
\end{equation}
Through the definition $\mathcal{R}_{2}(\lambda)$ in Eq.\eqref{eq:R2}, we have
\begin{equation}
\left(\frac{g'_{2}(\lambda)}{\mathcal{R}_{2}(\lambda)}\right)_{+}-\left(\frac{g'_{2}(\lambda)}{\mathcal{R}_{2}(\lambda)}\right)_{-}=-\frac{2\left[\chi+2\lambda\tau+\frac{\ii}{2}\left(\frac{1}{\lambda-\ii}-\frac{1}{\lambda+\ii}+\frac{1}{\lambda-k\ii}-\frac{1}{\lambda+k\ii}\right)\right]}{\mathcal{R}_{2,+}(\lambda)}.
\end{equation}
By using the Plemelj formula, we have
\begin{equation}
g'_{2}(\lambda)=\frac{\mathcal{R}_{2}(\lambda)}{2\pi\ii}\int_{\Sigma_{g_2}^{+}\cup \Sigma_{g_2}^{-}}-\frac{2\left[\chi+2\xi\tau+\frac{\ii}{2}\left(\frac{1}{\xi-\ii}-\frac{1}{\xi+\ii}+\frac{1}{\xi-k\ii}-\frac{1}{\xi+k\ii}\right)\right]}{\mathcal{R}_{2,+}(\xi)(\xi-\lambda)}d\xi.
\end{equation}
Moreover, with the generalized residue theorem, we have
\begin{equation}\label{eq:dg}
g_2'(\lambda){=}\frac{\ii}{2}\mathcal{R}_{2}(\lambda)\left(\frac{1}{\mathcal{R}_{2}(\ii)(\lambda-\ii)}{-}\frac{1}{\mathcal{R}_{2}(-\ii)(\lambda+\ii)}{+}\frac{1}{\mathcal{R}_{2}(k\ii)(\lambda-k\ii)}{-}\frac{1}{\mathcal{R}_{2}(-k\ii)(\lambda+k\ii)}{-}4\ii\tau\right){-}\vartheta'(\lambda; \chi, \tau).
\end{equation}
From the normalization condition in the RHP, expanding $g_2'(\lambda)$ at $\lambda\to\infty$, we can get two equations for the coefficients of $\lambda^0$ and $\lambda^{-1}$, that is
\begin{equation}\label{eq:two-eq-genus-zero}
\begin{aligned}
\mathcal{O}(\lambda^0):\quad&\frac{\ii}{\mathcal{R}_{2}(\ii)}-\frac{\ii}{\mathcal{R}_{2}(-\ii)}+\frac{\ii}{\mathcal{R}_{2}(k\ii)}-\frac{\ii}{\mathcal{R}_{2}(-k\ii)}-4\Re(a_{2})\tau-2\chi=0,\\
\mathcal{O}(\lambda^{-1}):\quad& \frac{1}{\mathcal{R}_{1}(\ii)}+\frac{1}{\mathcal{R}_{1}(-\ii)}+\frac{k}{\mathcal{R}_{2}(k\ii)}+\frac{k}{\mathcal{R}_{2}(-k\ii)}+4\Re(a_{2})^2\tau-2\Im(a_2)^2\tau+2\Re(a_{2})\chi=0.
\end{aligned}
\end{equation}

For fixed $\chi$ and $\tau$, we can solve Eq.\eqref{eq:two-eq-genus-zero} numerically. For convenience, we introduce a new function $h_{2}(\lambda)\equiv h_2(\lambda; \chi, \tau)$ as
\begin{equation}\label{eq:h2}
h_{2}(\lambda):=g_2(\lambda)+\vartheta(\lambda; \chi, \tau),
\end{equation}
which will become the new controlling phase term by replacing the original phase term $\vartheta(\lambda; \chi, \tau)$.

Compared to the large order soliton/rogue wave, in the genus-zero region, there are always two equations to determine two unknown parameters. But here the determining equations \eqref{eq:two-eq-genus-zero} are more complicated, if we try to separate the real part and the imaginary part, the last determining equation will become a higher order polynomial function, and then we have to solve it numerically. For the higher genus region in later sections, we will confront with the similar problems. All of the unknown parameters appearing in the $g$-function will be calculated numerically.

By choosing one suit $\chi$ and $\tau$ in this region, we give the contour plot about $\Im(h_2(\lambda))$ in Fig. \ref{fig:genus-zero}.
\begin{figure}[ht]
\centering
\includegraphics[width=0.45\textwidth]{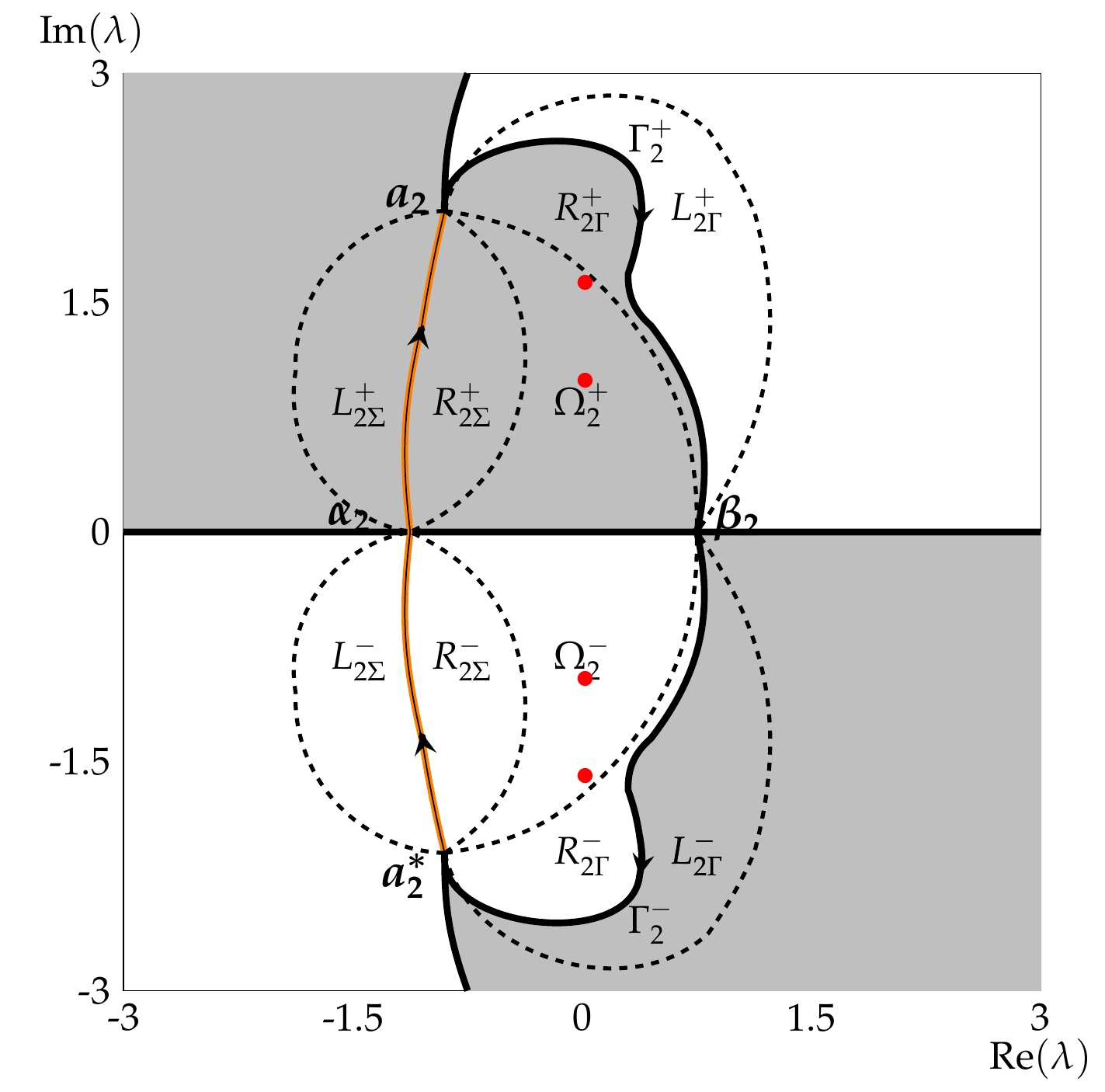}
\centering
\includegraphics[width=0.45\textwidth]{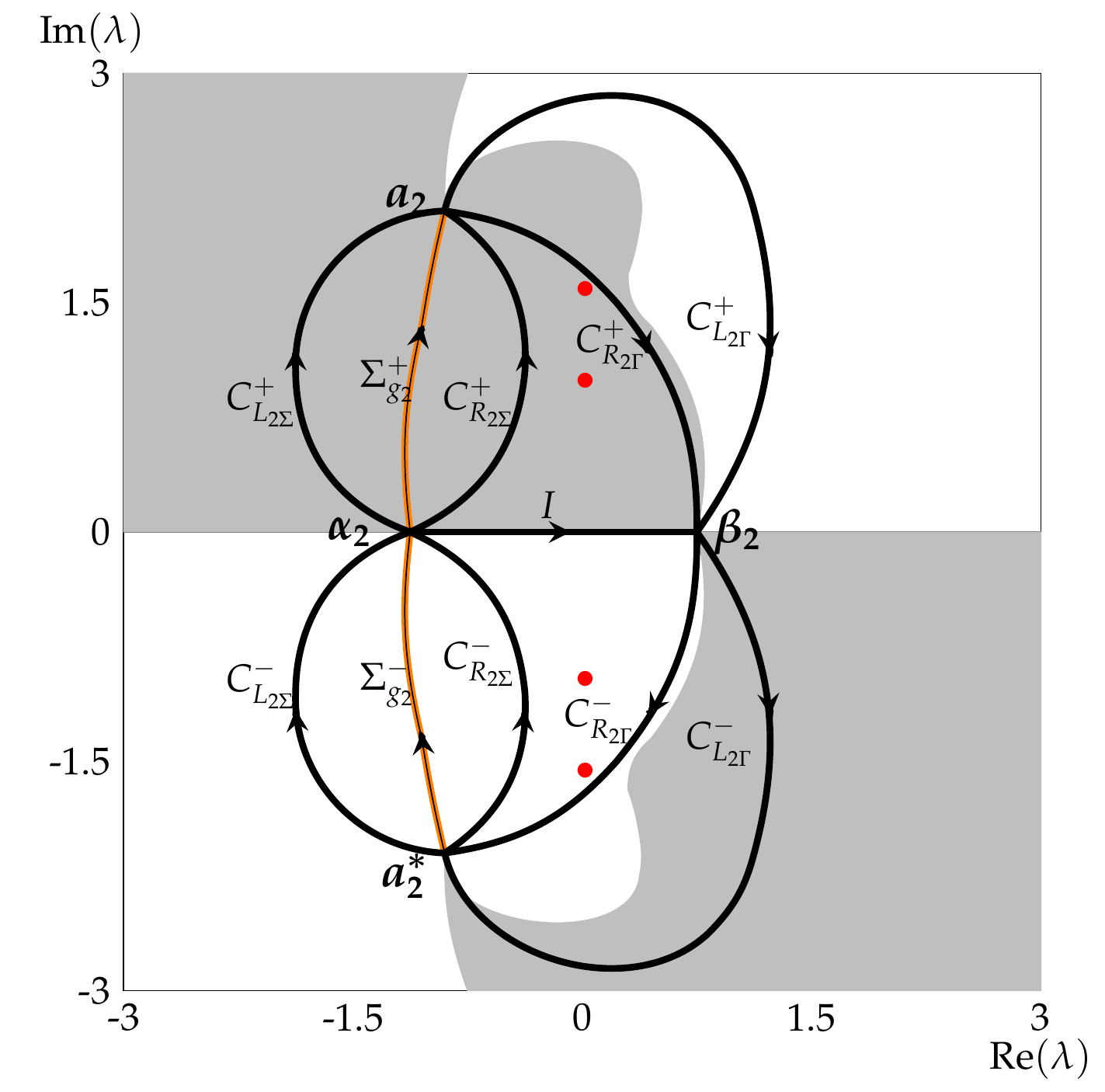}
\caption{The contour plot of ${\Im}(h_2(\lambda))$ in the genus-zero region with the parameters $\chi=\frac{4}{5}, \tau=\frac{1}{3}$ and $k=2$, where ${\Im}(h_2(\lambda))<0$ (shaded) and ${\Im}(h_2(\lambda))>0$ (unshaded), the red dots are the singularities $\lambda=\pm\ii, \lambda=\pm2\ii$. The left one gives the original contour and the right one is the corresponding contour deformation. }
\label{fig:genus-zero}
\end{figure}

With the sign chart in Fig. \ref{fig:genus-zero}, we can deform $\mathbf{R}(\lambda; \chi, \tau)$ by the Deift-Zhou nonlinear steepest descent method.

Set
\begin{equation}
\mathbf{S}_{2}(\lambda; \chi, \tau):=\left\{\begin{aligned}&\mathbf{R}(\lambda; \chi, \tau)\ee^{-\ii n\vartheta(\lambda; \chi, \tau)\sigma_3}\mathbf{Q}_{c}\ee^{\ii n\vartheta(\lambda; \chi, \tau)\sigma_3},\quad &\lambda\in D_0\cap\left(D_{2}^{+}\cup D_{2}^{-}\right)^{c},\\
&\mathbf{R}(\lambda; \chi, \tau),\quad &\text{otherwise},
\end{aligned}\right.
\end{equation}
where $D_{2}^{+}=R_{2\Sigma}^{+}\cup \Omega_{2}^{+}\cup R_{2\Gamma}^{+}, D_{2}^{-}=R_{2\Sigma}^{-}\cup \Omega_{2}^{-}\cup R_{2\Gamma}^{-},$
then the jump condition about $\mathbf{S}_{2}(\lambda; \chi, \tau)$ changes into the boundary of $D_{2}^{+}\cup D_{2}^{-}$. Moreover, set
\begin{equation}\label{eq:S-no}
\begin{aligned}
\mathbf{T}_{2}(\lambda; \chi, \tau):&=\mathbf{S}_{2}(\lambda; \chi, \tau)\ee^{-\ii n\vartheta(\lambda; \chi, \tau)\sigma_3}
\left(\mathbf{Q}_{R}^{[2]}\right)^{-1}\ee^{\ii n\vartheta(\lambda; \chi, \tau)\sigma_3}\ee^{\ii n g_{2}(\lambda)\sigma_3},\quad &\lambda\in L_{2\Gamma}^{+},\\
\mathbf{T}_{2}(\lambda; \chi, \tau):&=\mathbf{S}_{2}(\lambda; \chi, \tau)\mathbf{Q}_{L}^{[2]}\ee^{-\ii n\vartheta(\lambda; \chi, \tau)\sigma_3}
\mathbf{Q}_{C}^{[2]}\ee^{\ii n\vartheta(\lambda; \chi, \tau)\sigma_3}\ee^{\ii n g_{2}(\lambda)\sigma_3},\quad &\lambda\in R_{2\Gamma}^{+},\\
\mathbf{T}_{2}(\lambda; \chi, \tau):&=\mathbf{S}_{2}(\lambda; \chi, \tau)\mathbf{Q}_{L}^{[2]}\ee^{\ii n g_{2}(\lambda)\sigma_3},\quad &\lambda\in \Omega_{2}^{+},\\
\mathbf{T}_{2}(\lambda; \chi, \tau):&=\mathbf{S}_{2}(\lambda; \chi, \tau)\mathbf{Q}_{L}^{[2]}\ee^{-\ii n\vartheta(\lambda; \chi, \tau)\sigma_3}\mathbf{Q}_{L}^{[3]}\ee^{\ii n\vartheta(\lambda; \chi, \tau)\sigma_3}\ee^{\ii n g_{2}(\lambda)\sigma_3},\quad &\lambda\in R_{2\Sigma}^{+},\\
\mathbf{T}_{2}(\lambda; \chi, \tau):&=\mathbf{S}_{2}(\lambda; \chi, \tau)\ee^{-\ii n\vartheta(\lambda; \chi, \tau)\sigma_3}
\left(\mathbf{Q}_{R}^{[3]}\right)^{-1}\ee^{\ii n\vartheta(\lambda; \chi, \tau)\sigma_3}\ee^{\ii n g_{2}(\lambda)\sigma_3},\quad &\lambda\in L_{2\Sigma}^{+},\\
\mathbf{T}_{2}(\lambda; \chi, \tau):&=\mathbf{S}_{2}(\lambda; \chi, \tau)\ee^{-\ii n\vartheta(\lambda; \chi, \tau)\sigma_3}
\left(\mathbf{Q}_{R}^{[1]}\right)^{-1}\ee^{\ii n\vartheta(\lambda; \chi, \tau)\sigma_3}\ee^{\ii n g_{2}(\lambda)\sigma_3},\quad &\lambda\in L_{2\Gamma}^{-},\\
\mathbf{T}_{2}(\lambda; \chi, \tau):&=\mathbf{S}_{2}(\lambda; \chi, \tau)\mathbf{Q}_{L}^{[1]}\ee^{-\ii n\vartheta(\lambda; \chi, \tau)\sigma_3}
\mathbf{Q}_{C}^{[1]}
\ee^{\ii n\vartheta(\lambda; \chi, \tau)\sigma_3}\ee^{\ii n g_{2}(\lambda)\sigma_3},\quad &\lambda\in R_{2\Gamma}^{-},\\
\mathbf{T}_{2}(\lambda; \chi, \tau):&=\mathbf{S}_{2}(\lambda; \chi, \tau)\mathbf{Q}_{L}^{[1]}\ee^{\ii n g_{2}(\lambda)\sigma_3},\quad &\lambda\in \Omega_{2}^{-},\\
\mathbf{T}_{2}(\lambda; \chi, \tau):&=\mathbf{S}_{2}(\lambda; \chi, \tau)\mathbf{Q}_{L}^{[1]}\ee^{-\ii n\vartheta(\lambda; \chi, \tau)\sigma_3}
\mathbf{Q}_{L}^{[4]}\ee^{\ii n\vartheta(\lambda; \chi, \tau)\sigma_3}\ee^{\ii n g_{2}(\lambda)\sigma_3},\quad &\lambda\in R_{2\Sigma}^{-},\\
\mathbf{T}_{2}(\lambda; \chi, \tau):&=\mathbf{S}_{2}(\lambda; \chi, \tau)\ee^{-\ii n\vartheta(\lambda; \chi, \tau)\sigma_3}
\left(\mathbf{Q}_{R}^{[4]}\right)^{-1}\ee^{\ii n\vartheta(\lambda; \chi, \tau)\sigma_3}\ee^{\ii n g_{2}(\lambda)\sigma_3},\quad &\lambda\in L_{2\Sigma}^{-},
\end{aligned}
\end{equation}
in other regions, we set $\mathbf{T}_{2}(\lambda; \chi, \tau):=\mathbf{S}_{2}(\lambda; \chi, \tau)\ee^{\ii n g_{2}(\lambda)\sigma_3}$. Then the jump conditions about $\mathbf{T}_{2}(\lambda; \chi, \tau)$ can be changed into
\begin{equation}\label{eq:jump-genus-zero}
\begin{aligned}
\mathbf{T}_{2, +}(\lambda; \chi, \tau)&=\mathbf{T}_{2, -}(\lambda; \chi, \tau)\ee^{-\ii nh_2(\lambda)\sigma_3}
\mathbf{Q}_{R}^{[2]}\ee^{\ii nh_2(\lambda)\sigma_3},\quad &\lambda\in C_{L_{2\Gamma}}^{+},\\
\mathbf{T}_{2, +}(\lambda; \chi, \tau)&=\mathbf{T}_{2, -}(\lambda; \chi, \tau)\ee^{-\ii nh_2(\lambda)\sigma_3}\mathbf{Q}_{C}^{[2]}\ee^{\ii nh_2(\lambda)\sigma_3},\quad &\lambda\in C_{R_{2\Gamma}}^{+},\\
\mathbf{T}_{2, +}(\lambda; \chi, \tau)&=\mathbf{T}_{2, -}(\lambda; \chi, \tau)\ee^{-\ii nh_2(\lambda)\sigma_3}\mathbf{Q}_{L}^{[3]}\ee^{\ii nh_2(\lambda)\sigma_3},\quad &\lambda\in C_{R_{2\Sigma}}^{+},\\
\mathbf{T}_{2, +}(\lambda; \chi, \tau)&=\mathbf{T}_{2, -}(\lambda; \chi, \tau)\ee^{-\ii nh_2(\lambda)\sigma_3}\mathbf{Q}_{R}^{[3]}\ee^{\ii nh_2(\lambda)\sigma_3},\quad &\lambda\in C_{L_{2\Sigma}}^{+},\\
\mathbf{T}_{2, +}(\lambda; \chi, \tau)&=\mathbf{T}_{2, -}(\lambda; \chi, \tau)\ee^{-\ii nh_2(\lambda)\sigma_3}\mathbf{Q}_{R}^{[1]}\ee^{\ii nh_2(\lambda)\sigma_3},\quad &\lambda\in C_{L_{2\Gamma}}^{-},\\
\mathbf{T}_{2, +}(\lambda; \chi, \tau)&=\mathbf{T}_{2, -}(\lambda; \chi, \tau)\ee^{-\ii nh_2(\lambda)\sigma_3}\mathbf{Q}_{C}^{[1]}\ee^{\ii nh_2(\lambda)\sigma_3},\quad &\lambda\in C_{R_{2\Gamma}}^{-},\\
\mathbf{T}_{2, +}(\lambda; \chi, \tau)&=\mathbf{T}_{2, -}(\lambda; \chi, \tau)\ee^{-\ii nh_2(\lambda)\sigma_3}\mathbf{Q}_{L}^{[4]}\ee^{\ii nh_2(\lambda)\sigma_3},\quad &\lambda\in C_{R_{2\Sigma}}^{-},\\
\mathbf{T}_{2, +}(\lambda; \chi, \tau)&=\mathbf{T}_{2, -}(\lambda; \chi, \tau)\ee^{-\ii nh_2(\lambda)\sigma_3}\mathbf{Q}_{R}^{[4]}\ee^{\ii nh_2(\lambda)\sigma_3},\quad &\lambda\in C_{L_{2\Sigma}}^{-},\\
\mathbf{T}_{2, +}(\lambda; \chi, \tau)&=\mathbf{T}_{2, -}(\lambda; \chi, \tau)\begin{bmatrix}0&\ee^{-\ii n\left(2\vartheta(\lambda; \chi, \tau)+g_{2,+}(\lambda)+g_{2,-}(\lambda)\right)}\\
-\ee^{\ii n\left(2\vartheta(\lambda; \chi, \tau)+g_{2,+}(\lambda)+g_{2,-}(\lambda)\right)}&0
\end{bmatrix},&\lambda\in \Sigma_{g_2}^{+}\cup\Sigma_{g_2}^{-},\\
\mathbf{T}_{2, +}(\lambda; \chi, \tau)&=\mathbf{T}_{2, -}(\lambda; \chi, \tau)2^{\sigma_3},&\lambda\in I.
\end{aligned}
\end{equation}
When $n$ is large, the jump matrices on $C_{L_{2\Gamma}}^{\pm}, C_{R_{2\Gamma}}^{\pm}, C_{L_{2\Sigma}}^{\pm}$ and $C_{R_{2\Sigma}}^{\pm}$ will converge to the identity exponentially, and when $\lambda\in\Sigma_{g_2}^{\pm}\cup I$, the jump matrices become the constants due to the constraint condition in Eq.\eqref{eq:g2}, which is a solvable RHP except for some singularities. Next, we will construct its parametrix from the constant jump matrices.
\subsection{Parametrix construction}
From the constant jump matrices when $\lambda\in \Sigma_{g_2}^{\pm}$ and $\lambda\in I$, we can set the outer parametrix as
\begin{equation}\label{eq:T_2out}
\dot{\mathbf{T}}_{2}^{\rm out}(\lambda; \chi, \tau):=\mathbf{K}_{2}(\lambda; \chi, \tau)\left(\frac{\lambda-\alpha_2}{\lambda-\beta_2}\right)^{\ii p\sigma_3},\quad p=\frac{\log(2)}{2\pi},\quad \lambda\in\mathbb{C}\setminus\left(\Sigma_{g_2}^{\pm}\cup I\right),
\end{equation}
where the unknown matrix $\mathbf{K}_{2}(\lambda; \chi, \tau)$ should satisfy the following jump condition to match $\mathbf{T}_{2}(\lambda; \chi, \tau)$ on the arcs $\Sigma_{g_2}^{\pm}$,
\begin{equation}
\mathbf{K}_{2,+}(\lambda; \chi, \tau)=\mathbf{K}_{2,-}(\lambda; \chi, \tau)\left(\frac{\lambda-\alpha_2}{\lambda-\beta_2}\right)^{\ii p\sigma_3}\begin{bmatrix}0&\ee^{-\ii n\kappa_2}\\
-\ee^{\ii n\kappa_2}&0
\end{bmatrix}\left(\frac{\lambda-\alpha_2}{\lambda-\beta_2}\right)^{-\ii p\sigma_3},\quad \lambda\in\Sigma_{g_2}^{\pm}.
\end{equation}
It is easy to see that the jump matrix about $\mathbf{K}_{2}(\lambda; \chi, \tau)$ is a function with respect to $\lambda$, thus we should introduce another new matrix $\mathbf{G}_{2}(\lambda; \chi, \tau)$ to convert the jump matrix as a constant such that we can get its solution in an explicit formula.

Define
\begin{equation}
\mathbf{K}_{2}(\lambda; \chi, \tau)=\mathbf{G}_{2}(\lambda; \chi, \tau)\ee^{-(\ii k_2(\lambda)-\ii k_{2}(\infty))\sigma_3},
\end{equation}
where $k_2(\lambda)$ satisfies
\begin{equation}
k_{2,+}(\lambda)+k_{2,-}(\lambda)=2p\log\left(\frac{\lambda-\alpha_2}{\lambda-\beta_2}\right),\quad \lambda\in\Sigma_{g_2}^{\pm}.
\end{equation}
With the Plemelj formula and the generalized residue theorem, $k_2(\lambda)$ can be solved as
\begin{equation}
k_2(\lambda)=p\mathcal{R}_{2}(\lambda)\int_{\alpha_2}^{\beta_2}\frac{1}{\mathcal{R}_{2}(\xi)(\xi-\lambda)}d\xi+p\left(\frac{\lambda-\alpha_2}{\lambda-\beta_2}\right),
\end{equation}
especially, when $\lambda\to\infty$, we have
\begin{equation}\label{eq:mu-genus-zero}
 k_{2}(\infty)=\lim\limits_{\lambda\to\infty}k_2(\lambda)=-p\int_{\alpha_2}^{\beta_2}\frac{1}{\mathcal{R}_{2}(\xi)}d\xi.
\end{equation}
As a result, the jump matrix related by $\mathbf{G}_{2}(\lambda; \chi, \tau)$ is
\begin{equation}
\mathbf{G}_{2,+}(\lambda; \chi, \tau)=\mathbf{G}_{2,-}(\lambda; \chi, \tau)\begin{bmatrix}0&\ee^{-\ii n\kappa_2+2\ii  k_{2}(\infty)}\\
-\ee^{\ii n\kappa_2-2\ii k_{2}(\infty)}&0
\end{bmatrix}.
\end{equation}
Then we can obtain the solution about $\mathbf{G}_{2}(\lambda; \chi, \tau)$ by the linear algebra and Plemelj formula,
\begin{equation}
\mathbf{G}_{2}(\lambda; \chi, \tau)=\ee^{-\frac{\ii \pi}{4}\sigma_3}\ee^{\frac{2\ii k_{2}(\infty)-\ii n\kappa_2}{2}\sigma_3}\mathbf{Q}_c^{-1}\left(\frac{\lambda-a_2}{\lambda-a_2^*}\right)^{\frac{1}{4}\sigma_3}\mathbf{Q}_c\ee^{\frac{\ii n\kappa_2-2\ii k_{2}(\infty)}{2}\sigma_3}\ee^{\frac{\ii \pi}{4}\sigma_3},
\end{equation}
thus the outer parametrix $\dot{\mathbf{T}}_{2}^{\rm out}(\lambda; \chi, \tau)$ is given by \eqref{eq:T_2out}.
We hope that the parametrix $\dot{\mathbf{T}}_{2}^{\rm out}(\lambda; \chi, \tau)$ can be a good approximation to $\mathbf{T}_{2}(\lambda; \chi, \tau)$, but unfortunately, the outer parametrix $\dot{\mathbf{T}}_{2}^{\rm out}(\lambda; \chi, \tau)$ has singularities at $\lambda=\alpha_2, \lambda=\beta_2, \lambda=a_2, \lambda=a_2^*$. Thus we should consider the local analysis in the neighborhood of these four singularities. Firstly, we give the inner parametrix construction at these real singularities $\lambda=\alpha_2$ and $\lambda=\beta_2$. Similar to the inner parametrix construction in the algebraic-decay region, we set $D_{\alpha_2}(\delta)$ and $D_{\beta_2}(\delta)$ as the small disks centered at $\alpha_2$ and $\beta_2$ with the radius $\delta$. We define the conformal mappings $f_{\alpha_2}(\lambda)\equiv f_{\alpha_2}(\lambda; \chi, \tau)$ and $f_{\beta_2}(\lambda)\equiv f_{\beta_2}(\lambda; \chi, \tau)$ in the neighborhood of $\lambda=\alpha_2$ and $\lambda=\beta_2$ respectively,
\begin{equation}
f_{\alpha_2}^2(\lambda)=2\left(h_2(\alpha_2)-h_2(\lambda)\right), \quad f_{\beta_2}^2(\lambda)=2\left(h_2(\lambda)-h_2(\beta_2)\right).
\end{equation}
In this region, we still suppose $f'_{\alpha_2}(\alpha_2)=-\sqrt{-h_{2}''(\alpha_2)}<0, f'_{\beta_2}(\beta_2)=\sqrt{h_2''(\beta_2)}>0$. It should be noticed that at the point $\lambda=\alpha_2$, $h_2(\lambda)$ is discontinuous, for convenience, we set the value of $h_2(\alpha_2)$ in the right of $\Sigma_{g_2}^{\pm}$, that is $h_2(\alpha_2):=h_{2,-}(\alpha_2)$. Next, we introduce two matrices $\mathbf{U}_{\alpha_2}(\lambda; \chi, \tau)$ and $\mathbf{U}_{\beta_2}(\lambda; \chi, \tau)$ and new variables $\zeta_{\alpha_2}=n^{1/2}f_{\alpha_2}(\lambda)$ and $\zeta_{\beta_2}=n^{1/2}f_{\beta_2}(\lambda)$. In $D_{\beta_2}(\delta)$, define $\mathbf{U}_{\beta_2}(\lambda; \chi, \tau)$ as
\begin{equation}
\mathbf{U}_{\beta_2}(\lambda; \chi, \tau):=\mathbf{T}_{2}(\lambda; \chi, \tau)\ee^{-\ii nh_{2}(\beta_2)\sigma_3},\quad \lambda\in D_{\beta_2}(\delta),
\end{equation}
then the jump matrices satisfied by $\mathbf{U}_{\beta_2}(\lambda; \chi, \tau)$ change into five rays, which can be seen from Fig. \ref{fig:rays-para} with the variable $\zeta=\zeta_{\beta_2}$. Similarly, the inner parametrix $\dot{\mathbf{T}}_{2}^{\beta_2}(\lambda; \chi, \tau)$ can be defined as
\begin{multline}
\dot{\mathbf{T}}_{2}^{\beta_2}(\lambda; \chi, \tau)=\mathbf{K}_{2}(\lambda; \chi, \tau)n^{\ii p\sigma_3/2}\ee^{-\ii nh_2(\beta_2)\sigma_3}\mathbf{H}_{\beta_2}(\lambda; \chi, \tau)\mathbf{U}_{\beta_2}(\zeta_{\beta_2})\ee^{\ii nh_2(\beta_2)\sigma_3},\quad \lambda\in D_{\beta_2}(\delta),\\
\mathbf{H}_{\beta_2}(\lambda; \chi, \tau):=(\lambda-\alpha_2)^{\ii p\sigma_3}\left(\frac{f_{\beta_2}(\lambda; \chi, \tau)}{\lambda-\beta_2}\right)^{\ii p\sigma_3}.
\end{multline}
The inner parametrix construction at $\lambda=\alpha_2$ becomes a little more complex due to the cut at this point. Thus we can define $\mathbf{U}_{\alpha}(\lambda; \chi, \tau)$ as
\begin{equation}
\mathbf{U}_{\alpha_2}(\lambda; \chi, \tau):=\left\{\begin{aligned}&\mathbf{T}_{2}(\lambda; \chi, \tau)\ii^{\sigma_3}\ee^{-\ii nh_{2}(\alpha_2)\sigma_3}{\ii\sigma_2},&\quad\lambda\in D_{\alpha_2,-}(\delta),\\
&\mathbf{T}_{2}(\lambda; \chi, \tau)\ee^{-\ii n\frac{\kappa_2}{2}\sigma_3}(-\ii\sigma_2)\ee^{\ii n\frac{\kappa_2}{2}\sigma_3}\ii^{\sigma_3}\ee^{-\ii n h_{2}(\alpha_2)\sigma_3}\ii\sigma_2,&\quad\lambda\in D_{\alpha_2, +}(\delta),
\end{aligned}\right.
\end{equation}
where $D_{\alpha_2,-}(\delta)$ and $D_{\alpha_{2},+}(\delta)$ denote the right and the left hand side of $\Sigma_{g_2}^{+}\cup \Sigma_{g_2}^{-}$ in $D_{\alpha_2}(\delta)$.

Correspondingly, the inner parametrix in $D_{\alpha_2}(\delta)$ can be given as
\begin{equation}
\dot{\mathbf{T}}_{2}^{\alpha_2}(\lambda; \chi, \tau){:=}\left\{\begin{aligned}\mathbf{K}_{2}(\lambda; \chi, \tau)\mathbf{H}_{\alpha_2}(\lambda; \chi, \tau)n^{\ii p\sigma_3/2}\ee^{\ii nh_{2}(\alpha_2)\sigma_3}\ii^{-\sigma_3}&\mathbf{U}_{\alpha_2}(\zeta_{\alpha_2})(-\ii\sigma_2)\ii^{-\sigma_3}\ee^{\ii n h_{2}(\alpha_2)\sigma_3},\,\lambda\in D_{\alpha_2,-}(\delta),\\
\mathbf{K}_{2}(\lambda; \chi, \tau)\mathbf{H}_{\alpha_2}(\lambda; \chi, \tau)n^{\ii p\sigma_3/2}\ee^{\ii nh_{2}(\alpha_2)\sigma_3}\ii^{-\sigma_3}&\mathbf{U}_{\alpha_2}(\zeta_{\alpha_2})(-\ii\sigma_2)\ii^{-\sigma_3}\ee^{\ii n h_{2}(\alpha_2)\sigma_3}\\ &\cdot\ee^{-\ii n\frac{\kappa_2}{2}\sigma_3}(\ii\sigma_2)\ee^{\ii n\frac{\kappa_2}{2}\sigma_3},\,\lambda\in D_{\alpha_2,+}(\delta),\\
\mathbf{H}_{\alpha_2}(\lambda; \chi, \tau):&=\left(\frac{\alpha_2-\lambda}{f_{\alpha_2}(\lambda; \chi, \tau)}\right)^{\ii p\sigma_3}\left(\beta_2-\lambda\right)^{-\ii p\sigma_3}(\ii\sigma_2).
\end{aligned}\right.
\end{equation}
Next, we will consider the inner parametrix at the local points $a_2$ and $a_2^*$. Let $D_{a_2}(\delta)$ and $D_{a_2^*}(\delta)$ be the small disks centered at $a_2$ and $a_2^*$ with the radius $\delta$ respectively. Since $\lambda=a_2$ and $\lambda=a_2^*$ are the square roots of $h_2'(\lambda; \chi, \tau)$ and $h_2(a_2)=\frac{\kappa_2}{2}$, we can construct a univalent function $f_{a_2}\equiv f_{a_2}(\lambda; \chi, \tau)$ in $D_{a_2}(\delta)$ as
\begin{equation}
f_{a_2}(\lambda)^3=-\left(2h(\lambda)-\kappa_2\right)^2,
\end{equation}
and introduce a new variable $\zeta_{a_2}=n^{\frac{2}{3}}f_{a_2}(\lambda)$.
In the neighborhood of $D_{a_2}(\delta)$, we can set a local $\pmb{\mathcal{U}}_{a_2}(\lambda; \chi, \tau)$ as
\begin{equation}
\pmb{\mathcal{U}}_{a_2}(\lambda; \chi,\tau):=\left\{\begin{aligned}&\mathbf{T}_{2}(\lambda; \chi, \tau)\ee^{-\frac{\ii n\kappa_2}{2}\sigma_3},&\quad\lambda\in D_{a_2}(\delta)\setminus \left(D_{a_2}(\delta)\cap \Omega_{2}^{+}\right),\\
&\mathbf{T}_{2}(\lambda; \chi, \tau)\ee^{-\ii nh_2(\lambda)\sigma_3}\mathbf{Q}_{C}^{[2]}\ee^{\ii nh_2(\lambda)\sigma_3}\ee^{-\frac{\ii n\kappa_2}{2}\sigma_3},&\quad\lambda \in D_{a_2}(\delta)\cap \Omega_{2}.
\end{aligned}\right.
\end{equation}
Then $\pmb{\mathcal{U}}_{a_2}(\zeta_{a_2})$ can be given by the Airy function satisfying the following RHP.
\begin{rhp}\label{rhp-airy}
(Airy parametrix) Seek a $2\times 2$ matrix $\pmb{\mathcal{U}}(\zeta)$ with the following properties:
\begin{itemize}
\item {\bf Analyticity:} $\pmb{\mathcal{U}}(\zeta)$ is analytic except for the four rays in Fig. \ref{four:rays-para}.
\item {\bf Jump condition:} $\pmb{\mathcal{U}}_{+}(\zeta)=\pmb{\mathcal{U}}_{-}(\zeta)\pmb{\mathcal{V}}^{{\rm AI}}(\zeta)$, where $\pmb{\mathcal{V}}^{\rm AI}(\zeta)$ is given in Fig.\ref{four:rays-para}.
\item {\bf Normalization:} $\pmb{\mathcal{U}}(\zeta)\ee^{-\frac{\ii\pi}{4}\sigma_3}\mathbf{Q}_c^{-1}\zeta^{-\sigma_3/4}\to\mathbb{I}$ as $\zeta\to\infty$.
\end{itemize}
\end{rhp}
\begin{figure}[ht]
\centering
\includegraphics[width=0.45\textwidth]{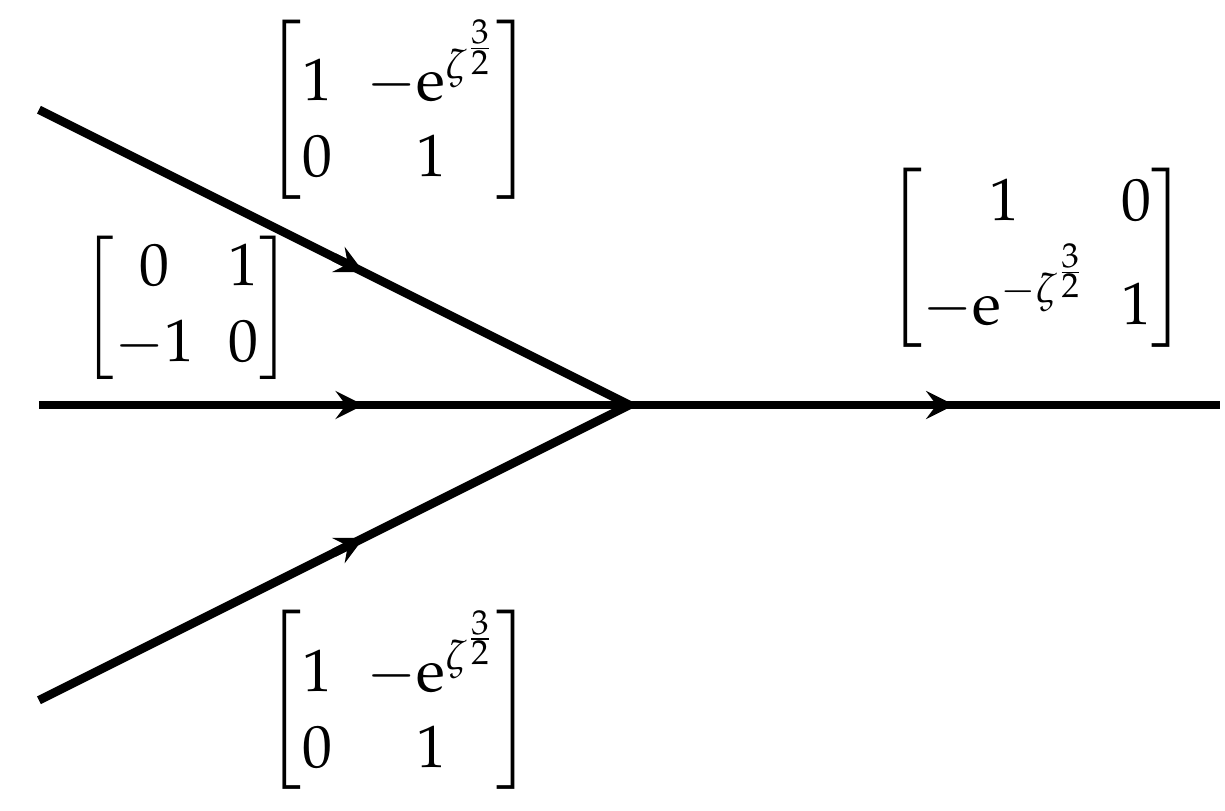}
\caption{The jump conditions about $\pmb{\mathcal{U}}=\pmb{\mathcal{U}}_{a_2}(\lambda; \chi, \tau)$ with the new conformal variables $\zeta_{a_2}$.}
\label{four:rays-para}
\end{figure}
In appendix \ref{app:Airy}, we give a detailed description about the RHP to Airy function, and we give an explicit solution for the RHP.

By using the Airy function, we can define the inner parametrix in $D_{a_2}(\delta)\cup D_{a_2^*}(\delta)$. Before doing it, we first define a new matrix $\pmb{\mathcal{H}}_{2,a_2}(\lambda; \chi, \tau)$  as
\begin{multline}\label{eq:h2a2}
\pmb{\mathcal{H}}_{2,a_2}(\lambda; \chi, \tau){:}{=}\dot{\mathbf{T}}^{\rm out}_{2}(\lambda; \chi, \tau)\ee^{\ii (w(\lambda){-} k_{2}(\infty))\sigma_3}\ee^{{-}\frac{\ii\pi}{4}\sigma_3}\ee^{\frac{2\ii k_{2}(\infty){-}\ii n\kappa_2}{2}\sigma_3}\mathbf{Q}_c^{-1}f_{a_2}(\lambda)^{{-}\frac{1}{4}\sigma_3},\\ w(\lambda){:}{=}p\mathcal{R}_{2}(\lambda)\int_{\alpha_2}^{\beta_2}\frac{1}{\mathcal{R}_{2}(\xi)(\xi{-}\lambda)}d\xi,
\end{multline}
it is clear that $\pmb{\mathcal{H}}_{2,a_2}(\lambda; \chi, \tau)$ is analytic in $D_{a_2}(\delta)$ from the definition of conformal map $\lambda\to f_{a_2}(\lambda)$ and the expression of $\dot{\mathbf{T}}_{2}^{\rm out}(\lambda; \chi, \tau)$. Then we can use the Airy function $\pmb{\mathcal{U}}_{a_2}(\lambda; \chi, \tau)$ and $\pmb{\mathcal{H}}_{2,a_2}(\lambda; \chi, \tau)$ to define the inner parametrix $\dot{\mathbf{T}}_{2}^{a_2}(\lambda; \chi, \tau)$,
\begin{equation}
\dot{\mathbf{T}}_{2}^{a_2}(\lambda; \chi, \tau):=\left\{\begin{aligned}&\pmb{\mathcal{H}}_{2,a_2}(\lambda; \chi, \tau)n^{-\frac{1}{6}\sigma_3}\pmb{\mathcal{U}}_{a_2}(n^{\frac{2}{3}}f_{a_2}(\lambda))\ee^{\frac{\ii n}{2}\kappa_2\sigma_3}\ee^{-\ii w(\lambda)\sigma_3},\quad\lambda\in D_{a_2}(\delta)\setminus(D_{a_2}(\delta)\cap\Omega_{2}),\\
&\pmb{\mathcal{H}}_{2,a_2}(\lambda; \chi, \tau)n^{-\frac{1}{6}\sigma_3}\pmb{\mathcal{U}}_{a_2}(n^{\frac{2}{3}}f_{a_2}(\lambda))\ee^{\frac{\ii n}{2}\kappa_2\sigma_3}\ee^{-\ii w(\lambda)\sigma_3}\ee^{-\ii nh_2(\lambda)\sigma_3}\left(\mathbf{Q}_{C}^{[2]}\right)^{-1}\ee^{\ii nh_2(\lambda)\sigma_3},\\
&\qquad\qquad\qquad\qquad\qquad\qquad\qquad\qquad\qquad\qquad\qquad\qquad\qquad\qquad\qquad\,\lambda\in D_{a_2}(\delta)\cap\Omega_{2}.\end{aligned}\right.
\end{equation}
In $D_{a_2^*}(\delta)$, we can define the local approximation depending on the symmetry $\mathbf{T}_{2}(\lambda; \chi, \tau)=\sigma_2\mathbf{T}(\lambda^*; \chi, \tau)^*\sigma_2$. Then the global parametrix constructing about $\mathbf{T}_{2}(\lambda; \chi, \tau)$ can be set as
\begin{equation}
\dot{\mathbf{T}}_{2}(\lambda; \chi, \tau):=\left\{\begin{aligned}&\dot{\mathbf{T}}_{2}^{\alpha_2}(\lambda; \chi, \tau), &\quad\lambda\in D_{\alpha_2}(\delta),\\
&\dot{\mathbf{T}}_{2}^{\beta_2}(\lambda; \chi, \tau),&\quad\lambda\in D_{\beta_2}(\delta),\\
&\dot{\mathbf{T}}_{2}^{a_2}(\lambda; \chi, \tau),&\quad\lambda\in D_{a_2}(\delta),\\
&\sigma_2\dot{\mathbf{T}}_{2}^{a_2}(\lambda^*; \chi, \tau)^*\sigma_2&\quad\lambda\in D_{a_2^*}(\delta),\\
&\dot{\mathbf{T}}_{2}^{\rm out}(\lambda; \chi, \tau),&\quad\lambda\in \mathbb{C}\setminus\left(\overline{D_{\alpha_2}(\delta)\cup D_{\beta_2}(\delta)\cup D_{a_2}(\delta)\cup D_{a_2^*}(\delta)}\cup \Sigma_{g_2}^{+}\cup \Sigma_{g_2}^{-}\cup I\right).\end{aligned}\right.
\end{equation}
Once the above parametrix is constructed, then we need to give the error analysis between $\mathbf{T}_{2}(\lambda; \chi, \tau)$ and $\dot{\mathbf{T}}_{2}(\lambda; \chi, \tau)$.
\subsection{Error analysis}
Set the error function between $\mathbf{T}_{2}(\lambda; \chi, \tau)$ and $\dot{\mathbf{T}}_{2}(\lambda; \chi, \tau)$ as
\begin{equation}
\mathcal{E}_{2}(\lambda; \chi, \tau):=\mathbf{T}_{2}(\lambda; \chi, \tau)\left(\dot{\mathbf{T}}_{2}(\lambda; \chi, \tau)\right)^{-1}.
\end{equation}
We denote $\mathbf{V}_{\mathcal{E}_{2}}(\lambda; \chi, \tau)$ as the jump matrix for $\mathcal{E}_{2}(\lambda; \chi, \tau)$ and $\Sigma_{\mathcal{E}_{2}}$ as the jump contours. From the definition of $\dot{\mathbf{T}}_{2}(\lambda; \chi, \tau)$, we can see that the jump arcs of $\mathbf{V}_{\mathcal{E}_{2}}(\lambda; \chi, \tau)$ are shown in Fig. \ref{fig:jump-E-no}.
\begin{figure}[ht]
\centering
\includegraphics[width=0.45\textwidth]{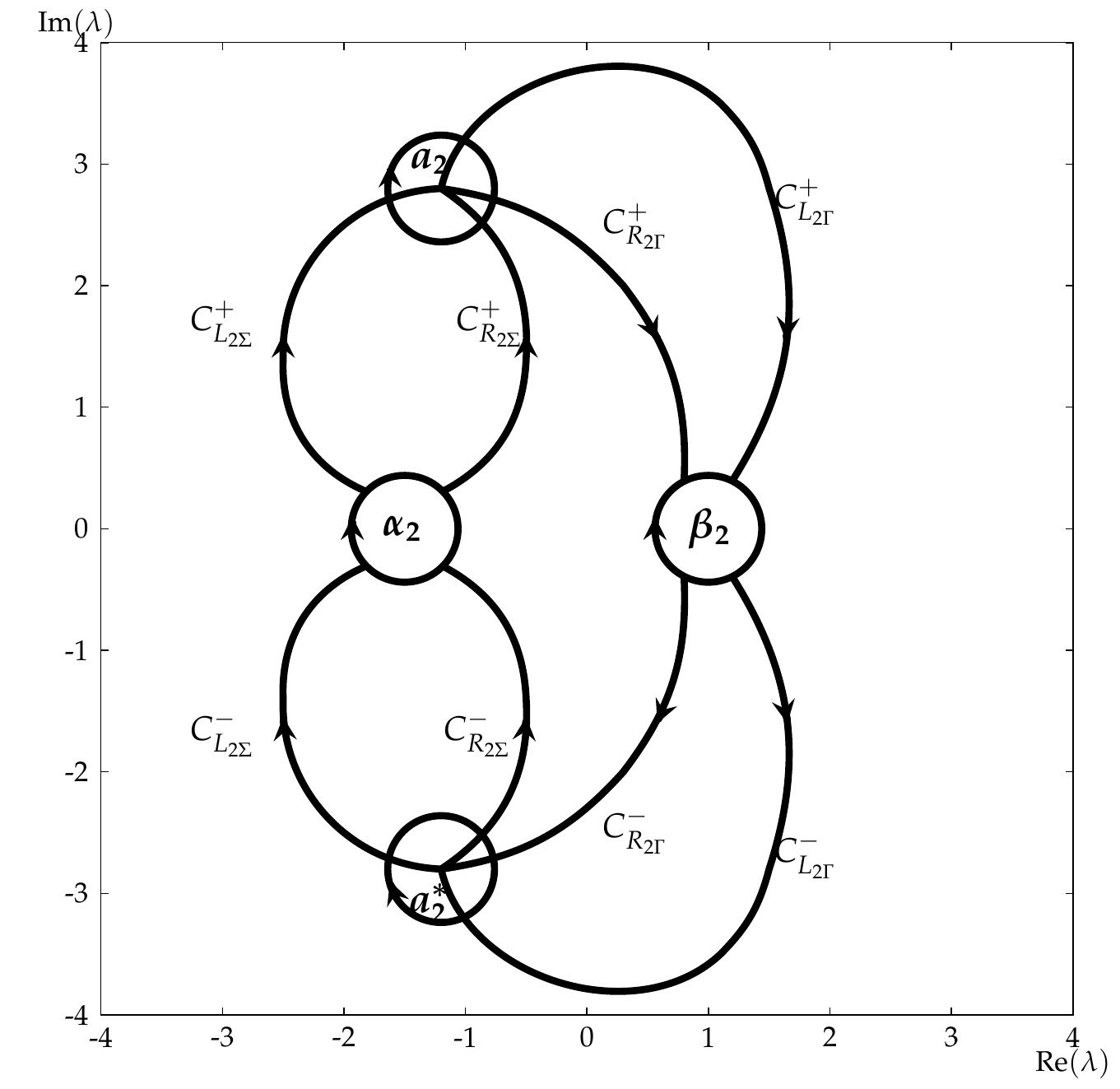}
\caption{The jump conditions about $\mathcal{E}_{2}(\lambda; \chi, \tau)$.}
\label{fig:jump-E-no}
\end{figure}

When $\lambda\in \left(C_{L_{2\Sigma}}^{+}\cup C_{R_{2\Sigma}}^{+}\cup C_{R_{2\Gamma}}^{+}\cup C_{L_{2\Gamma}}^{+}\right)\cap D_{a_2}(\delta)$, the jump conditions of $\mathbf{V}_{\mathcal{E}_{2}}(\lambda; \chi, \tau)$ can be given as
\begin{equation}
\begin{aligned}
&\mathbf{V}_{\mathcal{E}_{2}}(\lambda; \chi, \tau)=\dot{\mathbf{T}}^{a_2}_{2,-}(\lambda; \chi, \tau)\mathbf{T}_{2,-}(\lambda; \chi, \tau)^{-1}\mathbf{T}_{2,+}(\lambda; \chi, \tau)\left(\dot{\mathbf{T}}_{2,+}^{a_2}(\lambda;
\chi, \tau)\right)^{-1}
\\=&\left\{\begin{aligned}&\dot{\mathbf{T}}_{2,+}^{a_{2}}(\lambda; \chi, \tau)\begin{bmatrix}1&0\\
-\ee^{2\ii n h_{2}(\lambda)}\left(1-\ee^{-2\ii w(\lambda)}\right)&1
\end{bmatrix}\left(\dot{\mathbf{T}}_{2,+}^{a_2}(\lambda;
\chi, \tau)\right)^{-1},\quad \lambda\in C_{L_{2\Gamma}}^{+}\cap D_{a_2}(\delta),\\
&\dot{\mathbf{T}}_{2,+}^{a_{2}}(\lambda; \chi, \tau)\begin{bmatrix}1&\frac{1}{2}\ee^{-2\ii n h_{2}(\lambda)}\left(1-\ee^{2\ii w(\lambda)}\right)\\
0&1
\end{bmatrix}\left(\dot{\mathbf{T}}_{2,+}^{a_2}(\lambda;
\chi, \tau)\right)^{-1},\quad \lambda\in C_{R_{2\Gamma}}^{+}\cap D_{a_2}(\delta),\\
&\dot{\mathbf{T}}_{2,+}^{a_{2}}(\lambda; \chi, \tau)\begin{bmatrix}1&-\frac{1}{2}\ee^{-2\ii n h_{2}(\lambda)}\left(1-\ee^{2\ii w(\lambda)}\right)\\
0&1
\end{bmatrix}\left(\dot{\mathbf{T}}_{2,+}^{a_2}(\lambda;
\chi, \tau)\right)^{-1},\quad \lambda\in C_{R_{2\Sigma}}^{+}\cap D_{a_2}(\delta),\\
&\dot{\mathbf{T}}_{2,+}^{a_{2}}(\lambda; \chi, \tau)\begin{bmatrix}1&-\ee^{-2\ii n h_{2}(\lambda)}\left(1-\ee^{2\ii w(\lambda)}\right)\\
0&1
\end{bmatrix}\left(\dot{\mathbf{T}}_{2,+}^{a_2}(\lambda;
\chi, \tau)\right)^{-1},\quad \lambda\in C_{L_{2\Sigma}}^{+}\cap D_{a_2}(\delta).
\end{aligned}\right.
\end{aligned}
\end{equation}
From the definition of $w(\lambda)$ in Eq.\eqref{eq:h2a2}, when $\lambda=a_2, w(a_2)=0$, thus $\left|1-\ee^{\pm2\ii w(\lambda)}\right|$ is bounded in $D_{a_2}(\delta)$. When $n$ is large, $\mathbf{V}_{\mathcal{E}_{2}}(\lambda; \chi, \tau)$ will decay to the identity exponentially.

When $\lambda\in C_{L_{2\Sigma}}^{\pm}\cup C_{R_{2\Sigma}}^{\pm}\cup C_{R_{2\Gamma}}^{\pm}\cup C_{L_{2\Gamma}}^{\pm}\setminus D_{a_2}(\delta)$, from the jump conditions in Eq.\eqref{eq:jump-genus-zero}, $\mathbf{V}_{\mathcal{E}_{2}}(\lambda; \chi, \tau)$ will also decay to the identity exponentially. In the neighborhood of $D_{a_2^*}(\delta)$, we can get a similar result. Thus when $\lambda\in C_{L_{2\Sigma}}^{\pm}\cup C_{R_{2\Sigma}}^{\pm}\cup C_{R_{2\Gamma}}^{\pm}\cup C_{L_{2\Gamma}}^{\pm}$, there exists a positive constant $\mu_2$ such that $\|\mathbf{V}_{\mathcal{E}_{2}}(\lambda; \chi, \tau)-\mathbb{I}\|=\mathcal{O}(\ee^{-\mu_{2}n})$. But when $\lambda\in\partial D_{\alpha_2}(\delta)\cup \partial D_{\beta_2}(\delta)\cup\partial D_{a_2}(\delta)\cup\partial D_{a_2^*}(\delta)$, $\mathbf{V}_{\mathcal{E}_{2}}(\lambda; \chi, \tau)$ exhibits very differently, that is
\begin{equation}
\begin{aligned}
\mathbf{V}_{\mathcal{E}_{2}}(\lambda; \chi, \tau)&=\dot{\mathbf{T}}_{2}^{\alpha_2, \beta_2, a_2, a_2^*}(\lambda; \chi, \tau)\left(\dot{\mathbf{T}}_{2}^{\rm out}(\lambda; \chi, \tau)\right)^{-1},\quad\lambda\in\partial D_{\alpha_2,\beta_2, a_2, a_2^*}(\delta).
\end{aligned}
\end{equation}
If $\lambda\in\partial D_{\alpha_2}(\delta)$ and $\lambda\in\partial D_{\beta_2}(\delta)$, $\mathbf{V}_{\mathcal{E}_{2}}(\lambda; \chi, \tau) $ equals to
\begin{multline}\label{eq:E2alpha2}
\mathbf{V}_{\mathcal{E}_{2}}(\lambda; \chi, \tau)=\pmb{\mathcal{H}}_{2,\alpha_2}(\lambda; \chi, \tau)\mathbf{U}_{\alpha_2}(\zeta_{\alpha_2})\zeta_{\alpha_2}^{\ii p\sigma_3}\pmb{\mathcal{H}}_{2,\alpha_2}(\lambda; \chi, \tau)^{-1},\quad\lambda\in\partial D_{\alpha_2}(\delta),\\
\pmb{\mathcal{H}}_{2,\alpha_2}(\lambda; \chi, \tau):=\mathbf{K}_{2}(\lambda; \chi, \tau)\mathbf{H}_{\alpha_2}(\lambda; \chi, \tau)n^{\ii p\sigma_3/2}\ee^{\ii nh_{2}(\alpha_2)\sigma_3}\ii^{-\sigma_3},
\end{multline}
and
\begin{multline}\label{eq:E2beta2}
\mathbf{V}_{\mathcal{E}_{2}}(\lambda; \chi, \tau)=\pmb{\mathcal{H}}_{2,\beta_2}(\lambda; \chi, \tau)\mathbf{U}_{\beta_2}(\zeta_{\beta_2})\zeta_{\beta_2}^{\ii p\sigma_3}\pmb{\mathcal{H}}_{2,\beta_2}(\lambda; \chi, \tau)^{-1},\quad\lambda\in\partial D_{\beta_2}(\delta),\\
\pmb{\mathcal{H}}_{2,\beta_2}(\lambda; \chi, \tau):=\mathbf{K}_{2}(\lambda; \chi, \tau)\mathbf{H}_{\beta_2}(\lambda; \chi, \tau)n^{\ii p\sigma_3/2}\ee^{-\ii nh_{2}(\beta_2)\sigma_3}.
\end{multline}
When $\lambda\in\partial D_{a_2}(\delta)$ and $\lambda\in\partial D_{a_2^*}(\delta)$, the jump matrices $\mathbf{V}_{\mathcal{E}_{2}}(\lambda; \chi, \tau)$ become
\begin{equation}\label{eq:Ve2-2}
\begin{aligned}
\mathbf{V}_{\mathcal{E}_{2}}(\lambda; \chi, \tau)&=\pmb{\mathcal{H}}_{2,a_2}(\lambda; \chi, \tau)n^{-\frac{1}{6}\sigma_3}\pmb{\mathcal{U}}_{a_2}(\zeta_{a_2})\ee^{-\frac{\ii\pi}{4}\sigma_3}\mathbf{Q}_c^{-1}\zeta_{a_2}^{-\frac{1}{4}\sigma_3}n^{\frac{1}{6}\sigma_3}\pmb{\mathcal{H}}_{2,a_2}(\lambda; \chi, \tau)^{-1},\quad\lambda\in \partial D_{a_2}(\delta),\\
\mathbf{V}_{\mathcal{E}_{2}}(\lambda; \chi, \tau)&{=}\sigma_2\pmb{\mathcal{H}}_{2,a_2}(\lambda^*; \chi, \tau)^*n^{{-}\frac{1}{6}\sigma_3}\pmb{\mathcal{U}}_{a_2}(\zeta^*_{a_2})^*\ee^{\frac{\ii\pi}{4}\sigma_3}\mathbf{Q}_c^{-1}\zeta_{a_2}^{*,-\frac{1}{4}\sigma_3}n^{\frac{1}{6}\sigma_3}\pmb{\mathcal{H}}_{2,a_2}(\lambda^*; \chi, \tau)^{-1,*}\sigma_2,\,\lambda\in \partial D_{a_2^*}(\delta).
\end{aligned}
\end{equation}

Next, we will give the estimation to $\mathbf{V}_{\mathcal{E}_{2}}(\lambda; \chi, \tau)$. With the asymptotics of parabolic cylinder function in section \ref{sec:al-decay}, when $n$ is large, we have
\begin{equation}\label{eq:error-2-1}
\begin{aligned}
\|\mathbf{V}_{\mathcal{E}_2}(\lambda; \chi, \tau)-\mathbb{I}\|=\mathcal{O}(n^{-1/2}),\quad\lambda\in \partial D_{\alpha_2, \beta_2}(\delta).
\end{aligned}
\end{equation}
While for $\lambda\in\partial D_{a_2}(\delta)\cup \partial D_{a_2^*}(\delta)$, the errors of $\|\mathbf{V}_{\mathcal{E}_2}(\lambda; \chi, \tau)-\mathbb{I}\|$ are different, we will give them later. For $\lambda\in\partial D_{a_2}(\delta)$, substitute $\pmb{\mathcal{H}}_{2,a_2}$ in Eq.\eqref{eq:h2a2} and the asymptotic expression $\pmb{\mathcal{U}}_{a_2}(\zeta_{a_2})$ in Eq.\eqref{eq:nor-W} into Eq.\eqref{eq:Ve2-2}, we have
\begin{equation}
\mathbf{V}_{\mathcal{E}_2}(\lambda; \chi, \tau)=\begin{bmatrix}1-\frac{b_{12}\pmb{\mathcal{H}}_{2,a_2,11}\pmb{\mathcal{H}}_{2,a_2,21}}{n^{1/3}\zeta_{a_2}}+\mathcal{O}(n^{-1})&\frac{b_{12}\pmb{\mathcal{H}}_{2,a_2,11}^2}{n^{1/3}\zeta_{a_2}}+\mathcal{O}(n^{-1})\\
-\frac{b_{12}\pmb{\mathcal{H}}_{2,a_2,21}^2}{n^{1/3}\zeta_{a_2}}+\mathcal{O}(n^{-1})&1+\frac{b_{12}\pmb{\mathcal{H}}_{2,a_2,11}\pmb{\mathcal{H}}_{2,a_2,21}}{n^{1/3}\zeta_{a_2}}+\mathcal{O}(n^{-1})
\end{bmatrix}.
\end{equation}
And when $\lambda\in \partial D_{a_2^*}(\delta)$, we can get a similar result by the symmetry condition, thus we have
\begin{equation}\label{eq:error-2-2}
\begin{aligned}
\|\mathbf{V}_{\mathcal{E}_2}(\lambda; \chi, \tau)-\mathbb{I}\|=\mathcal{O}(n^{-1}),\quad\lambda\in \partial D_{a_2, a_2^*}(\delta).
\end{aligned}
\end{equation}
Under this case, the solution $q^{[n]}(n\chi, n\tau)$ can be given by
\begin{equation}\label{eq:qn-no-1}
\begin{aligned}
q^{[n]}(n\chi, n\tau)&=2\ii\lim\limits_{\lambda\to\infty}\lambda\mathbf{T}_{2}(\lambda; \chi, \tau)_{12}\\
&=2\ii\lim\limits_{\lambda\to\infty}\lambda\left(\mathcal{E}_{2}(\lambda; \chi, \tau)\dot{\mathbf{T}}^{\rm out}_{2}(\lambda; \chi, \tau)\right)_{12}\\
&=2\ii\lim\limits_{\lambda\to\infty}\lambda\left(\mathcal{E}_{2,11}(\lambda; \chi, \tau)\dot{\mathbf{T}}_{2,12}^{\rm out}(\lambda; \chi, \tau)+\mathcal{E}_{2,12}(\lambda; \chi, \tau)\dot{\mathbf{T}}_{2,22}^{\rm out}(\lambda; \chi, \tau)\right).
\end{aligned}
\end{equation}
When $\lambda\to\infty$, both $\dot{\mathbf{T}}^{\rm out}_{2}(\lambda; \chi, \tau)$ and $\mathcal{E}_{2}(\lambda; \chi, \tau)$ approach to the identity, thus Eq.\eqref{eq:qn-no-1} changes into
\begin{equation}
q^{[n]}(n\chi, n\tau)=2\ii\lim\limits_{\lambda\to\infty}\lambda\left(\dot{\mathbf{T}}_{2,12}^{\rm out}(\lambda; \chi, \tau)+\mathcal{E}_{2,12}(\lambda; \chi, \tau)\right).
\end{equation}

Next, we give a detailed calculation to $\mathcal{E}_{2}(\lambda; \chi, \tau)$. To study it, we rewrite the jump relation $\mathcal{E}_{2,+}(\lambda; \chi, \tau)=\mathcal{E}_{2,-}(\lambda; \chi, \tau)\mathbf{V}_{\mathcal{E}_{2}}(\lambda; \chi, \tau)$ to another equivalent formula,
\begin{equation}
\mathcal{E}_{2,+}(\lambda; \chi, \tau)-\mathcal{E}_{2,-}(\lambda; \chi, \tau)=\mathcal{E}_{2,-}(\lambda; \chi, \tau)\left(\mathbf{V}_{\mathcal{E}_{2}}(\lambda; \chi, \tau)-\mathbb{I}\right).
\end{equation}
From the error estimation in Eq.\eqref{eq:error-2-1} and Eq.\eqref{eq:error-2-2}, we know that the first order error term comes from the contour $\partial D_{\alpha_2}(\delta)$ and $\partial D_{\beta_2}(\delta)$, thus we only give a detailed calculation about these terms, the higher order error terms are ignored. When $\lambda\in\partial D_{\alpha_2}(\delta)\cup\partial D_{\beta_2}(\delta)$, with the Plemelj formula, we can get the solution of $\mathcal{E}_{2}(\lambda; \chi, \tau)$ as
\begin{equation}\label{eq:E2}
\mathcal{E}_{2}(\lambda; \chi, \tau)=\mathbb{I}+\frac{1}{2\pi\ii}\int_{\partial D_{\alpha_2}(\delta)\cup \partial D_{\beta_2}(\delta)}\frac{\mathcal{E}_{2,-}(\xi; \chi, \tau)(\mathbf{V}_{\mathcal{E}_{2}}(\xi; \chi, \tau)-\mathbb{I})}{\xi-\lambda}d\xi+\mathcal{O}\left(n^{-1}\right).
\end{equation}
Expanding Eq.\eqref{eq:E2} as $\lambda\to \infty$, we get
\begin{equation}
\mathcal{E}_{2}(\lambda; \chi, \tau)=\mathbb{I}-\frac{1}{2\pi\ii}\sum_{j=1}^{\infty}\lambda^{-j}\int_{\partial D_{\alpha_2}(\delta)\cup \partial D_{\beta_2}(\delta)}\mathcal{E}_{2,-}(\xi; \chi, \tau)(\mathbf{V}_{\mathcal{E}_{2}}(\xi; \chi, \tau)-\mathbb{I})\xi^{j-1}d\xi+\mathcal{O}\left(n^{-1}\right),\quad |\lambda|\to\infty.
\end{equation}
With a simple calculation, we further have
\begin{multline}
\lim \limits_{\lambda\to\infty}\lambda\mathcal{E}_{2,12}(\lambda; \chi, \tau)=-\frac{1}{2\pi\ii}\Big[\int_{\partial D_{\alpha_2}(\delta)\cup\partial D_{\beta_2}(\delta)}\mathcal{E}_{2,11,-}(\xi; \chi, \tau)\mathbf{V}_{\mathcal{E}_{2},12}(\xi; \chi, \tau)d\xi\\+\int_{\partial D_{\alpha_2}(\delta)\cup\partial D_{\beta_2}(\delta)}\mathcal{E}_{2,12,-}(\xi; \chi, \tau)\left(\mathbf{V}_{\mathcal{E}_{2},22}(\xi; \chi, \tau)-1\right)d\xi\Big]+\mathcal{O}\left(n^{-1}\right).
\end{multline}
Based on the definition of $\mathbf{V}_{\mathcal{E}_2}(\lambda; \chi, \tau)$ in Eq.\eqref{eq:E2alpha2} and Eq.\eqref{eq:E2beta2}, we can get the asymptotic expression of $q^{[n]}(n\chi, n\tau)$ as
\begin{equation}
q^{[n]}(n\chi, n\tau)=2\ii\lim\limits_{\lambda\to\infty}\lambda \dot{\mathbf{T}}_{2,12}^{\rm out}(\lambda; \chi, \tau)-\frac{1}{\pi}\int_{\partial D_{\alpha_2}(\delta)\cup\partial D_{\beta_2}(\delta)}\mathbf{V}_{\mathcal{E}_{2},12}(\xi; \chi, \tau)d\xi+\mathcal{O}\left(n^{-1}\right).
\end{equation}
Again, from the formula in Eq.\eqref{eq:E2alpha2} and Eq.\eqref{eq:E2beta2}, we have
\begin{multline}
\mathbf{V}_{\mathcal{E}_2,12}(\lambda; \chi, \tau)=-\ii \frac{n^{\ii p}\left(H_{\beta_2, 11}(\lambda)\right)^2\left(K_{2,11}(\lambda; \chi, \tau)\right)^2\ee^{-2\ii n h_2(\beta_2)}\alpha}{2n^{1/2}f_{\beta_2}(\lambda)}\\-\ii \frac{n^{-\ii p}\left(H_{\beta_2, 22}(\lambda)\right)^2\left(K_{2,12}(\lambda; \chi, \tau)\right)^2\ee^{2\ii n h_2(\beta_2)}\beta}{2n^{1/2}f_{\beta_2}(\lambda)}+\mathcal{O}(n^{-1}),\quad \lambda\in D_{\beta_2}(\delta),
\end{multline}
and
\begin{multline}
\mathbf{V}_{\mathcal{E}_2,12}(\lambda; \chi, \tau)=\ii \frac{n^{\ii p}\left(H_{\alpha_2,12}(\lambda)\right)^{-2}\left(K_{2,12, -}(\lambda; \chi, \tau)\right)^2\ee^{2\ii n h_{2,-}(\alpha_2)}\alpha}{2n^{1/2}f_{\alpha_2}(\lambda)}\\+\ii \frac{n^{-\ii p}\left(H_{\alpha_2, 12}(\lambda)\right)^2\left(K_{2,11,-}(\lambda; \chi, \tau)\right)^2\ee^{-2\ii n h_{2,-}(\alpha_2)}\beta}{2n^{1/2}f_{\alpha_2}(\lambda)}+\mathcal{O}(n^{-1}),\quad \lambda\in D_{\alpha_2}(\delta),
\end{multline}
where $\alpha$ and $\beta$ are given in the equation \eqref{eq:alphabeta}.
With the aid of residue theorem, we further have
\begin{multline}
-\frac{1}{\pi}\int_{\partial D_{\alpha_2}(\delta)\cup\partial D_{\beta_2}(\delta)}\mathbf{V}_{\mathcal{E}_{2},12}(\xi; \chi, \tau)d\xi\\=\frac{n^{\ii p}\left(H_{\alpha_2,12}(\alpha_2)\right)^{-2}\left(K_{2,12, -}(\alpha_2; \chi, \tau)\right)^2\ee^{2\ii n h_{2,-}(\alpha_2)}\alpha}{n^{1/2}\sqrt{-h''_{2, -}(\alpha_2)}}+ \frac{n^{-\ii p}\left(H_{\alpha_2,12}(\alpha_2)\right)^2\left(K_{2,-,11}(\alpha_2; \chi, \tau)\right)^2\ee^{-2\ii n h_{2,-}(\alpha_2)}\beta}{n^{1/2}\sqrt{-h''_{2,-}(\alpha_2)}}\\+\frac{n^{\ii p}\left(H_{\beta_2,11}(\beta_2)\right)^2\left(K_{2,11}(\beta_2; \chi, \tau)\right)^2\ee^{-2\ii n h_{2}(\beta_2)}\alpha}{n^{1/2}\sqrt{h''_{2}(\beta_2)}}+ \frac{n^{-\ii p}\left(H_{\beta_2,22}(\beta_2)\right)^2\left(K_{2,12}(\beta_2; \chi, \tau)\right)^2\ee^{2\ii n h_{2}(\beta_2)}\beta}{n^{1/2}\sqrt{h_{2}''(\beta_2)}}+\mathcal{O}(n^{-1}).
\end{multline}
Substitute $\mathbf{H}_{\alpha_2, \beta_2}(\lambda)$, $\mathbf{K}_{2}(\lambda; \chi,\tau)$ into the above equation, then the asymptotic expression in the genus-zero region can be written as Eq.\eqref{eq:q-gen-zero}, where
\begin{equation}\label{eq:para-genus-zero}
\begin{aligned}
\phi_{\alpha_2}&{=}\frac{\pi}{4}{+}\frac{\log(2)^2}{2\pi}{-}\arg\left(\Gamma\left(\ii\frac{\log(2)}{2\pi}\right)\right){+}2k_{2,-}(\alpha_2){+}2nh_{2,-}(\alpha_2){+}p\log\left({-}nh_{2,-}''(\alpha_2)(\beta_2-\alpha_2)^2\right){-}n\kappa_2,\\
\phi_{\beta_2}&{=}\frac{\pi}{4}{+}\frac{\log(2)^2}{2\pi}{-}\arg\left(\Gamma\left(\ii\frac{\log(2)}{2\pi}\right)\right){-}2k_2(\beta_2){-}2nh_{2}(\beta_2){+}p\log\left(nh_{2}''(\beta_2)(\beta_2-\alpha_2)^2\right){+}n\kappa_2,\\
m_{+}^{\alpha_2}&{=}\frac{1}{2}{+}\frac{1}{4}\left(\sqrt{\frac{\alpha_2{-}a_2}{\alpha_2{-}a^*_2}}{+}\left(\sqrt{\frac{\alpha_2{-}a_2}{\alpha_2{-}a^*_2}}\right)^{-1}\right),\quad
m_{-}^{\alpha_2}{=}\frac{1}{2}{-}\frac{1}{4}\left(\sqrt{\frac{\alpha_2{-}a_2}{\alpha_2{-}a^*_2}}{+}\left(\sqrt{\frac{\alpha_2{-}a_2}{\alpha_2{-}a^*_2}}\right)^{-1}\right),\\
m_{+}^{\beta_2}&{=}\frac{1}{2}{+}\frac{1}{4}\left(\sqrt{\frac{\beta_2{-}a_2}{\beta_2{-}a^*_2}}{+}\left(\sqrt{\frac{\beta_2{-}a_2}{\beta_2{-}a^*_2}}\right)^{-1}\right),\quad
m_{-}^{\beta_2}{=}\frac{1}{2}{-}\frac{1}{4}\left(\sqrt{\frac{\beta_2{-}a_2}{\beta_2{-}a^*_2}}{+}\left(\sqrt{\frac{\beta_2{-}a_2}{\beta_2{-}a^*_2}}\right)^{-1}\right),
\end{aligned}
\end{equation}
which completes the proof of theorem \ref{theo:genus-zero}.
In the next section, we will give the asymptotics in the genus-one region.
\section{The genus-one region}
\label{sec:genus-one}
Under the special choice $\lambda_1=\ii$ and $\lambda_2=k\ii$, we know that the velocity of high-order breather is zero. If we choose the spectral parameters $\lambda_1$ and $\lambda_2$ having distinct real parts, we could image that when $n$ is large, the large order asymptotics will contain two distinct genus-one regions by the result of \cite{Bilman-JDE-2021}, which is mainly determined by single high-order solitons with parameters $\lambda_1$ and $\lambda_2$ respectively. To study the large asymptotics of high-order breathers in the genus-one region, we first introduce a scalar RHP about $g_{3}(\lambda)\equiv g_{3}(\lambda; \chi, \tau)$ function.
\begin{rhp}
($g_{3}(\lambda)$-function in the genus-one region) For fixed $\left(\chi, \tau\right)$ in the genus-one region. There exist unique contours $\Sigma_{g_3}^{\pm}, \Gamma_{g_3}$, unique integration constants $\kappa_3$ and $d_3$ and unique $g_3(\lambda)$ function satisfying the following conditions.
\begin{itemize}
\item{\bf Analyticity:} $g_{3}(\lambda)$ is analytic in the complex plane except for $\Sigma_{g_3}^{\pm}, \Gamma_{g_3}$, where these three contours are to be determined.
\item{\bf Jump condition:} $g_{3}(\lambda)$ takes the continuous boundary conditions in these three contours and $g_{3,\pm}(\lambda)$ satisfies the following jump conditions,
    \begin{equation}
    \begin{aligned}
    &g_{3,+}(\lambda)+g_{3,-}(\lambda)+2\vartheta(\lambda; \chi, \tau)=\kappa_3,&\quad\lambda\in \Sigma_{g_3}^{+},\\
    &g_{3,+}(\lambda)+g_{3,-}(\lambda)+2\vartheta(\lambda; \chi, \tau)=\kappa_3,&\quad\lambda\in \Sigma_{g_3}^{-},\\
    &g_{3,+}(\lambda)-g_{3,-}(\lambda)=d_3,&\quad\lambda\in \Gamma_{g_3}.
    \end{aligned}
    \end{equation}
\item{\bf Normalization:} As $\lambda\to\infty$, $g_{3}(\lambda)$ satisfies the following normalization condition,
\begin{equation}
g_{3}(\lambda)\to\mathcal{O}(\lambda^{-1}).
\end{equation}
\item {\bf Symmetry condition:} $g_{3}(\lambda)$ has the symmetry condition $g_{3}(\lambda)=g_{3}(\lambda^*)^*$.
\end{itemize}
\end{rhp}
From the symmetry condition, we can find the fact that $d_3$ is a real number.

To solve this problem, we first take the derivative of $g_3(\lambda)$ with respect to $\lambda$ to eliminate the integration constants $\kappa_3$ and $d_3$. Then we will introduce an algebraic curve $\mathcal{R}_{3}(\lambda)$ with the genus-one, the corresponding cut is set as $\Sigma_{g_3}^+=[a_3, b_3], \Sigma_{g_3}^-=[a_3^*, b_3^*]$,
\begin{equation}
\mathcal{R}_{3}(\lambda)=\sqrt{(\lambda-a_3)(\lambda-a_3^*)(\lambda-b_3)(\lambda-b_3^*)}.
\end{equation}
By using $\mathcal{R}_{3}(\lambda)$, we can derive a jump relation about $\frac{g_3'(\lambda)}{\mathcal{R}_{3}(\lambda)}$,
\begin{equation}
\left(\frac{g_3'(\lambda)}{\mathcal{R}_{3}(\lambda)}\right)_{+}-\left(\frac{g_3'(\lambda)}{\mathcal{R}_{3}(\lambda)}\right)_{-}=\frac{-2\chi-4\lambda\tau-\frac{\ii}{\lambda-\ii}+\frac{\ii}{\lambda+\ii}-\frac{\ii}{\lambda-k\ii}+\frac{\ii}{\lambda+k\ii}}{\mathcal{R}_{3,+}(\lambda)},
\end{equation}
and as $\lambda\to\infty$, we have $$g_3'(\lambda)\to\mathcal{O}(\lambda^{-2}).$$
With the Plemelj formula, $g_3'(\lambda)$ can be given as
\begin{equation}
g_3'(\lambda)=\frac{\mathcal{R}_{3}(\lambda)}{2\pi\ii}\int_{\Sigma_{g_3}^+\cup \Sigma_{g_3}^-}\frac{-2\chi-4\xi\tau-\frac{\ii}{\xi-\ii}+\frac{\ii}{\xi+\ii}-\frac{\ii}{\xi-k\ii}+\frac{\ii}{\xi+k\ii}}{\mathcal{R}_{3,+}(\xi)(\xi-\lambda)}d\xi.
\end{equation}
By using a generalized residue theorem, we can rewrite $g_3'(\lambda)$ as an explicit formula,
\begin{equation}
g_3'(\lambda)=\frac{\ii}{2}\mathcal{R}_{3}(\lambda)\left(\frac{1}{\mathcal{R}_{3}(\ii)(\lambda-\ii)}-\frac{1}{\mathcal{R}(-\ii)(\lambda+\ii)}+\frac{1}{\mathcal{R}_{3}(k\ii)(\lambda-k\ii)}-\frac{1}{\mathcal{R}_{3}(-k\ii)(\lambda+k\ii)}\right)-\vartheta'(\lambda; \chi, \tau).
\end{equation}
From the normalization condition when $\lambda\to\infty$, we can get three equations about these unknown parameters,
\begin{equation}\label{eq:unknown-genus-one}
\begin{aligned}
\mathcal{O}(\lambda):&\frac{1}{2}\left(\frac{\ii}{\mathcal{R}_{3}(\ii)}-\frac{\ii}{\mathcal{R}(-\ii)}+\frac{\ii}{\mathcal{R}_{3}(k\ii)}-\frac{\ii}{\mathcal{R}_{3}(-k\ii)}\right)-2\tau=0,\\
\mathcal{O}(\lambda^0):&\frac{1}{2}\left(\frac{1}{\mathcal{R}_{3}(\ii)}+\frac{1}{\mathcal{R}_{3}(-\ii)}+\frac{k}{\mathcal{R}_{3}(k\ii)}+\frac{k}{\mathcal{R}_{3}(-k\ii)}\right)+\chi+2\left(\Re(a_3)+\Re(b_3)\right)\tau=0,\\
\mathcal{O}(\lambda^{-1}):&\frac{1}{2}\left(\frac{\ii}{\mathcal{R}_{3}(\ii)}-\frac{\ii}{\mathcal{R}_{3}(-\ii)}+\frac{\ii k^2}{\mathcal{R}_3(k\ii)}-\frac{\ii k^2}{\mathcal{R}_{3}(-k\ii)}+2\left(\Re(a_3)+\Re(b_3)\right)\chi\right)\\
\quad\quad\quad&+\left(2\Re(a_3)^2+2\Re(a_3)\Re(b_3)+2\Re(b_3)^2-\Im(a_3)^2-\Im(b_3)^2\right)\tau=0.
\end{aligned}
\end{equation}
The zeros of $\mathcal{R}_{3}(\lambda)$ contain four unknown parameters---the real part and the imaginary part of $a_3$ and $b_3$. Up to now, from Eq.\eqref{eq:unknown-genus-one}, we only have three equations about these parameters, thus we impose that the integration constant $\kappa_3$ is a real number, which can be regarded as the fourth condition to determine these parameters. By choosing one fixed $\chi, \tau$ and $k$ in this genus-one region, we give the contour plot about the imaginary part of $h_3(\lambda)\equiv h_3(\lambda; \chi, \tau):=g_3(\lambda)+\vartheta(\lambda; \chi, \tau)$ in Fig. \ref{fig:genus-one}. It should be emphasized that in this case we choose $k=2$ and there exist two closed contours involving the singularities $\lambda=\pm 2\ii$, whose analysis is similar to the asymptotics in the exponential-decay region, which can be obtained from Eq.\eqref{eq:exp-decay}. For this analysis, it will provide an additional exponential-decay term to the leading order term, but we omit it in this case since it does not affect the leading order term.
\begin{figure}[ht]
\centering
\includegraphics[width=0.45\textwidth]{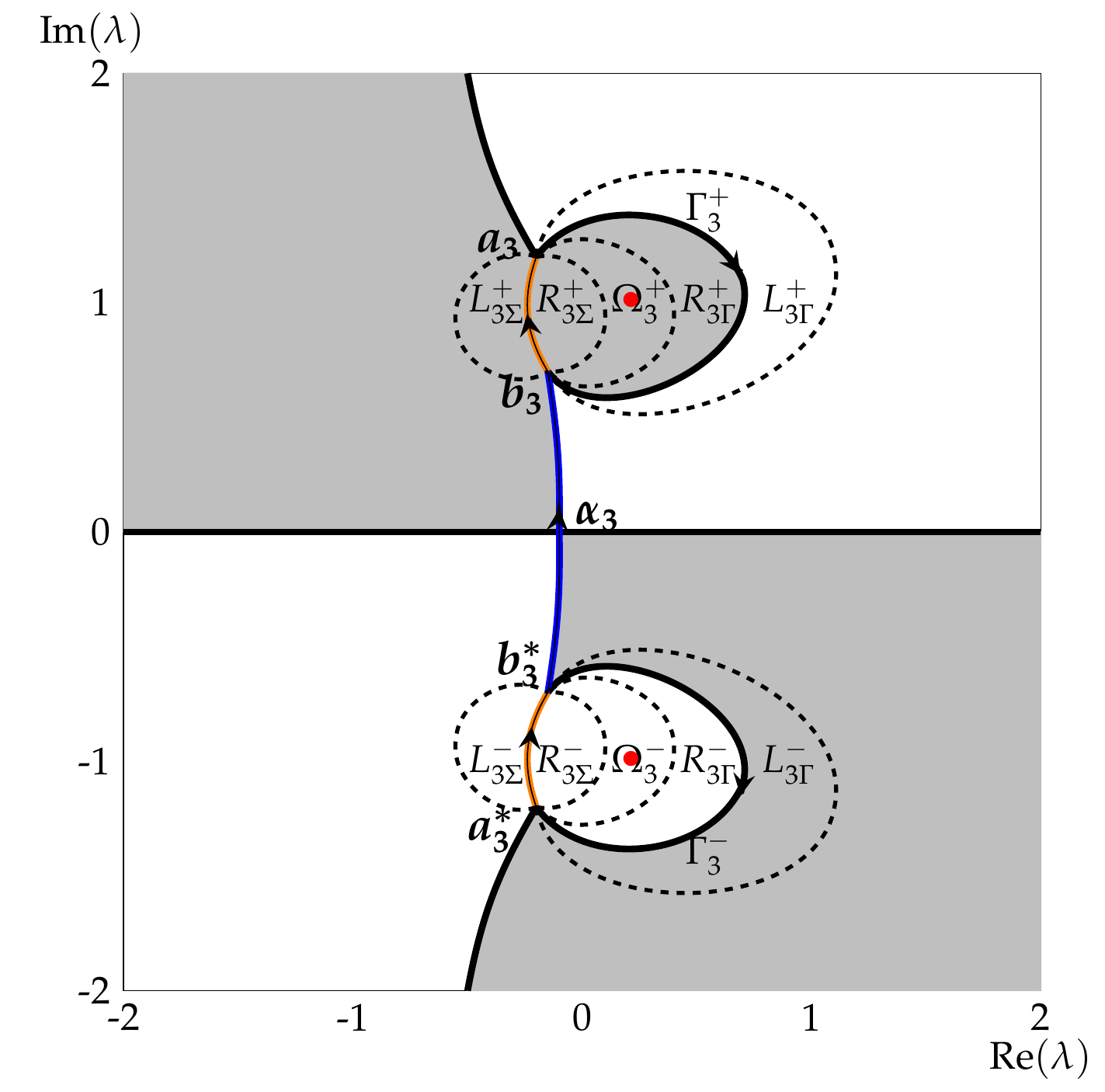}
\centering
\includegraphics[width=0.45\textwidth]{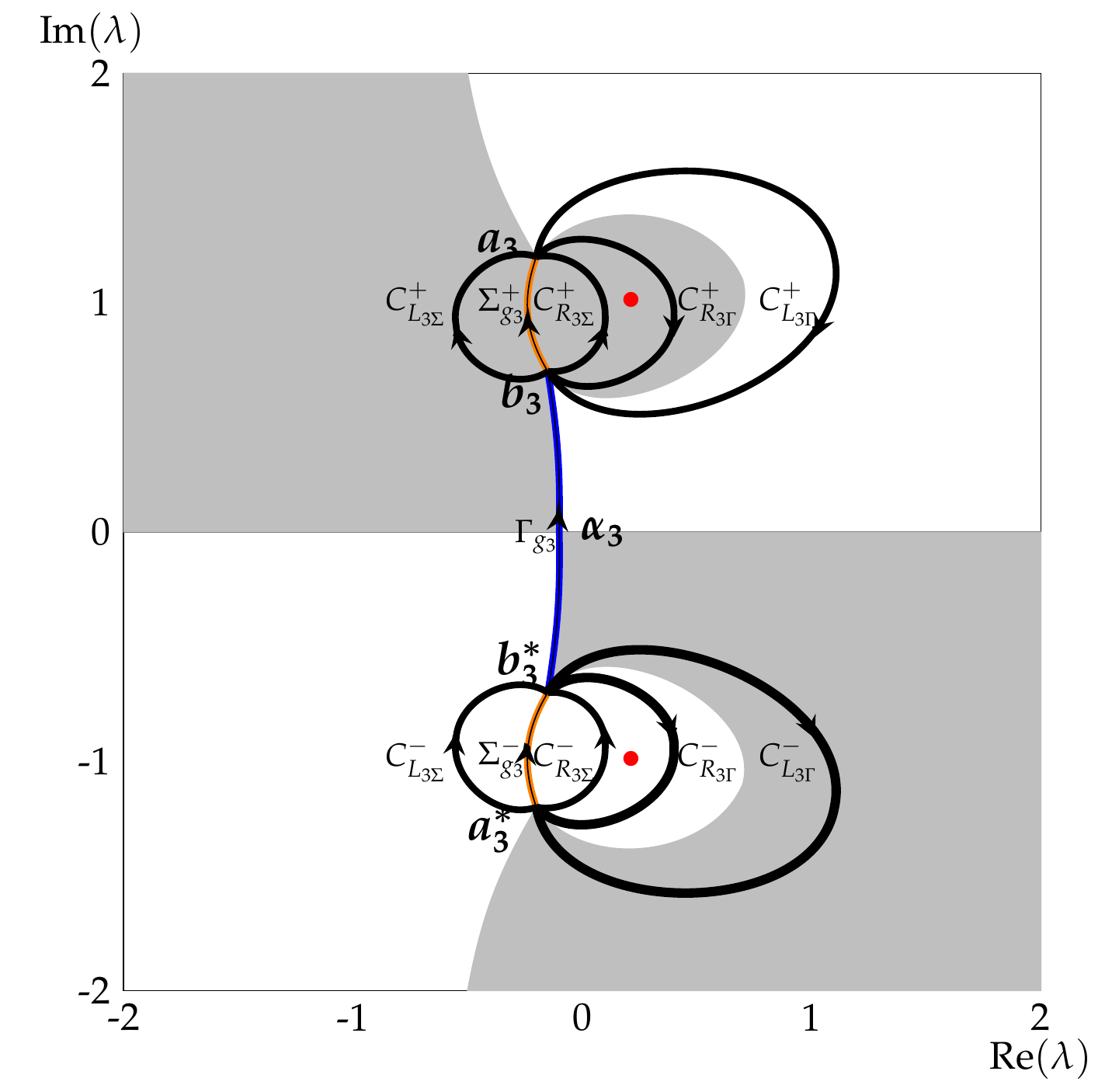}
\caption{The contour plot of ${\Im}\left(h_3\left(\lambda\right)\right)$ in the genus-one region with the parameters $\chi=\frac{9}{5}, \tau=1$ and $k=2$, where ${\Im}\left(h_3\left(\lambda\right)\right)<0$ (shaded) and ${\Im}\left(h_3\left(\lambda\right)\right)>0$ (unshaded), the red dots are the singularities $\lambda=\pm\ii$, the singularities $\lambda=\pm 2\ii$ are omitted. The left one gives the original contour and the right one is the corresponding contour deformation. }
\label{fig:genus-one}
\end{figure}

We shall analyze the asymptotics of genus-one region as shown in the contour plot of Fig.\ref{fig:genus-one}. Similar to the algebraic-decay and genus-zero region, we still would like to deform the original jump contours of $\mathbf{R}(\lambda; \chi, \tau)$ with the Deift-Zhou nonlinear steepest descent method.

Set
\begin{equation}
\mathbf{S}_{3}(\lambda; \chi, \tau):=\left\{\begin{aligned}&\mathbf{R}(\lambda; \chi, \tau)\ee^{-\ii n\vartheta(\lambda; \chi, \tau)\sigma_3}\mathbf{Q}_{c}\ee^{\ii n\vartheta(\lambda; \chi, \tau)\sigma_3},\quad &\lambda\in D_0\cap\left(D_{3}^{+}\cup D_{3}^{-}\right)^{c},\\
&\mathbf{R}(\lambda; \chi, \tau),\quad &\text{otherwise},
\end{aligned}\right.
\end{equation}
where $D_{3}^{+}=R_{3\Sigma}^{+}\cup \Omega_{3}^{+}\cup R_{3\Gamma}^{+}, D_{3}^{-}=R_{3\Sigma}^{-}\cup \Omega_{3}^{-}\cup R_{3\Gamma}^{-}.$
Under this case, the jump matrices will transfer to the boundary of $D_3^{+}$ and $D_3^{-}$. Moreover, set
\begin{equation}\label{eq:S-genus-one}
\begin{aligned}
\mathbf{T}_{3}(\lambda; \chi, \tau):&=\mathbf{S}_{3}(\lambda; \chi, \tau)\ee^{-\ii n\vartheta(\lambda; \chi, \tau)\sigma_3}
\left(\mathbf{Q}_{R}^{[2]}\right)^{-1}\ee^{\ii n\vartheta(\lambda; \chi, \tau)\sigma_3}\ee^{\ii n g_{3}(\lambda)\sigma_3},\quad &\lambda\in L_{3\Gamma}^{+},\\
\mathbf{T}_{3}(\lambda; \chi, \tau):&=\mathbf{S}_{3}(\lambda; \chi, \tau)\mathbf{Q}_{L}^{[2]}\ee^{-\ii n\vartheta(\lambda; \chi, \tau)\sigma_3}
\mathbf{Q}_{C}^{[2]}\ee^{\ii n\vartheta(\lambda; \chi, \tau)\sigma_3}\ee^{\ii n g_{3}(\lambda)\sigma_3},\quad &\lambda\in R_{3\Gamma}^{+},\\
\mathbf{T}_{3}(\lambda; \chi, \tau):&=\mathbf{S}_{3}(\lambda; \chi, \tau)\mathbf{Q}_{L}^{[2]}\ee^{\ii n g_{3}(\lambda; \chi, \tau)\sigma_3},\quad &\lambda\in \Omega_{3}^{+},\\
\mathbf{T}_{3}(\lambda; \chi, \tau):&=\mathbf{S}_{3}(\lambda; \chi, \tau)\mathbf{Q}_{L}^{[2]}\ee^{-\ii n\vartheta(\lambda; \chi, \tau)\sigma_3}\mathbf{Q}_{L}^{[3]}\ee^{\ii n\vartheta(\lambda; \chi, \tau)\sigma_3}\ee^{\ii n g_{3}(\lambda)\sigma_3},\quad &\lambda\in R_{3\Sigma}^{+},\\
\mathbf{T}_{3}(\lambda; \chi, \tau):&=\mathbf{S}_{3}(\lambda; \chi, \tau)\ee^{-\ii n\vartheta(\lambda; \chi, \tau)\sigma_3}
\left(\mathbf{Q}_{R}^{[3]}\right)^{-1}\ee^{\ii n\vartheta(\lambda; \chi, \tau)\sigma_3}\ee^{\ii n g_{3}(\lambda)\sigma_3},\quad &\lambda\in L_{3\Sigma}^{+},\\
\mathbf{T}_{3}(\lambda; \chi, \tau):&=\mathbf{S}_{3}(\lambda; \chi, \tau)\ee^{-\ii n\vartheta(\lambda; \chi, \tau)\sigma_3}
\left(\mathbf{Q}_{R}^{[1]}\right)^{-1}\ee^{\ii n\vartheta(\lambda; \chi, \tau)\sigma_3}\ee^{\ii n g_{3}(\lambda)\sigma_3},\quad &\lambda\in L_{3\Gamma}^{-},\\
\mathbf{T}_{3}(\lambda; \chi, \tau):&=\mathbf{S}_{3}(\lambda; \chi, \tau)\mathbf{Q}_{L}^{[1]}\ee^{-\ii n\vartheta(\lambda; \chi, \tau)\sigma_3}
\mathbf{Q}_{C}^{[1]}
\ee^{\ii n\vartheta(\lambda; \chi, \tau)\sigma_3}\ee^{\ii n g_{3}(\lambda)\sigma_3},\quad &\lambda\in R_{3\Gamma}^{-},\\
\mathbf{T}_{3}(\lambda; \chi, \tau):&=\mathbf{S}_{3}(\lambda; \chi, \tau)\mathbf{Q}_{L}^{[1]}\ee^{\ii n g_{3}(\lambda)\sigma_3},\quad &\lambda\in \Omega_{3}^{-},\\
\mathbf{T}_{3}(\lambda; \chi, \tau):&=\mathbf{S}_{3}(\lambda; \chi, \tau)\mathbf{Q}_{L}^{[1]}\ee^{-\ii n\vartheta(\lambda; \chi, \tau)\sigma_3}
\mathbf{Q}_{L}^{[4]}\ee^{\ii n\vartheta(\lambda; \chi, \tau)\sigma_3}\ee^{\ii n g_{3}(\lambda)\sigma_3},\quad &\lambda\in R_{3\Sigma}^{-},\\
\mathbf{T}_{3}(\lambda; \chi, \tau):&=\mathbf{S}_{3}(\lambda; \chi, \tau)\ee^{-\ii n\vartheta(\lambda; \chi, \tau)\sigma_3}
\left(\mathbf{Q}_{R}^{[4]}\right)^{-1}\ee^{\ii n\vartheta(\lambda; \chi, \tau)\sigma_3}\ee^{\ii n g_{3}(\lambda)\sigma_3},\quad &\lambda\in L_{3\Sigma}^{-},
\end{aligned}
\end{equation}
in other regions, we set $\mathbf{T}_{3}(\lambda; \chi, \tau):=\mathbf{S}_{3}(\lambda; \chi, \tau)\ee^{\ii n g_{3}(\lambda)\sigma_3}$. Then the jump conditions about $\mathbf{T}_{3}(\lambda; \chi, \tau)$ change into
\begin{equation}\label{eq:jump-genus-one}
\begin{aligned}
\mathbf{T}_{3, +}(\lambda; \chi, \tau)&=\mathbf{T}_{3, -}(\lambda; \chi, \tau)\ee^{-\ii nh_3(\lambda)\sigma_3}
\mathbf{Q}_{R}^{[2]}\ee^{\ii nh_3(\lambda)\sigma_3},\quad &\lambda\in C_{L_{3\Gamma}}^{+},\\
\mathbf{T}_{3, +}(\lambda; \chi, \tau)&=\mathbf{T}_{3, -}(\lambda; \chi, \tau)\ee^{-\ii nh_3(\lambda)\sigma_3}\mathbf{Q}_{C}^{[2]}\ee^{\ii nh_3(\lambda)\sigma_3},\quad &\lambda\in C_{R_{3\Gamma}}^{+},\\
\mathbf{T}_{3, +}(\lambda; \chi, \tau)&=\mathbf{T}_{3, -}(\lambda; \chi, \tau)\ee^{-\ii nh_3(\lambda)\sigma_3}\mathbf{Q}_{L}^{[3]}\ee^{\ii nh_3(\lambda)\sigma_3},\quad &\lambda\in C_{R_{3\Sigma}}^{+},\\
\mathbf{T}_{3, +}(\lambda; \chi, \tau)&=\mathbf{T}_{3, -}(\lambda; \chi, \tau)\ee^{-\ii nh_3(\lambda)\sigma_3}\mathbf{Q}_{R}^{[3]}\ee^{\ii nh_3(\lambda)\sigma_3},\quad &\lambda\in C_{L_{3\Sigma}}^{+},\\
\mathbf{T}_{3, +}(\lambda; \chi, \tau)&=\mathbf{T}_{3, -}(\lambda; \chi, \tau)\ee^{-\ii nh_3(\lambda)\sigma_3}\mathbf{Q}_{R}^{[1]}\ee^{\ii nh_3(\lambda)\sigma_3},\quad &\lambda\in C_{L_{3\Gamma}}^{-},\\
\mathbf{T}_{3, +}(\lambda; \chi, \tau)&=\mathbf{T}_{3, -}(\lambda; \chi, \tau)\ee^{-\ii nh_3(\lambda)\sigma_3}\mathbf{Q}_{C}^{[1]}\ee^{\ii nh_3(\lambda)\sigma_3},\quad &\lambda\in C_{R_{3\Gamma}}^{-},\\
\mathbf{T}_{3, +}(\lambda; \chi, \tau)&=\mathbf{T}_{3, -}(\lambda; \chi, \tau)\ee^{-\ii nh_3(\lambda)\sigma_3}\mathbf{Q}_{L}^{[4]}\ee^{\ii nh_3(\lambda)\sigma_3},\quad &\lambda\in C_{R_{3\Sigma}}^{-},\\
\mathbf{T}_{3, +}(\lambda; \chi, \tau)&=\mathbf{T}_{3, -}(\lambda; \chi, \tau)\ee^{-\ii nh_3(\lambda)\sigma_3}\mathbf{Q}_{R}^{[4]}\ee^{\ii nh_3(\lambda)\sigma_3},\quad &\lambda\in C_{L_{3\Sigma}}^{-},\\
\mathbf{T}_{3, +}(\lambda; \chi, \tau)&=\mathbf{T}_{3, -}(\lambda; \chi, \tau)\begin{bmatrix}0&\ee^{-\ii n\kappa_3}\\
-\ee^{\ii n\kappa_3}&0
\end{bmatrix},&\lambda\in \Sigma_{g_3}^{+}\cup \Sigma_{g_3}^{-},\\
\mathbf{T}_{3, +}(\lambda; \chi, \tau)&=\mathbf{T}_{3, -}(\lambda; \chi, \tau)\begin{bmatrix}\ee^{\ii nd_3}&0\\
0&\ee^{-\ii nd_3}
\end{bmatrix},&\lambda\in \Gamma_{g_3}.
\end{aligned}
\end{equation}
When $n$ is large, we know that the jump conditions in Eq.\eqref{eq:jump-genus-one} will tend to the identity matrix exponentially except for the contours $\Sigma_{g_3}^{+}\cup \Sigma_{g_3}^{-} $ and $\Gamma_{g_3}$. Then we give the parametrix construction for $\mathbf{T}_{3}(\lambda; \chi, \tau)$.
\subsection{Parametrix construction for $\mathbf{T}_{3}(\lambda; \chi, \tau)$}
Similar to the genus-zero region, we begin to establish the outer parametrix $\dot{\mathbf{T}}_{3}^{\rm out}(\lambda; \chi, \tau)$ for $\mathbf{T}_{3}(\lambda; \chi, \tau)$. With the jump conditions in Eq.\eqref{eq:jump-genus-one}, we give a RHP about $\dot{\mathbf{T}}_{3}^{\rm out}(\lambda; \chi, \tau)$.
\begin{rhp}\label{rhp:outer-genus-one}
(RHP for the outer parametrix $\dot{\mathbf{T}}_{3}^{\rm out}(\lambda; \chi, \tau)$) Seek a $2\times 2$ matrix $\dot{\mathbf{T}}_{3}^{\rm out}(\lambda; \chi, \tau)$ satisfying the following conditions.
\begin{itemize}
\item {\bf Analyticity:} $\dot{\mathbf{T}}_{3}^{\rm out}(\lambda; \chi, \tau)$ is analytic in $\mathbb{C}\setminus\left(\Sigma_{g_3}^{+}\cup \Sigma_{g_3}^{-}\cup \Gamma_{g_3}\right)$.
\item {\bf Jump condition:} $\dot{\mathbf{T}}_{3}^{\rm out}(\lambda; \chi, \tau)$ takes the continuous boundary values on $\Sigma_{g_3}^{+}\cup \Sigma_{g_3}^{-}\cup \Gamma_{g_3}$, and it is related by the jump conditions $\dot{\mathbf{T}}_{3,+}^{\rm out}(\lambda; \chi, \tau)=\dot{\mathbf{T}}_{3,-}^{\rm out}(\lambda; \chi, \tau)\mathbf{V}_{\dot{\mathbf{T}}_{3}^{\rm out}}(\lambda;\chi, \tau)$, where $\mathbf{V}_{\dot{\mathbf{T}}_{3}^{\rm out}}(\lambda;\chi, \tau)$ is
    \begin{equation}
   \mathbf{V}_{\dot{\mathbf{T}}_{3}^{\rm out}}(\lambda;\chi, \tau)=\left\{\begin{aligned} &\begin{bmatrix}0&\ee^{-\ii n\kappa_3}\\
-\ee^{\ii n\kappa_3}&0
\end{bmatrix},\quad\lambda\in \Sigma_{g_3}^{+}\cup\Sigma_{g_3}^{-},\\
&\begin{bmatrix}\ee^{\ii nd_3}&0\\
0&\ee^{-\ii nd_3}
\end{bmatrix},\quad\lambda\in \Gamma_{g_3}.
\end{aligned}\right.
    \end{equation}
\item {\bf Normalization:} When $\lambda\to\infty$, $\dot{\mathbf{T}}_{3}^{\rm out}(\lambda; \chi, \tau)\to\mathbb{I}$.
\end{itemize}
\end{rhp}
To solve $\dot{\mathbf{T}}_{3}^{\rm out}(\lambda; \chi, \tau)$, we first introduce a scalar function $F_{3}(\lambda; \chi, \tau)$ with the following jump conditions,
\begin{equation}
\begin{aligned}
F_{3,+}(\lambda; \chi, \tau)+F_{3,-}(\lambda; \chi, \tau)&=\ii n\kappa_3,\quad\lambda\in \Sigma_{g_3}^{+}\cup \Sigma_{g_3}^{-},\\
F_{3,+}(\lambda; \chi, \tau)-F_{3,-}(\lambda; \chi, \tau)&=\ii n d_3,\quad\lambda\in\Gamma_{g_3}.
\end{aligned}
\end{equation}
Then $F_{3}(\lambda; \chi, \tau)$ can be calculated by using $\mathcal{R}_{3}(\lambda)$, that is
\begin{equation}\label{eq:F3130}
F_3(\lambda; \chi, \tau)=\frac{\mathcal{R}_{3}(\lambda)}{2\pi\ii}\left(\int_{\Sigma_{g_3}^{+}}\frac{\ii n \kappa_3}{\mathcal{R}_{3}(\xi)(\xi-\lambda)}d\xi+\int_{\Sigma_{g_3}^{-}}\frac{\ii n\kappa_3}{\mathcal{R}_{3}(\xi)(\xi-\lambda)}d\xi+\int_{\Gamma_{g_3}}\frac{\ii nd_3}{\mathcal{R}_{3}(\xi)(\xi-\lambda)}d\xi\right).
\end{equation}
Expanding $F_3(\lambda; \chi, \tau)$ as $\lambda\to\infty$, we have the following series formula,
\begin{equation}
F_{3}(\lambda; \chi, \tau)=F_{31}\lambda+F_{30}+\mathcal{O}(\lambda^{-1}),
\end{equation}
where
\begin{equation}
\begin{aligned}
F_{31}&=-\frac{1}{2\pi \ii}\left(\int_{\Sigma_{g_3}^{+}}\frac{\ii n \kappa_3}{\mathcal{R}_{3}(\xi)}d\xi+\int_{\Sigma_{g_3}^{-}}\frac{\ii n\kappa_3}{\mathcal{R}_{3}(\xi)}d\xi+\int_{\Gamma_{g_3}}\frac{\ii nd_3}{\mathcal{R}_{3}(\xi)}d\xi\right),\\
F_{30}&=-\frac{1}{2\pi \ii}\left(\int_{\Sigma_{g_3}^{+}}\frac{\ii n \kappa_3}{\mathcal{R}_{3}(\xi)}\xi d\xi+\int_{\Sigma_{g_3}^{-}}\frac{\ii n\kappa_3}{\mathcal{R}_{3}(\xi)}\xi d\xi+\int_{\Gamma_{g_3}}\frac{\ii nd_3}{\mathcal{R}_{3}(\xi)}\xi d\xi\right)-\left(\Re(a_3)+\Re(b_3)\right)F_{31}.
\end{aligned}
\end{equation}
With a generalized residue theorem, we can simplify $F_{31}$ and $F_{30}$ as
\begin{equation}
F_{31}=-\frac{1}{2\pi \ii}\int_{\Gamma_{g_3}}\frac{\ii nd_3}{\mathcal{R}_{3}(\xi)}d\xi,\quad F_{30}=\frac{1}{2}\ii n\kappa_3-\frac{1}{2\pi \ii}\int_{\Gamma_{g_3}}\frac{\ii nd_3}{\mathcal{R}_{3}(\xi)}\xi d\xi-\left(\Re(a_3)+\Re(b_3)\right)F_{31}.
\end{equation}
Through $F_{3}(\lambda; \chi, \tau)$, we can redefine a new matrix $\mathbf{P}_{3}(\lambda; \chi, \tau)$ as
\begin{equation}
\mathbf{P}_{3}(\lambda; \chi, \tau)={\rm diag}\left(\ee^{F_{30}}, \ee^{-F_{30}}\right)\dot{\mathbf{T}}_{3}^{\rm out}(\lambda; \chi, \tau){\rm diag}\left(\ee^{-F_{3}(\lambda; \chi, \tau)}, \ee^{F_{3}(\lambda; \chi, \tau)}\right),
\end{equation}
which satisfies a constant jump matrix
\begin{equation}
\mathbf{P}_{3,+}(\lambda; \chi, \tau)=\mathbf{P}_{3,-}(\lambda; \chi, \tau)(\ii\sigma_2),\quad\lambda\in\Sigma_{g_3}^{\pm}.
\end{equation}
When $\lambda\to\infty$, we have the normalization condition
\begin{equation}\label{eq:P3-bc}
\mathbf{P}_{3}(\lambda; \chi, \tau){\rm diag}\left(\ee^{F_{31}\lambda}, \ee^{-F_{31}\lambda}\right)\to\mathbb{I}.
\end{equation}
Before solving $\mathbf{P}_{3}(\lambda; \chi, \tau)$ with the boundary condition, we first give a Riemann surface $\Sigma_{1}$ with two sheets $\Sigma_{11}$ and $\Sigma_{12}$, whose basis can be set as $\{\alpha_{11}, \beta_{11}\}$ cycles, which is shown in Fig. \ref{circle:genus-one}.
\begin{figure}[ht]
\centering
\includegraphics[width=0.3\textwidth]{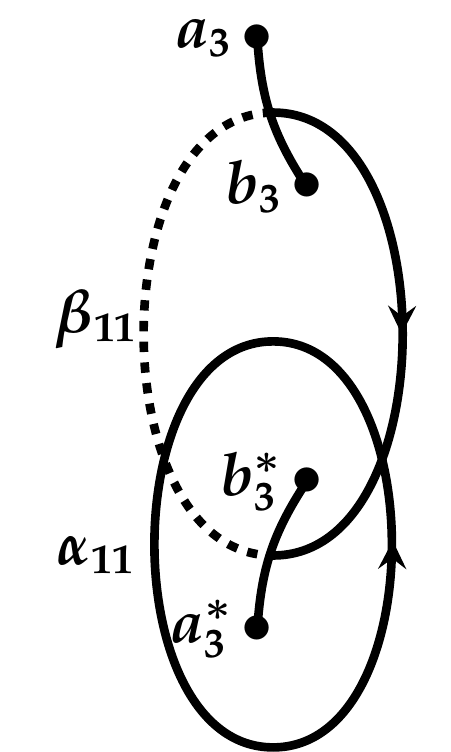}
\caption{Homology basis for the Riemann surface of genus-one. The solid paths indicate the first sheet and the dashed lines lie in the second sheet.}
\label{circle:genus-one}
\end{figure}

Then we give the Abel integrals
\begin{equation}\label{eq:Abel-int}
A(\lambda):=\frac{2\pi\ii}{\oint_{\alpha_{11}}\frac{1}{\mathcal{R}_{3}(\xi)} d\xi}\int_{a_3^*}^{\lambda}\frac{1}{\mathcal{R}_{3}(\xi)}d\xi,\quad B:=\frac{2\pi\ii}{\oint_{\alpha_{11}}\frac{1}{\mathcal{R}_{3}(\xi)} d\xi}\oint_{\beta_{11}}\frac{1}{\mathcal{R}_{3}(\xi)}d\xi.
\end{equation}
These two Abel integrals have the normalization condition
\begin{equation}
A_{+}(b_3^*)=-\ii\pi,\quad B=2\pi\ii \varsigma,
\end{equation}
where $\varsigma$ can be reduced by the elliptic integrals:
\begin{equation}\label{eq:varsigma}
\varsigma=\frac{\ii K(1-m)}{K(m)}.
\end{equation}
In this genus-one region, we define the lattice $\Lambda_{1}$ as
\begin{equation}
\Lambda_{1}:=2\pi\ii j+ Bk, \quad j,k\in\mathbb{Z}.
\end{equation}
Then the Abel integrals have the following relations:
\begin{equation}\label{eq:Arelation}
\begin{aligned}
A_{+}(\lambda)+A_{-}(\lambda)=-B\mod 2\pi \ii \mathbb{Z} ,\quad\lambda\in\Sigma_{g_3}^{+},\\
A_{+}(\lambda)-A_{-}(\lambda)=-2\pi\ii\mod 2\pi \ii \mathbb{Z},\quad\lambda\in\Gamma_{g_3},\\
A_{+}(\lambda)+A_{-}(\lambda)=0\mod 2\pi \ii \mathbb{Z},\quad\lambda\in\Sigma_{g_3}^{-}.
\end{aligned}
\end{equation}
Next, we introduce another Abel differential $\Delta$ as
\begin{equation}\label{eq:U-genus-one}
\Delta:=\frac{\xi^2-(\Re(a_3)+\Re(b_3))\xi+c}{\mathcal{R}_{3}(\xi)}d\xi,
\end{equation}
where $c$ is chosen such that $\Delta$ satisfies a normalization condition $\oint_{\alpha_{11}}\Delta=0.$ Correspondingly, we set the integration of $\oint_{\beta_{11}}\Delta$ as $U.$ Under this condition, $\Delta$ satisfies the following relations,
\begin{equation}\label{eq:Urelation}
\begin{aligned}
\int_{a_{3}^*}^{\lambda^+}\Delta+\int_{a_{3}^*}^{\lambda^{-}}\Delta=-U,\quad\lambda\in \Sigma_{g_3}^{+},\\
\int_{a_{3}^*}^{\lambda^+}\Delta+\int_{a_{3}^*}^{\lambda^{-}}\Delta=0,\quad\lambda\in \Sigma_{g_3}^{-},
\end{aligned}
\end{equation}
and when $\lambda\to\infty$,
\begin{equation}
J_{31}:=\lim\limits_{\lambda\to\infty}\left(\int_{a_3^*}^{\lambda}\Delta -\lambda\right)
\end{equation}
exists.

Following the result in \cite{Bilman-JDE-2021}, we know that the solution of $\mathbf{P}_{3}(\lambda; \chi, \tau)$ can be given by the Theta function, thus we give a brief review about the definition of Theta function.
\begin{definition}\label{prop:theta}
The $\Theta(\lambda)$ function is defined as
\begin{equation}
\begin{aligned}
&\Theta(\lambda)\equiv\Theta(\lambda; \mathbf{B}):=\sum\limits_{\mathbf{m}\in\mathbb{Z}^g}\ee^{\frac{1}{2}\langle \mathbf{m}, \mathbf{B}\mathbf{m}\rangle+\langle\mathbf{m},\lambda\rangle},\\
&\Theta(\lambda+2\pi\ii\mathbf{e}_j)=\Theta(\lambda),\qquad \Theta(\lambda+\mathbf{B}\mathbf{e}_j)=\ee^{-\frac{1}{2}B_{jj}-\lambda_j}\Theta(\lambda),
\end{aligned}
\end{equation}
where $\mathbf{e}_j$s are the unit basis vectors in $\mathbb{C}^{g}$ with the coordinates $(\mathbf{e}_{j})_{k}=\delta_{jk}$, and $\mathbf{B}\mathbf{e}_{j}$s are the $j-th$ column of matrix $\mathbf{B}$.
Especially, when $g=1$, we have
\begin{equation}
\Theta(\lambda; \mathbf{B}):=\sum\limits_{k\in\mathbb{Z}}\ee^{k\lambda+\frac{1}{2}Bk^2}=\theta_{3}\left(\frac{\lambda}{2\ii}, \ee^{\frac{B}{2}}\right),\\
\end{equation}
where $\theta_3(\lambda; \wp)$ is the third Jacobi-Theta function, defined as
\begin{equation}
\theta_3(\lambda; \wp)=\sum\limits_{k\in\mathbb{Z}}\ee^{2\ii k\lambda}\wp^{k^2}.
\end{equation}
\end{definition}
Consider the matrix $\pmb{\mathcal{P}}_{3}(\lambda; \chi, \tau)$ as
\begin{equation}
\pmb{\mathcal{P}}_{3}(\lambda; \chi, \tau):=\begin{bmatrix}\frac{\Theta\left(A(\lambda)+A(Q)+\ii\pi+\frac{B}{2}-F_{31}U\right)}{\Theta\left(A(\lambda)+A(Q)+\ii\pi+\frac{B}{2}\right)}&
\frac{\Theta\left(A(\lambda)-A(Q)-\ii\pi-\frac{B}{2}+F_{31}U\right)}{\Theta\left(A(\lambda)-A(Q)-\ii\pi-\frac{B}{2}\right)}\\
\frac{\Theta\left(A(\lambda)-A(Q)-\ii\pi-\frac{B}{2}-F_{31}U\right)}{\Theta\left(A(\lambda)-A(Q)-\ii\pi-\frac{B}{2}\right)}&
\frac{\Theta\left(A(\lambda)+A(Q)+\ii\pi+\frac{B}{2}+F_{31}U\right)}{\Theta\left(A(\lambda)+A(Q)+\ii\pi+\frac{B}{2}\right)}
\end{bmatrix}{\rm diag}\left(\ee^{-F_{31}\int_{a_3^*}^{\lambda}\Delta},\ee^{F_{31}\int_{a_3^*}^{\lambda}\Delta}\right),
\end{equation}
where the factor $\ii\pi+\frac{B}{2}$ is the Riemann constant \cite{belokolos1994algebro,kotlyarov2017planar} for genus-one and $Q$ is a constant to be determined. From the relations Eq.\eqref{eq:Arelation} and Eq.\eqref{eq:Urelation}, we know that $\pmb{\mathcal{P}}_3(\lambda; \chi, \tau)$ has the jump relation
\begin{equation}
\pmb{\mathcal{P}}_{3,+}(\lambda; \chi, \tau)=\pmb{\mathcal{P}}_{3,-}(\lambda; \chi, \tau)\begin{bmatrix}0&1\\
1&0\end{bmatrix},\quad \lambda\in\Sigma_{g_3}^{+}\cup\Sigma_{g_3}^{-}.
\end{equation}
Using this new matrix $\pmb{\mathcal{P}}_{3}(\lambda; \chi, \tau)$, we can get the solution of $\mathbf{P}_{3}(\lambda; \chi, \tau)$ as
\begin{equation}
\mathbf{P}_{3}(\lambda; \chi, \tau):=\frac{1}{2}{\rm diag}\left(C_{31}, C_{32}\right)\begin{bmatrix}\left(\gamma_{3}(\lambda)+\frac{1}{\gamma_{3}(\lambda)}\right)\pmb{\mathcal{P}}_{3}(\lambda; \chi, \tau)_{11}&-\ii \left(\gamma_{3}(\lambda)-\frac{1}{\gamma_{3}(\lambda)}\right)\pmb{\mathcal{P}}_{3}(\lambda; \chi, \tau)_{12}\\
\ii \left(\gamma_{3}(\lambda)-\frac{1}{\gamma_{3}(\lambda)}\right)\pmb{\mathcal{P}}_{3}(\lambda; \chi, \tau)_{21}&\left(\gamma_{3}(\lambda)+\frac{1}{\gamma_{3}(\lambda)}\right)\pmb{\mathcal{P}}_{3}(\lambda; \chi, \tau)_{22}
\end{bmatrix},
\end{equation}
where $C_{31}$ and $C_{32}$ are given by
\begin{equation}
\begin{aligned}
C_{31}:=\frac{\Theta\left(A(\infty)+A(Q)+\ii\pi+\frac{B}{2}\right)}{\Theta\left(A(\infty)+A(Q)+\ii\pi+\frac{B}{2}-F_{31}U\right)}\ee^{J_{31}F_{31}},\\
C_{32}:=\frac{\Theta\left(A(\infty)+A(Q)+\ii\pi+\frac{B}{2}\right)}{\Theta\left(A(\infty)+A(Q)+\ii\pi+\frac{B}{2}+F_{31}U\right)}\ee^{-J_{31}F_{31}},
\end{aligned}
\end{equation}
and $\gamma_{3}(\lambda)=\left(\frac{(\lambda-b_3)(\lambda-a_3^*)}{(\lambda-b_3^*)(\lambda-a_3)}\right)^{\frac{1}{4}}$ satisfies $\gamma_{3,+}(\lambda)=\ii\gamma_{3,-}(\lambda).$

Observing the off-diagonal element of $\mathbf{P}_{3}(\lambda; \chi, \tau)$, we see that $\lambda=\frac{\Re(a_3)\Im(b_3)-\Re(b_3)\Im(a_3)}{\Im(b_3)-\Im(a_3)}$ is a simple zero of $\mathbf{P}_{3}(\lambda; \chi, \tau)$ and $\lambda=Q$ is a simple pole, thus we can choose $Q=\frac{\Re(a_3)\Im(b_3)-\Re(b_3)\Im(a_3)}{\Im(b_3)-\Im(a_3)}$ to make sure that $\mathbf{P}_{3}(\lambda; \chi, \tau)$ is well defined in the complex plane. Then the outer parametrix has been constructed completely. The outer parametrix can match $\mathbf{T}_{3}(\lambda; \chi, \tau)$ except for the singularities $a_{3}, a_{3}^*, b_{3}$ and $b_{3}^*$. Similar to the genus-zero region, in the neighborhood of $a_{3}, a_{3}^*, b_{3}$ and $b_{3}^*$, we can set the local parametrix as $\dot{\mathbf{T}}_{3}^{a_3}(\lambda; \chi, \tau), \dot{\mathbf{T}}_{3}^{a_3^*}(\lambda; \chi, \tau), \dot{\mathbf{T}}_{3}^{b_3}(\lambda; \chi, \tau), \dot{\mathbf{T}}_{3}^{b_3^*}(\lambda; \chi, \tau)$ in these small disks $D_{a_3}(\delta), D_{a_3^*}(\delta), D_{b_3}(\delta)$ and $D_{b_3^*}(\delta)$ respectively. Then the global paramtrix of $\mathbf{T}_{3}(\lambda; \chi, \tau)$ is
\begin{equation}
\dot{\mathbf{T}}_{3}(\lambda; \chi, \tau):=\left\{\begin{aligned}&\dot{\mathbf{T}}_{3}^{a_3}(\lambda; \chi, \tau),\quad\lambda\in D_{a_3}(\delta),\\&\dot{\mathbf{T}}_{3}^{a_3^*}(\lambda; \chi, \tau),\quad\lambda\in D_{a_3^*}(\delta),\\&\dot{\mathbf{T}}_{3}^{b_3}(\lambda; \chi, \tau),\quad\lambda\in D_{b_3}(\delta),\\
&\dot{\mathbf{T}}_{3}^{b_3^*}(\lambda; \chi, \tau),\quad\lambda\in D_{b_3^*}(\delta),\\
&\dot{\mathbf{T}}_{3}^{\rm out}(\lambda; \chi, \tau),\quad\lambda\in\mathbb{C}\setminus\left(\overline{D_{a_3}(\delta)\cup D_{a_3^*}(\delta)\cup D_{b_3}(\delta)\cup D_{b_3^*}(\delta)}\cup \Sigma_{g_3}^{+}\cup \Sigma_{g_3}^{-}\cup \Gamma_{g_3}\right).
\end{aligned}\right.
\end{equation}
Next, we will give the error analysis between $\mathbf{T}_{3}(\lambda; \chi, \tau)$ and its parametrix $\dot{\mathbf{T}}_{3}(\lambda; \chi, \tau)$.
\subsection{Error analysis}
Set the error function between $\mathbf{T}_{3}(\lambda; \chi, \tau)$ and $\dot{\mathbf{T}}_{3}(\lambda; \chi, \tau)$ as
\begin{equation}
\mathcal{E}_{3}(\lambda; \chi, \tau):=\mathbf{T}_{3}(\lambda; \chi, \tau)\left(\dot{\mathbf{T}}_{3}(\lambda; \chi, \tau)\right)^{-1}.
\end{equation}
Then the jump matrices about $\mathcal{E}_{3}(\lambda; \chi, \tau)$ can be set as $\mathbf{V}_{\mathcal{E}_{3}}(\lambda; \chi, \tau)$. Similar to the analysis in genus-zero region, the errors on the contours $\partial D_{a_3}(\delta), \partial D_{a_3^*}(\delta), \partial D_{b_3}(\delta),\partial D_{b_3^*}(\delta)$ can be given by the Airy function, and the jump matrix $\mathbf{V}_{\mathcal{E}_{3}}(\lambda; \chi, \tau)$ has the error estimation,
\begin{equation}
\begin{aligned}
\|\mathbf{V}_{\mathcal{E}_{3}}(\lambda; \chi, \tau)-\mathbb{I}\|&=\mathcal{O}\left(\ee^{-\mu_3 n}\right)(\mu_3>0),\quad \lambda\in C_{L_{3\Sigma}}^{\pm}{\cup} C_{R_{3\Sigma}}^{\pm}{\cup} C_{R_{3\Gamma}}^{\pm}{\cup} C_{L_{3\Gamma}}^{\pm},\\
\|\mathbf{V}_{\mathcal{E}_{3}}(\lambda; \chi, \tau)-\mathbb{I}\|&=\mathcal{O}(n^{-1}),\quad \lambda\in\partial D_{a_3}(\delta)\cup \partial D_{a_3^*}(\delta)\cup\partial D_{b_3}(\delta)\cup\partial D_{b_3^*}(\delta).\\
\end{aligned}
\end{equation}
Under this case, the potential function $q^{[n]}(n\chi, n\tau)$ can be given by
\begin{equation}\label{eq:qn-genus-1}
\begin{aligned}
q^{[n]}(n\chi, n\tau)&=2\ii\lim\limits_{\lambda\to\infty}\lambda\mathbf{T}_{3}(\lambda; \chi, \tau)_{12}\\
&=2\ii\lim\limits_{\lambda\to\infty}\lambda\left(\mathcal{E}_{3}(\lambda; \chi, \tau)\dot{\mathbf{T}}^{\rm out}_{3}(\lambda; \chi, \tau)\right)_{12}\\
&=2\ii\lim\limits_{\lambda\to\infty}\lambda\left(\mathcal{E}_{3,11}(\lambda; \chi, \tau)\dot{\mathbf{T}}_{3,12}^{\rm out}(\lambda; \chi, \tau)+\mathcal{E}_{3,12}(\lambda; \chi, \tau)\dot{\mathbf{T}}_{3,22}^{\rm out}(\lambda; \chi, \tau)\right)\\
&=2\ii\lim\limits_{\lambda\to\infty}\lambda\dot{\mathbf{T}}_{3,12}^{\rm out}(\lambda; \chi, \tau)+\mathcal{O}(n^{-1}).
\end{aligned}
\end{equation}
Substituting $\dot{\mathbf{T}}_3^{\rm out}(\lambda; \chi, \tau)$ into Eq.\eqref{eq:qn-genus-1}, we have
\begin{multline}\label{eq:qn-genus-1-1}
q^{[n]}(n\chi, n\tau)=\frac{\Theta\left(A(\infty)+A(Q)+\ii\pi+\frac{B}{2}\right)}{\Theta\left(A(\infty)+A(Q)+\ii\pi+\frac{B}{2}-F_{31}U\right)}\frac{\Theta\left(A(\infty)-A(Q)-\ii\pi-\frac{B}{2}+F_{31}U\right)}{\Theta\left(A(\infty)-A(Q)-\ii\pi-\frac{B}{2}\right)}\\
\times \ii\left(\Im(a_3)-\Im(b_3)\right)\ee^{2F_{31}J_{31}-2F_{30}}+\mathcal{O}(n^{-1}).
\end{multline}
Moreover, we can simplify $q^{[n]}(n\chi, n\tau)$ into another equivalent formula following the method \cite{Biondini-CPAM-2017}. Consider a new function $f(\lambda)$
\begin{equation}
f(\lambda)=\gamma_3(\lambda)\left(\gamma_{3}(\lambda)^{-1}-\gamma_{3}(\lambda)\right)
\end{equation}
in the Riemann surface. By a simple calculation, we get that $\lambda=a_3$ and $\lambda=b_3^*$ are two poles of $f(\lambda)$ and $\lambda=Q$ and $\lambda=\infty$ are two zeros of it. With the aid of the Abel theorem \cite{belokolos1994algebro}, we have
\begin{equation}\label{eq:AQ}
A(Q)=A(a_3)+A(b_3^*)-A(\infty).
\end{equation}
Plugging $A(Q)$ into Eq.\eqref{eq:qn-genus-1-1}, we have
\begin{equation}\label{eq:qn-genus-1-2}
\begin{aligned}
q^{[n]}(n\chi, n\tau)&=\ii\left(\Im(a_3)-\Im(b_3)\right)\frac{\Theta\left(2A(\infty)+\ii n\varsigma d_3 \right)\Theta(0)}{\Theta\left(2A(\infty)\right)\Theta(-\ii n\varsigma d_3)}\ee^{2F_{31}J_{31}-2F_{30}}+\mathcal{O}(n^{-1})\\
&=\ii\left(\Im(a_3)-\Im(b_3)\right)\frac{\theta_{2}(\pi C+\frac{n\varsigma d_3}{2})\theta_3(0)}{\theta_{2}(\pi C)\theta_3\left(-\frac{n\varsigma d_3}{2}\right)}\ee^{\frac{\ii n\varsigma d_3}{2}+2F_{31}J_{31}-2F_{30}}+\mathcal{O}(n^{-1}),
\end{aligned}
\end{equation}
where $C$ is a new real constant defined as $$C:=\frac{1}{\oint_{\alpha_{11}}\frac{1}{\mathcal{R}_{3}(\xi)}d\xi}\left(\int_{a_3^*}^{\infty}\frac{1}{\mathcal{R}_{3}(\xi)}d\xi+\int_{a_3}^{\infty}\frac{1}{\mathcal{R}_{3}(\xi)}d\xi\right),$$
and the last relation in Eq.\eqref{eq:qn-genus-1-2} is given by using some shift properties of Jacobi-Theta function \cite{wang1989special}.

Moreover, $\left|q^{[n]}(n\chi, n\tau)\right|^2$ can be simplified further by the following calculation,
\begin{equation}\label{eq:qn-genus-1-3}
\begin{aligned}
\left|q^{[n]}(n\chi, n\tau)\right|^2&=\left(\Im(a_3)-\Im(b_3)\right)^2\frac{\theta_{2}\left(\pi C+\frac{n\varsigma d_3}{2}\right)\theta_2(\pi C-\frac{n\varsigma d_3 }{2})\theta_3^2(0)}{\theta_{2}^2\left(\pi C\right)\theta_3^2\left(-\frac{n\varsigma d_3 }{2}\right)}\ee^{\ii n \varsigma d_3+4\left|F_{31}J_{31}\right|^2}+\mathcal{O}(n^{-1})\\
&=\left(\Im(a_3)-\Im(b_3)\right)^2\left(1-\frac{\theta_4^2\left(\pi C\right)\theta_2^2(0)\theta_4^2(0)}{\theta_2^2(\pi C)\theta_3^4(0)}{\rm sd^2}(u,m)\right)\ee^{\ii n\varsigma d_3+4\left|F_{31}J_{31}\right|^2}+\mathcal{O}(n^{-1}),
\end{aligned}
\end{equation}
where $u:=-\frac{n\varsigma d_3}{2}\theta_{3}^2(0)$.

Next, we give a lemma to determine the constant $\frac{\theta_4^2(\pi C)}{\theta_2^2(\pi C)}$.

\begin{lemma}
\begin{equation}\label{eq:theta-simp}
\frac{\theta_{4}^2(\pi C)}{\theta_2^2(\pi C)}=-4\frac{\theta_{4}^2(0)}{\theta_{2}^2(0)}\frac{\frac{a_3^*-b_3^*}{a_3-b_3}}{\left(\frac{a_3^*-b_3^*}{a_3-b_3}-1\right)^2}.
\end{equation}
\begin{proof}
According to the shift relation between the Jacobi-Theta function, we convert this constant into another equivalent formula,
\begin{equation}\label{eq:Ainfty}
\frac{\theta_4^2(\pi C)}{\theta_{2}^2(\pi C)}=\frac{\theta_4^2\left(-\ii A(\infty)+\frac{\pi\varsigma}{2}\right)}{\theta_{2}^2\left(-\ii A(\infty)+\frac{\pi\varsigma}{2}\right)}=-\frac{\theta_1^2\left(2\pi \tilde{C}\right)}{\theta_3^2\left(2\pi\tilde{C}\right)},
\end{equation}
where $\tilde{C}=-\ii A(\infty):=\frac{1}{\oint_{\alpha_{11}}\frac{1}{\mathcal{R}(s)}ds}\int_{a^*_3}^{\infty}\frac{1}{\mathcal{R}(s)}ds$. Then we can use the relation between $\theta(2C)$ function and $\theta(C)$ function \cite{armitage2006elliptic} to convert this formula further,
\begin{equation}
\frac{\theta_4^2(\pi C)}{\theta_{2}^2(\pi C)}=-4\frac{\theta_4^2(0)}{\theta_2^2(0)}\left(\frac{\varrho}{\varrho^2-1}\right)^2,\,\,\,\, \varrho:=\frac{\theta_3(\pi\tilde{C})\theta_4(\pi\tilde{C})}{\theta_1(\pi\tilde{C})\theta_2(\pi\tilde{C})}.
\end{equation}
Next, we only need to determine the constant $\varrho$. To realize it, we can construct an auxiliary function $\varrho(k)$ on the Riemann surface with
\begin{equation}
\varrho(k):=\frac{\theta_3\left(\frac{\pi}{\oint_{\alpha_{11}}\frac{1}{\mathcal{R}(s)}ds}\int_{a_3^*}^{k}\frac{1}{\mathcal{R}(s)}ds\right)\theta_4\left(\frac{\pi}{\oint_{\alpha_{11}}\frac{1}{\mathcal{R}(s)}ds}\int_{a_3^*}^{k}\frac{1}{\mathcal{R}(s)}ds\right)}
{\theta_1\left(\frac{\pi}{\oint_{\alpha_{11}}\frac{1}{\mathcal{R}(s)}ds}\int_{a_3^*}^{k}\frac{1}{\mathcal{R}(s)}ds\right)\theta_2\left(\frac{\pi}{\oint_{\alpha_{11}}\frac{1}{\mathcal{R}(s)}ds}\int_{a_3^*}^{k}\frac{1}{\mathcal{R}(s)}ds\right)},
\end{equation}
then $\varrho$ equals to $\varrho=\varrho(\infty)$. With this definition, we know that $\varrho^2(k)$ is meromorphic on the Riemann surface. Then $\varrho(k)$ can be determined by its zeros and the poles uniquely. Now we recall the properties of the zeros about the Jacobi-Theta function,
\begin{equation}
\begin{aligned}
\theta_1(z)&=0,\quad z=n\pi +m\pi \varsigma,\,\,\,\,\, m,n\in \mathbb{Z}\\
\theta_{2}(z)&=0,\quad z=n\pi +m\pi \varsigma+\frac{1}{2}\pi,\\
\theta_{3}(z)&=0,\quad z=n\pi +m\pi \varsigma+\frac{1}{2}\pi+\frac{\pi}{2}\varsigma,\\
\theta_{4}(z)&=0,\quad z=n\pi +m\pi \varsigma+\frac{\pi}{2}\varsigma,
\end{aligned}
\end{equation}
thus $\varrho^2(k)$ can be set as
\begin{equation}
\varrho^2(k)=r^2\frac{\left(k-a_3\right)\left(k-b_3\right)}{\left(k-a_3^*\right)\left(k-b_3^*\right)},
\end{equation}
where $r^2$ is a constant to be determined. Moreover, by calculating the residue at the point $k=a_3^*$, we get
\begin{equation}
r^2=\frac{(a_3^*-b_3^*)}{(a_3^*-a_3)(a_3^*-b_3)}\underset{k=a_3^*}{\rm Res}\left(\varrho^2(k)\right)=\frac{\theta_3^2(0)\theta_4^2(0)\left(\oint_{\alpha_{11}}\frac{1}{\mathcal{R}(s)}ds\right)^2\left(a_3^*-b_3^*\right)^2}{4\pi^2\left(\theta_1'(0)\right)^2\theta_2^2(0)}.
\end{equation}
With a simple calculation, we have
\begin{equation}
r^2=\frac{a_3^*-b_3^*}{a_3-b_3},
\end{equation}
which infers that $\varrho=\varrho(\infty)=\left(\frac{a_3^*-b_3^*}{a_3-b_3}\right)^{1/2}$, it completes the proof.
\end{proof}
\end{lemma}
Substituting Eq.\eqref{eq:theta-simp} into Eq.\eqref{eq:qn-genus-1-3}, we can rewrite $\left|q^{[n]}(n\chi, n\tau)\right|^2$ as Eq.\eqref{eq:qn-genus-one-1} by using a shift property between the ${\rm sd}(u)$ function and the ${\rm cn}(u)$ function. Thus we complete the asymptotic analysis in the genus-one region (Theorem \ref{theo:genus-one}).

\section{The genus-three region}
\label{sec:genus-three}
In this section, we will give the asymptotic analysis for the genus-three region. Compared to the result in \cite{Bilman-JDE-2021}, the phase term $\vartheta(\lambda; \chi, \tau)$ in this paper contains two spectral parameters $\lambda_1=\ii$ and $\lambda_2=k\ii$, it is naturally to think that there might add a closed contour that involving another spectral parameter, which will match the genus-three region very well. To study this asymptotics, we firstly find a scalar RHP about $g_{4}(\lambda)\equiv g_4(\lambda; \chi, \tau)$ function.
\begin{rhp}
($g_4(\lambda)$-function in the genus-three region) For fixed $(\chi, \tau)$ in this region, there are unique contours $\Sigma_{g_5}^{\pm}, \Sigma_{g_4}^{\pm}, \Gamma_{g_5}^{\pm}, \Gamma_{g_4}$ and the integration constants $\kappa_4, \kappa_5, d_4, d_5$, this unique $g_{4}(\lambda)$ function satisfies the following conditions.
\begin{itemize}
\item {\bf Analyticity:} $g_4(\lambda)$ is analytic for $\lambda\in\mathbb{C}\setminus\left(\Sigma_{g_5}^{\pm}\cup \Sigma_{g_4}^{\pm}\cup \Gamma_{g_5}^{\pm}\cup \Gamma_{g_4}\right)$, where these contours are to be determined.
\item {\bf Jump condition:} The boundary values of $g_{4}(\lambda)$ on these contours are related by the following jump conditions,
\begin{equation}
\begin{aligned}
&g_{4,+}(\lambda)-g_{4,-}(\lambda)=d_4,\quad\lambda\in \Gamma_{g_4}, \quad g_{4,+}(\lambda)-g_{4,-}(\lambda)=d_5, \quad\lambda\in\Gamma_{g_5}^{\pm},\\
&g_{4,+}(\lambda)+g_{4,-}(\lambda)+2\vartheta(\lambda; \chi, \tau)=\kappa_j, \quad\lambda\in \Sigma_{g_{j}}^{\pm}\quad(j=4,5).
\end{aligned}
\end{equation}
\item {\bf Normalization:} As $\lambda\to\infty$, $g_{4}(\lambda)$ has the normalization condition
\begin{equation}
g(\lambda)=\mathcal{O}(\lambda^{-1}).
\end{equation}
\item {\bf Symmetry condition:} $g_{4}(\lambda)$ satisfies the symmetry relation $g_{4}(\lambda)=g_{4}(\lambda^*)^*$.
\end{itemize}
\end{rhp}
Suppose that there are eight complex numbers $a_4\equiv a_{4}(\chi, \tau), b_4\equiv b_{4}(\chi, \tau), c_4\equiv c_{4}(\chi, \tau), d_4\equiv d_{4}(\chi, \tau)$ and their conjugates, with $\Im(a_4)>\Im(b_4)>\Im(c_4)>\Im(d_4)$. Assume that $\Sigma_{g_5}^{+}$ is oriented from $b_4$ to $a_4$, $\Sigma_{g_4}^{+}$ is oriented from $d_4$ to $c_4$, $\Gamma_{g_5}^{+}$ is oriented from $c_4$ to $b_4$ ,$\Gamma_{g_4}$ is oriented from $d_4^*$ to $d_4$ and $\Sigma_{g_5}^{-}, \Sigma_{g_4}^{-}, \Gamma_{g_5}^{-}$ are the reflection of $\Sigma_{g_5}^{+}, \Sigma_{g_4}^{+}, \Gamma_{g_5}^{+}$ with the real axis respectively, where the orientation is upward. To solve $g_{4}(\lambda)$ in this RHP, we can introduce a function $\mathcal{R}_{4}(\lambda)$ defined as
\begin{equation}
\begin{aligned}
\mathcal{R}_{4}(\lambda)\equiv\mathcal{R}_{4}(\lambda; \chi, \tau)&=\sqrt{(\lambda-a_4)(\lambda-a_4^*)(\lambda-b_4)(\lambda-b_4^*)(\lambda-c_4)(\lambda-c_4^*)(\lambda-d_4)(\lambda-d_4^*)}\\
:&=\sqrt{\lambda^8-s_1\lambda^7+s_2\lambda^6-s_3\lambda^5+s_4\lambda^4-s_5\lambda^3+s_6\lambda^2-s_7\lambda+s_8}.
\end{aligned}
\end{equation}
For convenience, we take the derivative of $g_{4}(\lambda)$ with respect to $\lambda$ to eliminate the integration constants $\kappa_4, \kappa_5, d_4$, $d_5$, that is
\begin{equation}
g_{4,+}'(\lambda)+g_{4,-}'(\lambda)=-2\chi-4\lambda\tau-\frac{\ii}{\lambda-\ii}+\frac{\ii}{\lambda+\ii}-\frac{\ii}{\lambda-k\ii}+\frac{\ii}{\lambda+k\ii},\quad\lambda\in \Sigma_{g_4}^{\pm}\cup \Sigma_{g_5}^{\pm}.
\end{equation}
By dividing $\mathcal{R}_{4,+}(\lambda)$ on both sides, we have
\begin{equation}
\left(\frac{g_{4}'(\lambda)}{\mathcal{R}_{4}(\lambda)}\right)_{+}-\left(\frac{g_{4}'(\lambda)}{\mathcal{R}_{4}(\lambda)}\right)_{-}=-\frac{2\chi+4\lambda\tau+\frac{\ii}{\lambda-\ii}-\frac{\ii}{\lambda+\ii}+\frac{\ii}{\lambda-k\ii}-\frac{\ii}{\lambda+k\ii}}{\mathcal{R}_{4,+}(\lambda)},\quad\lambda\in \Sigma_{g_4}^{\pm}\cup \Sigma_{g_5}^{\pm},
\end{equation}
then $g'_{4}(\lambda)$ can be given by the Plemelj formula,
\begin{equation}
g'_{4}(\lambda)=\frac{\mathcal{R}_{4}(\lambda)}{2\pi\ii}\int_{\Sigma_{g_4}^{\pm}\cup \Sigma_{g_5}^{\pm}}-\frac{2\chi+4\xi\tau+\frac{\ii}{\xi-\ii}-\frac{\ii}{\xi+\ii}+\frac{\ii}{\xi-k\ii}-\frac{\ii}{\xi+k\ii}}{\mathcal{R}_{4,+}(\xi)(\xi-\lambda)}d\xi.
\end{equation}
Moreover, $g'_{4}(\lambda)$ can be given by the generalized residue theorem,
\begin{equation}
g'_{4}(\lambda)=\frac{\ii}{2}\mathcal{R}_{4}(\lambda)\left(\frac{1}{\mathcal{R}_{4}(\ii)(\lambda-\ii)}-\frac{1}{\mathcal{R}_{4}(-\ii)(\lambda+\ii)}+\frac{1}{\mathcal{R}_{4}(k\ii)(\lambda-k\ii)}-\frac{1}{\mathcal{R}_{4}(-k\ii)(\lambda+k\ii)}\right)-\vartheta'(\lambda; \chi, \tau).
\end{equation}
From the normalization condition of $g_4(\lambda)$ when $\lambda\to\infty$, we know that $g'_{4}(\lambda)\to\mathcal{O}(\lambda^{-2})$, which leads to five equations about the unknown parameters,
\begin{equation}
\begin{aligned}
\mathcal{O}(\lambda^3): &\frac{\ii}{2}\left(\frac{1}{\mathcal{R}_{4}(\ii)}-\frac{1}{\mathcal{R}_{4}(-\ii)}+\frac{1}{\mathcal{R}_{4}(k\ii)}-\frac{1}{\mathcal{R}_{4}(-k\ii)}\right)=0,\\
\mathcal{O}(\lambda^2):&-\frac{1}{2}\left(\frac{1}{\mathcal{R}_{4}(\ii)}+\frac{1}{\mathcal{R}_{4}(-\ii)}+\frac{k}{\mathcal{R}_{4}(k\ii)}+\frac{k}{\mathcal{R}_{4}(-k\ii)}\right)=0,\\
\mathcal{O}(\lambda):&-\frac{\ii}{2}\left(\frac{1}{\mathcal{R}_{4}(\ii)}-\frac{1}{\mathcal{R}_{4}(-\ii)}+\frac{k^2}{\mathcal{R}_{4}(k\ii)}-\frac{k^2}{\mathcal{R}_{4}(-k\ii)}-4\ii \tau\right)=0,\\
\mathcal{O}(\lambda^0): &\frac{1}{2}\left(\frac{1}{\mathcal{R}_{4}(\ii)}+\frac{1}{\mathcal{R}(-\ii)}+\frac{k^3}{\mathcal{R}_{4}(k\ii)}+\frac{k^3}{\mathcal{R}_{k^3}(-k\ii)}-2s_1\tau-2\chi\right)=0,\\
\mathcal{O}(\lambda^{-1}):&\frac{\ii}{2}\left(\frac{1}{\mathcal{R}_{4}(\ii)}-\frac{1}{\mathcal{R}_{4}(-\ii)}+\frac{k^4}{\mathcal{R}_{4}(k\ii)}-\frac{k^4}{\mathcal{R}_{4}(-k\ii)}+\left(\frac{3\ii s_1^2}{2}-2\ii s_2\right)\tau+\ii s_1\chi\right)=0.
\end{aligned}
\end{equation}
For fixed $\chi$ and $\tau$, we need to compute the unknown eight parameters, thus we need another three conditions to determine them. From the symmetry of $g_4(\lambda)$, we know that the integration constant $d_4$ is a real number. Moreover, we impose that the other three integration constants are all real numbers as the rest conditions for calculating the eight branch points.

By choosing one fixed $\chi$ and $\tau$, we give the contour plot of ${\Im}\left(h_4\left(\lambda\right)\equiv h_{4}(\lambda; \chi, \tau):=g_{4}(\lambda)+\vartheta(\lambda; \chi, \tau)\right)$ in Fig. \ref{fig:genus-three}.
\begin{figure}[ht]
\centering
\includegraphics[width=0.45\textwidth]{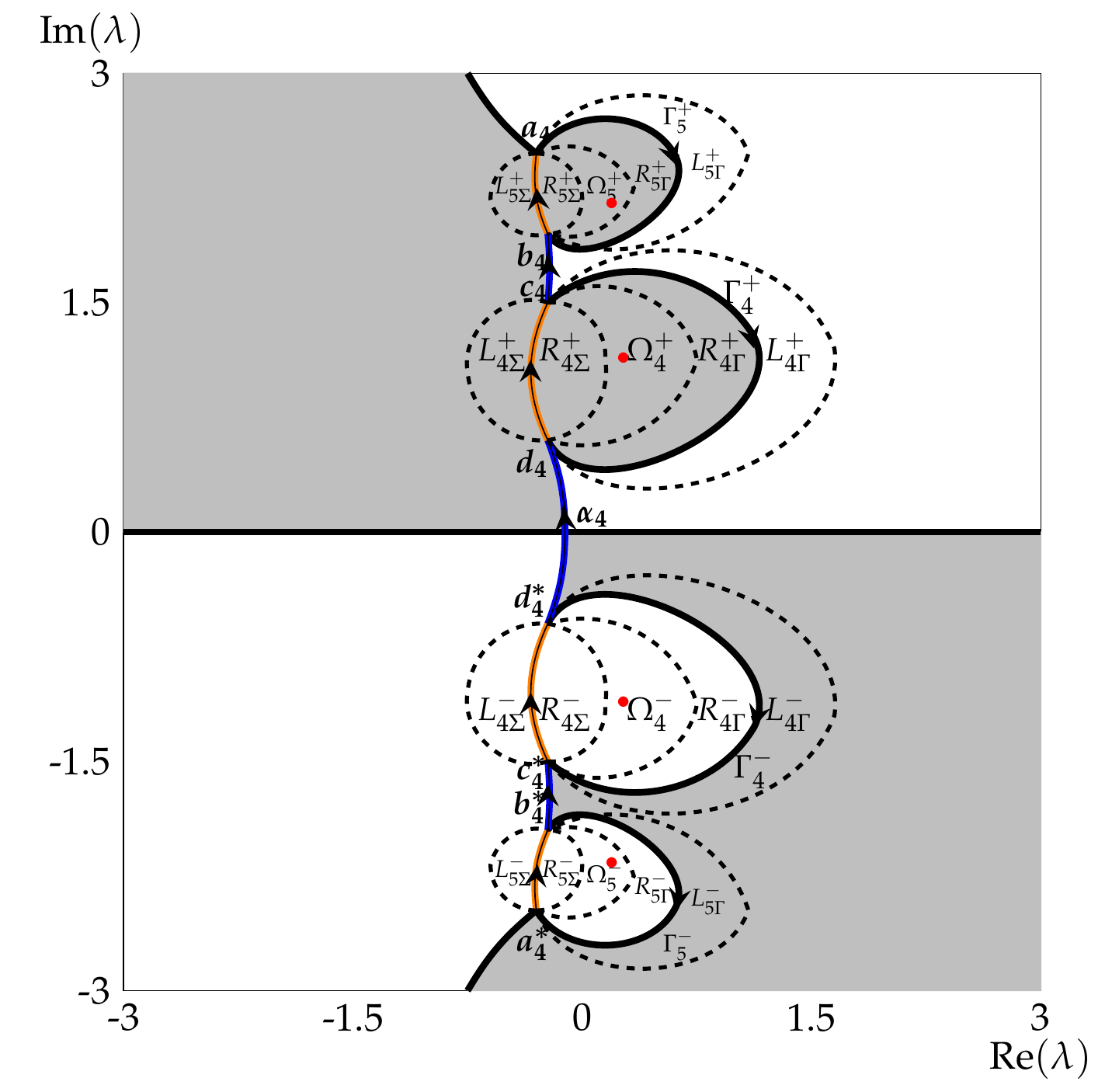}
\centering
\includegraphics[width=0.45\textwidth]{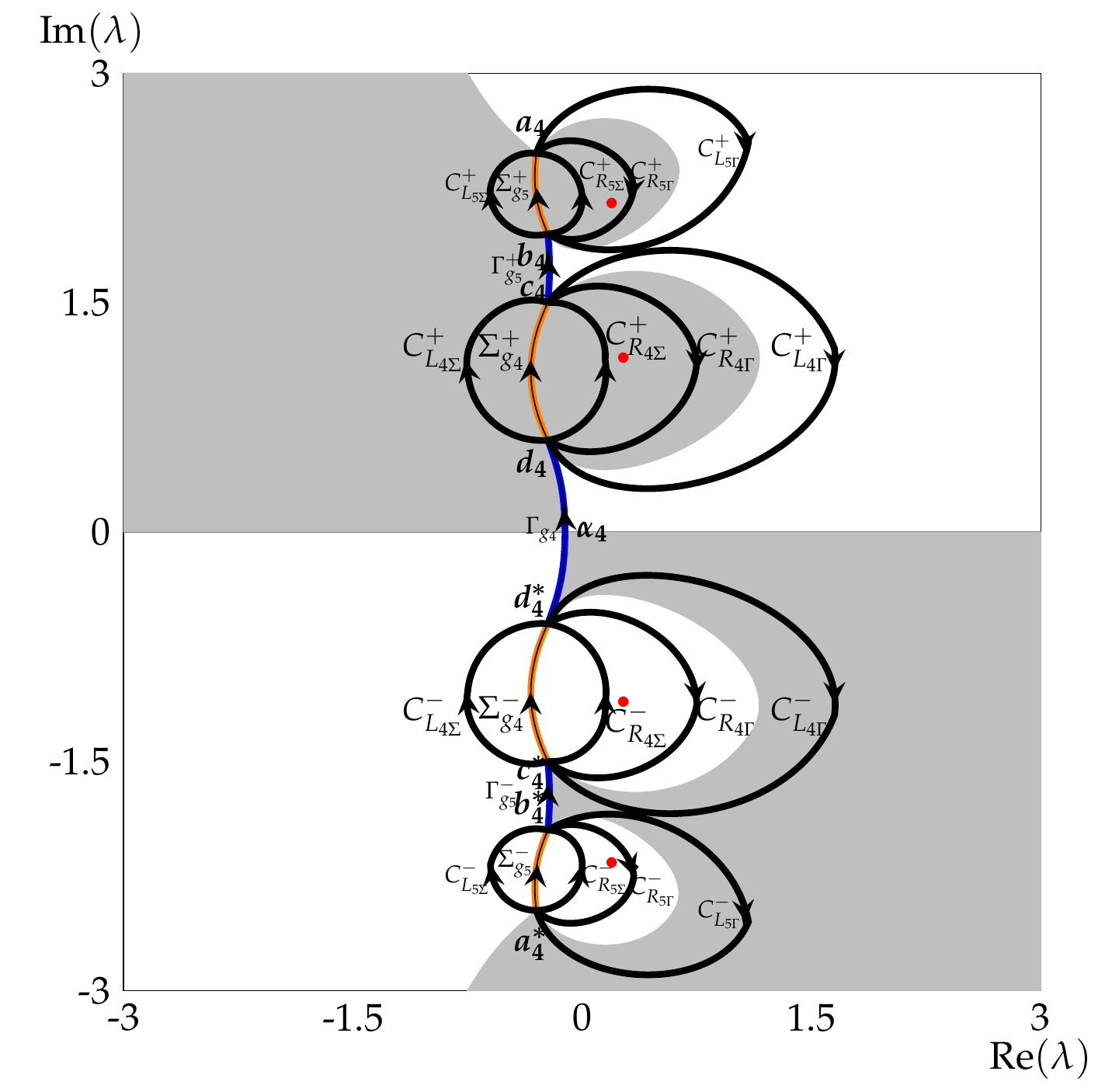}
\caption{The contour plot of ${\Im}\left(h_4\left(\lambda\right)\right)$ in the genus-three region with the parameters $\chi=\frac{11}{10}, \tau=1$ and $k=2$, where ${\Im}\left(h_4\left(\lambda\right)\right)<0$ (shaded) and ${\Im}\left(h_4\left(\lambda\right)\right)>0$ (unshaded), the red dots are the singularities $\lambda=\pm\ii$ and $\lambda=\pm 2\ii$. The left one gives the original contour and the right one is the corresponding contour deformation. }
\label{fig:genus-three}
\end{figure}

Similarly, we still analyze the asymptotics with the Deift-Zhou nonlinear steepest descent method.

Set
\begin{equation}
\mathbf{S}_{4}(\lambda; \chi, \tau):=\left\{\begin{aligned}&\mathbf{R}(\lambda; \chi, \tau)\ee^{-\ii n\vartheta(\lambda; \chi, \tau)\sigma_3}\mathbf{Q}_{c}\ee^{\ii n\vartheta(\lambda; \chi, \tau)\sigma_3},\quad &\lambda\in D_0\cap\left(D_{4}^{+}\cup D_{4}^{-}\right)^{c},\\
&\mathbf{R}(\lambda; \chi, \tau),\quad &\text{otherwise},
\end{aligned}\right.
\end{equation}
where $D_{4}^{+}=R_{4\Sigma}^{+}\cup \Omega_{4}^{+}\cup R_{4\Gamma}^{+}\cup R_{5\Sigma}^{+}\cup \Omega_{5}^{+}\cup R_{5\Gamma}^{+}, D_{4}^{-}=R_{4\Sigma}^{-}\cup \Omega_{4}^{-}\cup R_{4\Gamma}^{-}\cup R_{5\Sigma}^{-}\cup \Omega_{5}^{-}\cup R_{5\Gamma}^{-}.$
Then the jump contour of $\mathbf{S}_{4}(\lambda; \chi, \tau)$ transfers into the boundary of $D_4^{+}$ and $D_4^{-}$, which is suit for deforming the contour with the Deift-Zhou nonlinear steepest descent method. Set
\begin{equation}\label{eq:S-genus-three}
\begin{aligned}
\mathbf{T}_{4}(\lambda; \chi, \tau):&=\mathbf{S}_{4}(\lambda; \chi, \tau)\ee^{-\ii n\vartheta(\lambda; \chi, \tau)\sigma_3}
\left(\mathbf{Q}_{R}^{[2]}\right)^{-1}\ee^{\ii n\vartheta(\lambda; \chi, \tau)\sigma_3}\ee^{\ii n g_{4}(\lambda)\sigma_3},\quad &\lambda\in L_{4\Gamma}^{+}\cup L_{5\Gamma}^{+},\\
\mathbf{T}_{4}(\lambda; \chi, \tau):&=\mathbf{S}_{4}(\lambda; \chi, \tau)\mathbf{Q}_{L}^{[2]}\ee^{-\ii n\vartheta(\lambda; \chi, \tau)\sigma_3}
\mathbf{Q}_{C}^{[2]}\ee^{\ii n\vartheta(\lambda; \chi, \tau)\sigma_3}\ee^{\ii n g_{4}(\lambda)\sigma_3},\quad &\lambda\in R_{4\Gamma}^{+}\cup R_{5\Gamma}^{+},\\
\mathbf{T}_{4}(\lambda; \chi, \tau):&=\mathbf{S}_{4}(\lambda; \chi, \tau)\mathbf{Q}_{L}^{[2]}\ee^{\ii n g_{4}(\lambda)\sigma_3},\quad &\lambda\in \Omega_{4}^{+}\cup \Omega_{5}^{+},\\
\mathbf{T}_{4}(\lambda; \chi, \tau):&=\mathbf{S}_{4}(\lambda; \chi, \tau)\mathbf{Q}_{L}^{[2]}\ee^{-\ii n\vartheta(\lambda; \chi, \tau)\sigma_3}\mathbf{Q}_{L}^{[3]}\ee^{\ii n\vartheta(\lambda; \chi, \tau)\sigma_3}\ee^{\ii n g_{4}(\lambda)\sigma_3},\quad &\lambda\in R_{4\Sigma}^{+}\cup R_{5\Sigma}^{+},\\
\mathbf{T}_{4}(\lambda; \chi, \tau):&=\mathbf{S}_{4}(\lambda; \chi, \tau)\ee^{-\ii n\vartheta(\lambda; \chi, \tau)\sigma_3}
\left(\mathbf{Q}_{R}^{[3]}\right)^{-1}\ee^{\ii n\vartheta(\lambda; \chi, \tau)\sigma_3}\ee^{\ii n g_{4}(\lambda)\sigma_3},\quad &\lambda\in L_{4\Sigma}^{+}\cup L_{5\Sigma}^{+},\\
\mathbf{T}_{4}(\lambda; \chi, \tau):&=\mathbf{S}_{4}(\lambda; \chi, \tau)\ee^{-\ii n\vartheta(\lambda; \chi, \tau)\sigma_3}
\left(\mathbf{Q}_{R}^{[1]}\right)^{-1}\ee^{\ii n\vartheta(\lambda; \chi, \tau)\sigma_3}\ee^{\ii n g_{4}(\lambda)\sigma_3},\quad &\lambda\in L_{4\Gamma}^{-}\cup L_{5\Gamma}^{-},\\
\mathbf{T}_{4}(\lambda; \chi, \tau):&=\mathbf{S}_{4}(\lambda; \chi, \tau)\mathbf{Q}_{L}^{[1]}\ee^{-\ii n\vartheta(\lambda; \chi, \tau)\sigma_3}
\mathbf{Q}_{C}^{[1]}
\ee^{\ii n\vartheta(\lambda; \chi, \tau)\sigma_3}\ee^{\ii n g_{4}(\lambda)\sigma_3},\quad &\lambda\in R_{4\Gamma}^{-}\cup R_{5\Gamma}^{-},\\
\mathbf{T}_{4}(\lambda; \chi, \tau):&=\mathbf{S}_{4}(\lambda; \chi, \tau)\mathbf{Q}_{L}^{[1]}\ee^{\ii n g_{4}(\lambda)\sigma_3},\quad &\lambda\in \Omega_{4}^{-}\cup \Omega_{5}^{-},\\
\mathbf{T}_{4}(\lambda; \chi, \tau):&=\mathbf{S}_{4}(\lambda; \chi, \tau)\mathbf{Q}_{L}^{[1]}\ee^{-\ii n\vartheta(\lambda; \chi, \tau)\sigma_3}
\mathbf{Q}_{L}^{[4]}\ee^{\ii n\vartheta(\lambda; \chi, \tau)\sigma_3}\ee^{\ii n g_{4}(\lambda)\sigma_3},\quad &\lambda\in R_{4\Sigma}^{-}\cup R_{5\Sigma}^{-},\\
\mathbf{T}_{4}(\lambda; \chi, \tau):&=\mathbf{S}_{4}(\lambda; \chi, \tau)\ee^{-\ii n\vartheta(\lambda; \chi, \tau)\sigma_3}
\left(\mathbf{Q}_{R}^{[4]}\right)^{-1}\ee^{\ii n\vartheta(\lambda; \chi, \tau)\sigma_3}\ee^{\ii n g_{4}(\lambda)\sigma_3},\quad &\lambda\in L_{4\Sigma}^{-}\cup L_{5\Sigma}^{-},
\end{aligned}
\end{equation}
in other regions, we set $\mathbf{T}_{4}(\lambda; \chi, \tau):=\mathbf{S}_{4}(\lambda; \chi, \tau)\ee^{\ii n g_{4}(\lambda)\sigma_3}$. With a simple calculation, the jump conditions about $\mathbf{T}_{4}(\lambda; \chi, \tau)$ change into
\begin{equation}\label{eq:jump-genus-three}
\begin{aligned}
\mathbf{T}_{4, +}(\lambda; \chi, \tau)&=\mathbf{T}_{4, -}(\lambda; \chi, \tau)\ee^{-\ii nh_4(\lambda)\sigma_3}
\mathbf{Q}_{R}^{[2]}\ee^{\ii nh_4(\lambda)\sigma_3},\quad &\lambda\in C_{L_{4\Gamma}}^{+}\cup C_{L_{5\Gamma}}^{+},\\
\mathbf{T}_{4, +}(\lambda; \chi, \tau)&=\mathbf{T}_{4, -}(\lambda; \chi, \tau)\ee^{-\ii nh_4(\lambda)\sigma_3}\mathbf{Q}_{C}^{[2]}\ee^{\ii nh_4(\lambda)\sigma_3},\quad &\lambda\in C_{R_{4\Gamma}}^{+}\cup C_{R_{5\Gamma}}^{+},\\
\mathbf{T}_{4, +}(\lambda; \chi, \tau)&=\mathbf{T}_{4, -}(\lambda; \chi, \tau)\ee^{-\ii nh_4(\lambda)\sigma_3}\mathbf{Q}_{L}^{[3]}\ee^{\ii nh_4(\lambda)\sigma_3},\quad &\lambda\in C_{R_{4\Sigma}}^{+}\cup C_{R_{5\Sigma}}^{+},\\
\mathbf{T}_{4, +}(\lambda; \chi, \tau)&=\mathbf{T}_{4, -}(\lambda; \chi, \tau)\ee^{-\ii nh_4(\lambda)\sigma_3}\mathbf{Q}_{R}^{[3]}\ee^{\ii nh_4(\lambda)\sigma_3},\quad &\lambda\in C_{L_{4\Sigma}}^{+}\cup C_{L_{5\Sigma}}^{+},\\
\mathbf{T}_{4, +}(\lambda; \chi, \tau)&=\mathbf{T}_{4, -}(\lambda; \chi, \tau)\ee^{-\ii nh_4(\lambda)\sigma_3}\mathbf{Q}_{R}^{[1]}\ee^{\ii nh_4(\lambda)\sigma_3},\quad &\lambda\in C_{L_{4\Gamma}}^{-}\cup C_{L_{5\Gamma}}^{-},\\
\mathbf{T}_{4, +}(\lambda; \chi, \tau)&=\mathbf{T}_{4, -}(\lambda; \chi, \tau)\ee^{-\ii nh_4(\lambda)\sigma_3}\mathbf{Q}_{C}^{[1]}\ee^{\ii nh_4(\lambda)\sigma_3},\quad &\lambda\in C_{R_{4\Gamma}}^{-}\cup C_{R_{5\Gamma}}^{-},\\
\mathbf{T}_{4, +}(\lambda; \chi, \tau)&=\mathbf{T}_{4, -}(\lambda; \chi, \tau)\ee^{-\ii nh_4(\lambda)\sigma_3}\mathbf{Q}_{L}^{[4]}\ee^{\ii nh_4(\lambda)\sigma_3},\quad &\lambda\in C_{R_{4\Sigma}}^{-}\cup C_{R_{5\Sigma}}^{-},\\
\mathbf{T}_{4, +}(\lambda; \chi, \tau)&=\mathbf{T}_{4, -}(\lambda; \chi, \tau)\ee^{-\ii nh_4(\lambda)\sigma_3}\mathbf{Q}_{R}^{[4]}\ee^{\ii nh_4(\lambda)\sigma_3},\quad &\lambda\in C_{L_{4\Sigma}}^{-}\cup C_{L_{5\Sigma}}^{-},\\
\mathbf{T}_{4, +}(\lambda; \chi, \tau)&=\mathbf{T}_{4, -}(\lambda; \chi, \tau)\begin{bmatrix}0&\ee^{-\ii n\kappa_4}\\
-\ee^{\ii n\kappa_4}&0
\end{bmatrix},&\lambda\in \Sigma_{g_4}^{+}\cup \Sigma_{g_4}^{-},\\
\mathbf{T}_{4, +}(\lambda; \chi, \tau)&=\mathbf{T}_{4, -}(\lambda; \chi, \tau)\begin{bmatrix}0&\ee^{-\ii n\kappa_5}\\
-\ee^{\ii n\kappa_5}&0
\end{bmatrix},&\lambda\in \Sigma_{g_5}^{+}\cup \Sigma_{g_5}^{-},\\
\mathbf{T}_{4, +}(\lambda; \chi, \tau)&=\mathbf{T}_{4, -}(\lambda; \chi, \tau)\begin{bmatrix}\ee^{\ii nd_5}&0\\
0&\ee^{-\ii nd_5}
\end{bmatrix},&\lambda\in \Gamma_{g_5}^{+}\cup \Gamma_{g_5}^{-},\\
\mathbf{T}_{4, +}(\lambda; \chi, \tau)&=\mathbf{T}_{4, -}(\lambda; \chi, \tau)\begin{bmatrix}\ee^{\ii nd_4}&0\\
0&\ee^{-\ii nd_4}
\end{bmatrix},&\lambda\in \Gamma_{g_4}.
\end{aligned}
\end{equation}
When $n$ is large, the jump conditions in Eq.\eqref{eq:jump-genus-three} will tend to the identity matrix exponentially except for the contours $\Sigma_{g_4}^{\pm}\cup \Sigma_{g_5}^{\pm}\cup \Gamma_{g_5}^{\pm}$ and $\Gamma_{g_4}$. In the next subsection, we will give the parametrix construction about $\mathbf{T}_{4}(\lambda; \chi, \tau)$ relying on the constant jump matrix in Eq.\eqref{eq:jump-genus-three}.
\subsection{Parametrix construction for $\mathbf{T}_{4}(\lambda; \chi, \tau)$}
Similar to the analysis in the genus-one region, we can construct the outer parametrix $\dot{\mathbf{T}}_{4}^{\rm out}(\lambda; \chi, \tau)$ satisfying the following RHP.
\begin{rhp}
(RHP for the outer parametrix $\dot{\mathbf{T}}^{\rm out}_{4}(\lambda; \chi, \tau)$) Seek $2\times 2$ matrix $\dot{\mathbf{T}}_{4}^{\rm out}(\lambda; \chi, \tau)$ with the following conditions.
\begin{itemize}
\item {\bf Analyticity:} $\dot{\mathbf{T}}_{4}^{\rm out}(\lambda; \chi, \tau)$ is analytic in $\lambda\in \mathbb{C}\setminus\left(\Sigma_{g_5}^{\pm}\cup \Sigma_{g_4}^{\pm}\cup \Gamma_{g_5}^{\pm}\cup \Gamma_{g_4}\right)$.
\item {\bf Jump condition:} $\dot{\mathbf{T}}_{4}^{\rm out}(\lambda; \chi, \tau)$ takes continuous boundary values on $\left(\Sigma_{g_5}^{\pm}\cup \Sigma_{g_4}^{\pm}\cup \Gamma_{g_5}^{\pm}\cup \Gamma_{g_4}\right)$, and it is related by the jump condition $\dot{\mathbf{T}}_{4,+}^{\rm out}(\lambda; \chi, \tau)=\dot{\mathbf{T}}_{4,-}^{\rm out}(\lambda; \chi, \tau)\mathbf{V}_{\dot{\mathbf{T}}_{4}^{\rm out}}(\lambda; \chi, \tau)$, where $\mathbf{V}_{\dot{\mathbf{T}}_{4}^{\rm out}}(\lambda; \chi, \tau)$ equals to
    \begin{equation}
    \mathbf{V}_{\dot{\mathbf{T}}_{4}^{\rm out}}(\lambda; \chi, \tau)=\left\{\begin{aligned}&\begin{bmatrix}0&\ee^{-\ii n\kappa_4}\\
-\ee^{\ii n\kappa_4}&0
\end{bmatrix},&\lambda\in \Sigma_{g_4}^{+}\cup\Sigma_{g_4}^{-},\\
&\begin{bmatrix}0&\ee^{-\ii n\kappa_5}\\
-\ee^{\ii n\kappa_5}&0
\end{bmatrix},&\lambda\in \Sigma_{g_5}^{+}\cup\Sigma_{g_5}^{-},\\
&\begin{bmatrix}\ee^{\ii nd_5}&0\\
0&\ee^{-\ii nd_5}
\end{bmatrix},&\lambda\in \Gamma_{g_5}^{+}\cup \Gamma_{g_5}^{-},\\
&\begin{bmatrix}\ee^{\ii nd_4}&0\\
0&\ee^{-\ii nd_4}
\end{bmatrix},&\lambda\in \Gamma_{g_4}.
    \end{aligned}\right.
    \end{equation}
\item {\bf Normalization:} When $\lambda\to\infty$, $\dot{\mathbf{T}}_{4}^{\rm out}(\lambda; \chi, \tau)$ has the following normalization condition,
\begin{equation}
\dot{\mathbf{T}}_{4}^{\rm out}(\lambda; \chi, \tau)\to\mathbb{I}.
\end{equation}
\end{itemize}
\end{rhp}
Compared to the genus-one region, there are four jump contours more than the ones to the outer parametrix in the genus-one region, we will introduce a scalar function $F_{4}(\lambda; \chi, \tau)$ with the following jump conditions,
\begin{equation}
\begin{aligned}
&F_{4,+}(\lambda; \chi, \tau)+F_{4,-}(\lambda; \chi, \tau)=\ii n\kappa_4,\quad\lambda\in \Sigma_{g_4}^{+}\cup \Sigma_{g_4}^{-},\\
&F_{4,+}(\lambda; \chi, \tau)+F_{4,-}(\lambda; \chi, \tau)=\ii n\kappa_5,\quad\lambda\in \Sigma_{g_5}^{+}\cup \Sigma_{g_5}^{-},\\
&F_{4,+}(\lambda; \chi, \tau)-F_{4,-}(\lambda; \chi, \tau)=\ii nd_5, \quad\lambda\in\Gamma_{g_5}^{+}\cup \Gamma_{g_5}^-,\\
&F_{4,+}(\lambda; \chi, \tau)-F_{4,-}(\lambda; \chi, \tau)=\ii nd_4, \quad\lambda\in\Gamma_{g_4}.
\end{aligned}
\end{equation}
By using the function $\mathcal{R}_{4}(\lambda)$, $F_{4}(\lambda; \chi, \tau)$ can be expressed as
\begin{multline}\label{eq:F4}
F_{4}(\lambda;\chi,\tau)=\frac{\mathcal{R}_{4}(\lambda)}{2\pi\ii}\Bigg[\int_{\Sigma_{g_4}^{+}\cup \Sigma_{g_4}^{-}}\frac{\ii n\kappa_4}{\mathcal{R}_{4}(\xi)(\xi-\lambda)}d\xi+\int_{\Sigma_{g_5}^{+}\cup \Sigma_{g_5}^{-}}\frac{\ii n\kappa_5}{\mathcal{R}_{4}(\xi)(\xi-\lambda)}d\xi+\int_{\Gamma_{g_4}}\frac{\ii nd_4}{\mathcal{R}_{4}(\xi)(\xi-\lambda)}d\xi\\+\int_{\Gamma_{g_5}^{+}\cup \Gamma_{g_5}^-}\frac{\ii nd_5}{\mathcal{R}_{4}(\xi)(\xi-\lambda)}d\xi\Bigg].
\end{multline}
Expanding $F_{4}(\lambda; \chi, \tau)$ as $\lambda\to\infty$, we have the following series expansion,
\begin{equation}
F_{4}(\lambda; \chi, \tau)=F_{43}\lambda^3+F_{42}\lambda^2+F_{41}\lambda+F_{40}+\mathcal{O}(\lambda^{-1}),
\end{equation}
where
\begin{equation}\label{eq:F4342}
\begin{aligned}
F_{43}&=-\frac{1}{2\pi\ii}\left(\int_{\Sigma_{g_4}^{+}\cup \Sigma_{g_4}^{-}}\frac{\ii n\kappa_4}{\mathcal{R}_{4}(\xi)}d\xi+\int_{\Sigma_{g_5}^{+}\cup \Sigma_{g_5}^{-}}\frac{\ii n\kappa_5}{\mathcal{R}_{4}(\xi)}d\xi+\int_{\Gamma_{g_5}^{+}\cup \Gamma_{g_5}^-}\frac{\ii nd_5}{\mathcal{R}_{4}(\xi)}d\xi+\int_{\Gamma_{g_4}}\frac{\ii nd_4}{\mathcal{R}_{4}(\xi)}d\xi\right),\\
F_{42}&=-\frac{1}{2\pi\ii}\left(\int_{\Sigma_{g_4}^{+}\cup \Sigma_{g_4}^{-}}\frac{\ii n\kappa_4}{\mathcal{R}_{4}(\xi)}\xi d\xi+\int_{\Sigma_{g_5}^{+}\cup \Sigma_{g_5}^{-}}\frac{\ii n\kappa_5}{\mathcal{R}_{4}(\xi)}\xi d\xi+\int_{\Gamma_{g_5}^{+}\cup \Gamma_{g_5}^-}\frac{\ii nd_5}{\mathcal{R}_{4}(\xi)}\xi d\xi+\int_{\Gamma_{g_4}}\frac{\ii nd_4}{\mathcal{R}_{4}(\xi)}\xi d\xi\right)\\&-\frac{s_1}{2}F_{43},\\
F_{41}&=-\frac{1}{2\pi\ii}\left(\int_{\Sigma_{g_4}^{+}\cup \Sigma_{g_4}^{-}}\frac{\ii n\kappa_4}{\mathcal{R}_{4}(\xi)}\xi^2 d\xi+\int_{\Sigma_{g_5}^{+}\cup \Sigma_{g_5}^{-}}\frac{\ii n\kappa_5}{\mathcal{R}_{4}(\xi)}\xi^2 d\xi+\int_{\Gamma_{g_5}^{+}\cup \Gamma_{g_5}^-}\frac{\ii nd_5}{\mathcal{R}_{4}(\xi)}\xi^2 d\xi+\int_{\Gamma_{g_4}}\frac{\ii nd_4}{\mathcal{R}_{4}(\xi)}\xi^2 d\xi\right)\\&
-\frac{s_1}{2}F_{42}+\left(\frac{s_2}{2}-\frac{3s_1^2}{8}\right)F_{43},\\
F_{40}&=-\frac{1}{2\pi\ii}\left(\int_{\Sigma_{g_4}^{+}\cup \Sigma_{g_4}^{-}}\frac{\ii n\kappa_4}{\mathcal{R}_{4}(\xi)}\xi^3 d\xi+\int_{\Sigma_{g_5}^{+}\cup \Sigma_{g_5}^{-}}\frac{\ii n\kappa_5}{\mathcal{R}_{4}(\xi)}\xi^3 d\xi+\int_{\Gamma_{g_5}^{+}\cup \Gamma_{g_5}^-}\frac{\ii nd_5}{\mathcal{R}_{4}(\xi)}\xi^3 d\xi+\int_{\Gamma_{g_4}}\frac{\ii nd_4}{\mathcal{R}_{4}(\xi)}\xi^3 d\xi\right)\\&
-\frac{s_1}{2}F_{41}+\left(\frac{s_2}{2}-\frac{3s_1^2}{8}\right)F_{42}-\left(\frac{s_3}{2}-\frac{3s_1s_2}{4}+\frac{5s_1^3}{16}\right)F_{43}.
\end{aligned}
\end{equation}

With the aid of function $F_{4}(\lambda; \chi, \tau)$, we can redefine a new matrix $\mathbf{P}_{4}(\lambda; \chi, \tau)$ as
\begin{equation}
\mathbf{P}_{4}(\lambda; \chi, \tau)={\rm diag}\left(\ee^{F_{40}}, \ee^{-F_{40}}\right)\dot{\mathbf{T}}_{4}^{\rm out}(\lambda; \chi, \tau){\rm diag}\left(\ee^{-F_{4}(\lambda; \chi, \tau)}, \ee^{F_{4}(\lambda; \chi, \tau)}\right),
\end{equation}
which satisfies a constant jump matrix on $\lambda\in\Sigma_{g_4}^{\pm}\cup\Sigma_{g_5}^{\pm}$,
\begin{equation}
\mathbf{P}_{4,+}(\lambda; \chi, \tau)=\mathbf{P}_{4,-}(\lambda; \chi, \tau)(\ii\sigma_2),\quad \lambda\in\Sigma_{g_4}^{\pm}\cup\Sigma_{g_5}^{\pm}.
\end{equation}
Moreover, when $\lambda\to\infty$, $\mathbf{P}_{4}(\lambda; \chi, \tau)$ has the following normalization condition,
\begin{equation}\label{eq:P4-bc}
\mathbf{P}_{4}(\lambda; \chi, \tau){\rm diag}\left(\ee^{F_{43}\lambda^3+F_{42}\lambda^2+F_{41}\lambda}, \ee^{-F_{43}\lambda^3-F_{42}\lambda^2-F_{41}\lambda}\right)\to\mathbb{I}.
\end{equation}
Compared to the genus-one region, in this region, we should introduce a genus-three Riemann surface $\Sigma_{3}$ with two sheets $\Sigma_{31}$ and $\Sigma_{32}$, and the corresponding basis can be set as $\alpha_{31}, \beta_{31}, \alpha_{32}, \beta_{32}, \alpha_{33},\beta_{33},$ which is shown in Fig.\ref{circle:genus-three}.
\begin{figure}[ht]
\centering
\includegraphics[width=0.3\textwidth]{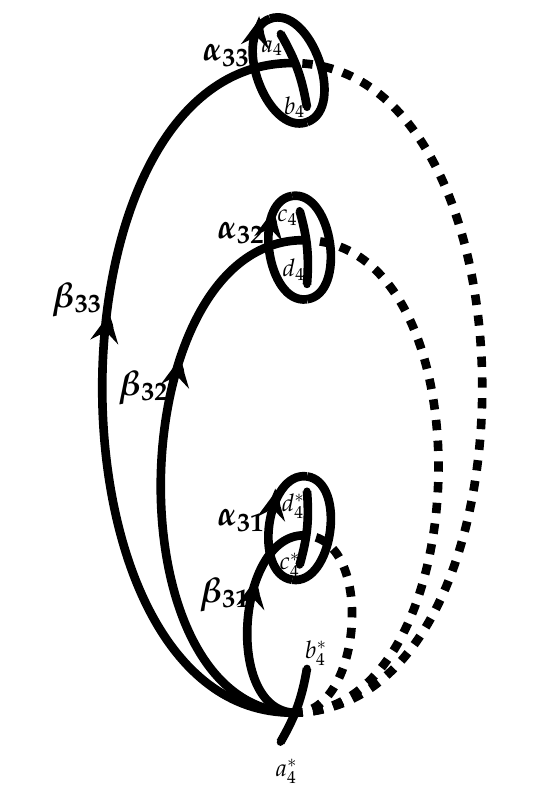}
\caption{Homology basis for the Riemann surface of genus-three. The solid paths indicate the first sheet and the dashed lines lie in the second sheet.}
\label{circle:genus-three}
\end{figure}
Similarly, we introduce the Abel integrals for the genus-three region,
\begin{equation}
\omega_{j}(\lambda)=\int_{a_{4}^*}^{\lambda}\psi_{j}(\xi)d\xi,\quad j=1,2,3,
\end{equation}
where $\psi_{j}(\xi)$ is defined as
\begin{equation}\label{eq:cji}
\psi_{j}(\xi):=\frac{\sum\limits_{i=1}^{3}c_{ji}\xi^{3-i}}{\mathcal{R}_{4}(\xi)},
\end{equation}
and $d\omega_{j}(\mathcal{P})$ is a holomorphic differential on the Riemann surface $\pmb{\chi}$. The coefficients $c_{ji}$s are uniquely determined with the following normalization conditions,
\begin{equation}
\int_{\alpha_{3l}}d\omega_{j}(\mathcal{P})=2\pi\ii\delta_{jl},\quad j,l=1,2,3.
\end{equation}
Correspondingly, the normalized holomorphic differential can give the $b$-period matrix $B_{jl}$ as
\begin{equation}
B_{jl}=\int_{\beta_{3l}}d\omega_{j}(\mathcal{P}),
\end{equation}
which is a symmetric matrix and its real part is negatively defined.
We use them to define a lattice $\Lambda_{3}$ in $\mathbb{C}^{3}$ as
\begin{equation}
\Lambda_{3}=\{2\pi\ii N+\mathbf{B}M, N, M\in\mathbb{Z}^{3}\}.
\end{equation}
By this definition,  the Jacobian variety of $\pmb{\chi}$ is the complex torus $Jac\{\pmb{\chi}\}=\mathbb{C}^{3}/\Lambda_3$. Then we can define the Abel mapping $\mathbf{A}:\pmb{\chi}\to Jac\{\pmb{\chi}\}$ as
\begin{equation}\label{eq:abel-map}
A_{j}(\mathcal{P})=\int_{\mathcal{P}_{0}}^{\mathcal{P}}d\omega_j(\mathcal{Q}),\quad j=1,2,3,
\end{equation}
where the point $\mathcal{P}_{0}$ is given from the base point $a_4^*$ by the condition $\pi(\mathcal{P}_{0})=a_4^*$ and $\mathcal{Q}$ is the integration variable. Especially, for the integral divisors $\mathcal{D}=\mathcal{P}_{1}+\mathcal{P}_{2}+\mathcal{P}_{3}$, we have its Abel mapping as
\begin{equation}
\mathbf{A}(\mathcal{D})=\mathbf{A}(\mathcal{P}_1)+\mathbf{A}(\mathcal{P}_2)+\mathbf{A}(\mathcal{P}_3),
\end{equation}
which is useful for the later asymptotic analysis.

Moreover, the Abel integrals $\mathbf{A}(\lambda)$ has the following properties,
\begin{equation}
\begin{aligned}
\mathbf{A}_{+}(\lambda)-\mathbf{A}_{-}(\lambda)&=0\quad \mod 2\pi\ii \mathbb{Z}^3, \quad \lambda\in \left(b_4^*, c_4^*\right)\cup\left(d_4^*, d_4\right)\cup \left(c_4, b_4\right),\\
\mathbf{A}_{+}(\lambda)+\mathbf{A}_{-}(\lambda)&=\mathbf{B}\mathbf{e}_1\mod 2\pi\ii \mathbb{Z}^3, \quad \lambda\in \left(c_4^*, d_4^*\right),\\
\mathbf{A}_{+}(\lambda)+\mathbf{A}_{-}(\lambda)&=\mathbf{B}\mathbf{e}_2\mod 2\pi\ii \mathbb{Z}^3, \quad \lambda\in \left(d_4, c_4\right),\\
\mathbf{A}_{+}(\lambda)+\mathbf{A}_{-}(\lambda)&=\mathbf{B}\mathbf{e}_3\mod 2\pi\ii \mathbb{Z}^3, \quad \lambda\in \left(b_4, a_4\right).
\end{aligned}
\end{equation}
Next, we define the Abel integrals with the singularities at the point $P_{\infty}$,
\begin{equation}
\begin{aligned}
\Omega_j(\lambda)=\int_{a_4^*}^{\lambda}\Psi_j(\xi)d\xi,\quad j=1,2,3,
\end{aligned}
\end{equation}
where
\begin{equation}
\Psi_j(\xi)=\frac{\sum\limits_{i=1}^{7}s_{ji}\xi^{7-i}}{\mathcal{R}_{4}(\xi)},
\end{equation}
and the unknown coefficients $s_{ji}$s are determined by the normalization conditions
\begin{equation}
\begin{aligned}
&\Omega_1(\lambda)\to \lambda+\mathcal{O}(1),\quad \Omega_2(\lambda)\to\lambda^2+\mathcal{O}(1),\quad \Omega_{3}(\lambda)\to\lambda^3+\mathcal{O}(1), \quad P\to P_{\infty},\\
&\oint_{\alpha_{3l}}d\Omega_{j}(\mathcal{P})=0,\quad j,l=1,2,3.
\end{aligned}
\end{equation}
Correspondingly, the $b$-periods of these integrals are set as
\begin{equation}\label{eq:UVW}
\mathcal{U}_{j}=\int_{\beta_{3j}}d\Omega_{1}(\mathcal{P}), \quad \mathcal{V}_{j}=\int_{\beta_{3j}}d\Omega_{2}(\mathcal{P}), \quad \mathcal{W}_{j}=\int_{\beta_{3j}}d\Omega_{3}(\mathcal{P}),\quad j=1,2,3.
\end{equation}
From the normalization conditions of $\Omega_{j}(\lambda)$, we know that the limit $J_{4j}\,(j=1,2,3)$ defined as
\begin{equation}\label{eq:J-cons}
J_{41}:=\lim\limits_{\lambda\to\infty}\int_{a_4^*}^{\lambda}d\Omega_{1}(\mathcal{P})-\lambda,\quad J_{42}:=\lim\limits_{\lambda\to\infty}\int_{a_4^*}^{\lambda}d\Omega_{2}(\mathcal{P})-\lambda^2,\quad
J_{43}:=\lim\limits_{\lambda\to\infty}\int_{a_4^*}^{\lambda}d\Omega_{3}(\mathcal{P})-\lambda^3
\end{equation}
are existent. Similar to the genus-one region, the solution of $\mathbf{P}_{4}(\lambda; \chi, \tau)$ can be given by using the Riemann-Theta function in the high genus case with the definition \ref{prop:theta}.

Consider the matrix $\pmb{\mathcal{P}}_4(\lambda; \chi, \tau)$ as
\begin{equation}
\pmb{\mathcal{P}}_4(\lambda; \chi, \tau):=\begin{bmatrix}\frac{\Theta\left(\mathbf{A}(\lambda)+\mathbf{d}-\pmb{\mathcal{U}}F_{41}-\pmb{\mathcal{V}}F_{42}-\pmb{\mathcal{W}}F_{43}\right)}{\Theta\left(\mathbf{A}(\lambda)+\mathbf{d}\right)}&
\frac{\Theta\left(\mathbf{A}(\lambda)-\mathbf{d}+\pmb{\mathcal{U}}F_{41}+\pmb{\mathcal{V}}F_{42}+\pmb{\mathcal{W}}F_{43}\right)}{\Theta\left(\mathbf{A}(\lambda)-\mathbf{d}\right)}\\
\frac{\Theta\left(\mathbf{A}(\lambda)-\mathbf{d}-\pmb{\mathcal{U}}F_{41}-\pmb{\mathcal{V}}F_{42}-\pmb{\mathcal{W}}F_{43}\right)}{\Theta\left(\mathbf{A}(\lambda)-\mathbf{d}\right)}&
\frac{\Theta\left(\mathbf{A}(\lambda)+\mathbf{d}+\pmb{\mathcal{U}}F_{41}+\pmb{\mathcal{V}}F_{42}+\pmb{\mathcal{W}}F_{43}\right)}{\Theta\left(\mathbf{A}(\lambda)+\mathbf{d}\right)}\end{bmatrix}
\ee^{-\left(\Omega_1(\lambda)F_{41}+\Omega_2(\lambda)F_{42}+\Omega_3(\lambda)F_{43}\right)\sigma_3},
\end{equation}
then we have
\begin{equation}
\pmb{\mathcal{P}}_{4,+}(\lambda; \chi, \tau)=\pmb{\mathcal{P}}_{4,-}(\lambda; \chi, \tau)\begin{bmatrix}0&1\\
1&0
\end{bmatrix},\quad\lambda\in\Sigma_{g_4}^{\pm}\cup\Sigma_{g_5}^{\pm}.
\end{equation}
Then the solution of $\mathbf{P}_{4}(\lambda; \chi, \tau)$ can be given as
\begin{equation}
\mathbf{P}_{4}(\lambda; \chi, \tau):=\frac{1}{2}{\rm diag}\left(C_{41}, C_{42}\right)\begin{bmatrix}\left(\gamma_{4}(\lambda)+\frac{1}{\gamma_{4}(\lambda)}\right)\pmb{\mathcal{P}}_{4}(\lambda; \chi, \tau)_{11}&-\ii \left(\gamma_{4}(\lambda)-\frac{1}{\gamma_{4}(\lambda)}\right)\pmb{\mathcal{P}}_{4}(\lambda; \chi, \tau)_{12}\\
\ii \left(\gamma_{4}(\lambda)-\frac{1}{\gamma_{4}(\lambda)}\right)\pmb{\mathcal{P}}_{4}(\lambda; \chi, \tau)_{21}&\left(\gamma_{4}(\lambda)+\frac{1}{\gamma_{4}(\lambda)}\right)\pmb{\mathcal{P}}_{4}(\lambda; \chi, \tau)_{22}
\end{bmatrix},
\end{equation}
where
$C_{41}$ and $C_{42}$ are defined as
\begin{equation}
\begin{aligned}
C_{41}&=\frac{\Theta\left(\mathbf{A}(\infty)+\mathbf{d}\right)}{\Theta\left(\mathbf{A}(\infty)+\mathbf{d}-\pmb{\mathcal{U}}F_{41}-\pmb{\mathcal{V}}F_{42}-\pmb{\mathcal{W}}F_{43}\right)}\ee^{J_{41}F_{41}+J_{42}F_{42}+J_{43}F_{43}},\\
C_{42}&=\frac{\Theta\left(\mathbf{A}(\infty)+\mathbf{d}\right)}{\Theta\left(\mathbf{A}(\infty)+\mathbf{d}+\pmb{\mathcal{U}}F_{41}+\pmb{\mathcal{V}}F_{42}+\pmb{\mathcal{W}}F_{43}\right)}\ee^{-J_{41}F_{41}-J_{42}F_{42}-J_{43}F_{43}},\\
\end{aligned}
\end{equation}
and $\gamma_4(\lambda)=\left(\frac{(\lambda-a_4^*)(\lambda-b_4)(\lambda-c_4^*)(\lambda-d_4)}{(\lambda-a_4)(\lambda-b_4^*)(\lambda-c_4)(\lambda-d_4^*)}\right)^{\frac{1}{4}}$ satisfies $\gamma_{4,+}=\ii\gamma_{4,-}$. And $\gamma_{4}-\frac{1}{\gamma_{4}}$ has three zeros $\mathcal{P}_{1}, \mathcal{P}_{2}, \mathcal{P}_{3}$ at the first sheet $\Sigma_{31}$. Following the idea in \cite{belokolos1994algebro,kotlyarov2017planar}, the constant $\mathbf{d}$ can be given as
\begin{equation}\label{eq:d}
\mathbf{d}=\mathbf{A}\left(\mathcal{D}\right)+\mathbf{K},
\end{equation}
where $\mathbf{K}$ is the Riemann-Theta constant vector, modulo the lattice $\Lambda_3$, with
\begin{equation}
K_{j}=\frac{2\pi\ii+B_{jj}}{2}-\frac{1}{2\pi\ii}\sum\limits_{l=1,l\neq j}^{3}\int_{\alpha_{3l}}\left(\int_{\mathcal{P}_{0}}^{\mathcal{Q}}\omega_{j}\right)\omega_l.
\end{equation}
Especially, in the hyperelliptic case, $K_j$ can be reduced into
\begin{equation}
K_j=\frac{1}{2}\sum\limits_{l=1}^{3}B_{lj}+\pi\ii\left(j-2\right).
\end{equation}
Under this condition, $\mathbf{P}_{4}(\lambda; \chi, \tau)$ is well defined in the complex plane. Then the outer parametrix has been constructed completely.
Similar to the genus-zero and genus-one region, in the neighborhood of $a_{4}, a_{4}^*, b_{4}, b_{4}^*, c_{4}, c_4^*, d_4$ and $d_{4}^*$, we can set the local parametrix as $\dot{\mathbf{T}}_{4}^{a_4,a_4^*,b_4,b_4^*,c_4,c_4^*,d_4,d_4^*}(\lambda; \chi, \tau), $ the corresponding small disks are set as $D_{a_4,a_4^*,b_4,b_4^*,c_4,c_4^*,d_4,d_4^*}(\delta)$. Then the global paramtrix of $\mathbf{T}_{4}(\lambda; \chi, \tau)$ is
\begin{equation}
\dot{\mathbf{T}}_{4}(\lambda; \chi, \tau):=\left\{\begin{aligned}&\dot{\mathbf{T}}_{4}^{a_4}(\lambda; \chi, \tau),\quad\lambda\in D_{a_4}(\delta),\\&\dot{\mathbf{T}}_{4}^{a_4^*}(\lambda; \chi, \tau),\quad\lambda\in D_{a_4^*}(\delta),\\&\dot{\mathbf{T}}_{4}^{b_4}(\lambda; \chi, \tau),\quad\lambda\in D_{b_4}(\delta),\\
&\dot{\mathbf{T}}_{4}^{b_4^*}(\lambda; \chi, \tau),\quad\lambda\in D_{b_4^*}(\delta),\\
&\dot{\mathbf{T}}_{4}^{c_4}(\lambda; \chi, \tau),\quad\lambda\in D_{c_4}(\delta),\\&\dot{\mathbf{T}}_{4}^{c_4^*}(\lambda; \chi, \tau),\quad\lambda\in D_{c_4^*}(\delta),\\&\dot{\mathbf{T}}_{4}^{d_4}(\lambda; \chi, \tau),\quad\lambda\in D_{d_4}(\delta),\\
&\dot{\mathbf{T}}_{4}^{d_4^*}(\lambda; \chi, \tau),\quad\lambda\in D_{d_4^*}(\delta),\\
&\dot{\mathbf{T}}_{4}^{\rm out}(\lambda; \chi, \tau),\quad\lambda\in\mathbb{C}\setminus\left(\overline{D_{a_4,a_4*,b_4,b_4^*,c_4,c_4^*,d_4,d_4^*}(\delta)}\cup \Sigma_{g_4}^{\pm}\cup \Sigma_{g_5}^{\pm}\cup \Gamma_{g_5}^{\pm}\cup\Gamma_{g_4}\right).
\end{aligned}\right.
\end{equation}
Next, we will give the error analysis between $\mathbf{T}_{4}(\lambda; \chi, \tau)$ and its parametrix $\dot{\mathbf{T}}_{4}(\lambda; \chi, \tau)$.
\subsection{Error analysis}
Set the error function between $\mathbf{T}_{4}(\lambda; \chi, \tau)$ and $\dot{\mathbf{T}}_{4}(\lambda; \chi, \tau)$ as
\begin{equation}
\mathcal{E}_{4}(\lambda; \chi, \tau):=\mathbf{T}_{4}(\lambda; \chi, \tau)\left(\dot{\mathbf{T}}_{4}(\lambda; \chi, \tau)\right)^{-1}.
\end{equation}
Then the jump condition about $\mathcal{E}_{4}(\lambda; \chi, \tau)$ can be set as $\mathbf{V}_{\mathcal{E}_{4}}(\lambda; \chi, \tau)$. Similar to the genus-one region, the jump matrix $\mathbf{V}_{\mathcal{E}_{4}}(\lambda; \chi, \tau)$ has the error formula,
\begin{equation}
\begin{aligned}
\|\mathbf{V}_{\mathcal{E}_{4}}(\lambda; \chi, \tau)-\mathbb{I}\|&=\mathcal{O}\left(\ee^{-\mu_4 n}\right)\,\,(\mu_4>0),\, \lambda\in C_{L_{j\Sigma}}^{\pm}\cup C_{R_{j\Sigma}}^{\pm}\cup C_{R_{j\Gamma}}^{\pm}\cup C_{L_{j\Gamma}}^{\pm}\,\,(j=4,5),\\
\|\mathbf{V}_{\mathcal{E}_{4}}(\lambda; \chi, \tau)-\mathbb{I}\|&=\mathcal{O}(n^{-1}),\quad \lambda\in\partial D_{a_4,a_4^*,b_4,b_4^*,c_4,c_4^*,d_4,d_4^*}(\delta).\\
\end{aligned}
\end{equation}
The potential function $q^{[n]}(n\chi, n\tau)$ in the genus-three region can be given by
\begin{equation}\label{eq:qn-genus-3}
\begin{aligned}
q^{[n]}(n\chi, n\tau)&=2\ii\lim\limits_{\lambda\to\infty}\lambda\mathbf{T}_{4}(\lambda; \chi, \tau)_{12}\\
&=2\ii\lim\limits_{\lambda\to\infty}\lambda\left(\mathcal{E}_{4}(\lambda; \chi, \tau)\dot{\mathbf{T}}^{\rm out}_{4}(\lambda; \chi, \tau)\right)_{12}\\
&=2\ii\lim\limits_{\lambda\to\infty}\lambda\left(\mathcal{E}_{4,11}(\lambda; \chi, \tau)\dot{\mathbf{T}}_{4,12}^{\rm out}(\lambda; \chi, \tau)+\mathcal{E}_{4,12}(\lambda; \chi, \tau)\dot{\mathbf{T}}_{4,22}^{\rm out}(\lambda; \chi, \tau)\right)\\
&=2\ii\lim\limits_{\lambda\to\infty}\lambda\dot{\mathbf{T}}_{4,12}^{\rm out}(\lambda; \chi, \tau)+\mathcal{O}(n^{-1}).
\end{aligned}
\end{equation}
Substituting $\dot{\mathbf{T}}_4^{\rm out}(\lambda; \chi, \tau)$ into Eq.\eqref{eq:qn-genus-3}, then we get the asymptotics Eq.\eqref{eq:qn-genus-3-1} for the genus-three region (Theorem \ref{theo:genus-three}).

%
%
%
\appendix

\titleformat{\section}[display]
{\centering\LARGE\bfseries}{ }{11pt}{\LARGE}
\renewcommand{\appendixname}{Appendix \ \Alph{section} }
\setcounter{equation}{0}
\setcounter{equation}{0}
\renewcommand\theequation{\Alph{section}.\arabic{equation}}
\setcounter{definition}{0}
\renewcommand\thedefinition{\Alph{section}.\arabic{definition}}

\setcounter{lemma}{0}
\renewcommand\thelemma{\Alph{section}.\arabic{lemma}}
\section{\appendixname: Exact solution of RHP \ref{rhp-airy} in terms of Airy function}
\label{app:Airy}
In this appendix, we will build the solution of RHP \ref{rhp-airy} in an explicit formula. Firstly, we can introduce a normalization transformation $\mathbf{W}(\zeta)=\pmb{\mathcal{U}}(\zeta)\ee^{\frac{\zeta^{3/2}}{2}\sigma_3}$ to convert the jump matrices into the following constant ones, that is
\begin{equation}\label{eq:jump-W}
\mathbf{W}_{+}(\zeta)=\mathbf{W}_{-}(\zeta)=\begin{bmatrix}1&0\\-1&1
\end{bmatrix},\quad\arg(\zeta)=0,\quad \mathbf{W}_{+}(\zeta)=\mathbf{W}_{-}(\zeta)=\begin{bmatrix}0&1\\-1&0
\end{bmatrix},\quad\arg(-\zeta)=0,
\end{equation}
and
\begin{equation}\label{eq:jump-W-1}
\mathbf{W}_{+}(\zeta)=\mathbf{W}_{-}(\zeta)\begin{bmatrix}1&-1\\
0&1
\end{bmatrix},\quad \arg(\zeta)=\pm\frac{2}{3}\pi.
\end{equation}
The normalization condition about $\mathbf{W}(\zeta)$ will change into
\begin{equation}
\lim\limits_{\zeta\to\infty}\mathbf{W}\ee^{-\frac{\zeta^{\frac{3}{2}}}{2}\sigma_3}\ee^{-\frac{\ii \pi}{4}\sigma_3}\mathbf{Q}_{c}^{-1}\zeta^{-\frac{1}{4}\sigma_3}\to\mathbb{I}.
\end{equation}
Then we will find the solution $\mathbf{W}(\zeta)$ by the special function. Since the jump matrices about $\mathbf{W}(\zeta)$ are constants in all directions, and $\frac{d\mathbf{W}(\zeta)}{d\zeta}$ also satisfies the exact jump condition with $\mathbf{W}(\zeta)$. As a result, $\mathbf{Y}(\zeta):=\frac{d\mathbf{W}(\zeta)}{d\zeta}\mathbf{W}(\zeta)^{-1}$ is an entire function in the complex $\zeta$ plane. Suppose that there exists a constant matrix $\mathbf{B}=\begin{bmatrix}b_{11}&b_{12}\\
b_{21}&b_{22}
\end{bmatrix}$ such that the normalization condition have the following formula,
\begin{equation}\label{eq:nor-W}
\mathbf{W}(\zeta)\ee^{-\frac{\zeta^{\frac{3}{2}}}{2}\sigma_3}\ee^{-\frac{\ii\pi}{4}\sigma_3}\mathbf{Q}_c^{-1}\zeta^{-\frac{1}{4}\sigma_3}=\mathbb{I}+\mathbf{B}\zeta^{-1}+\mathcal{O}(\zeta^{-2}).
\end{equation}
Then the entire function $\mathbf{Y}(\zeta)$ can be given as
\begin{equation}
\mathbf{Y}(\zeta)=\frac{3}{4}\begin{bmatrix}b_{21}&b_{22}-b_{11}-\zeta\\
-1&-b_{21}
\end{bmatrix}.
\end{equation}
By a direct calculation, we know that the second row of $\mathbf{W}(\zeta)$ satisfies the Airy function by using a scale transformation, that is
\begin{equation}
\frac{d^2 w_{2}}{d\eta^2}-\eta w_2=0,\quad \eta=\left(\frac{3}{4}\right)^{\frac{2}{3}}\left(\zeta-c\right), \quad c:=-b_{21}^2-b_{11}+b_{22}.
\end{equation}
When $\eta$ is large, we have the following asymptotic expansion about the Airy solution,
\begin{equation}\label{eq:asy-airy}
\begin{aligned}
{\rm Ai}(\eta)&=\frac{\ee^{-\frac{2}{3}\eta^{\frac{3}{2}}}}{2\sqrt{\pi}\zeta^{\frac{1}{4}}}\sum\limits_{k=0}^{\infty}\frac{(-1)^ku_k}{\left(\frac{2}{3}\eta^{\frac{3}{2}}\right)^k}, \quad \eta\to\infty, \quad \left|\arg(\eta)\right|\leq \pi-\delta, \\ {\rm Ai}'(\eta)&=-\frac{\zeta^{\frac{1}{4}}\ee^{-\frac{2}{3}\eta^{\frac{3}{2}}}}{2\sqrt{\pi}}\sum\limits_{k=0}^{\infty}\frac{(-1)^k(6k+1)u_k}{(1-6k)\left(\frac{2}{3}\eta^{\frac{3}{2}}\right)^k},
\end{aligned}
\end{equation}
where $\delta$ is an enough small positive constant and $u_k$ has a recurrence relation
\begin{equation}
u_{0}=1,\quad u_{k}=\frac{(6k-5)(6k-3)(6k-1)}{216(2k-1)k}u_{k-1}.
\end{equation}
Let us first consider this solution in the sector $0<\arg(\zeta)<\frac{2}{3}\pi$. From this asymptotics in Eq.\eqref{eq:asy-airy}, we can choose the basic solution as $w_2={\rm Ai}(\eta)$ and $w_2={\rm Ai}\left(\eta\ee^{-\frac{2\ii\pi}{3}}\right)$, which can be expanded by Eq.\eqref{eq:asy-airy} directly in this sector. Then the solution of the second row of $\mathbf{W}(\zeta)$ can be assumed as
\begin{equation}\label{eq:airy-W}
\begin{bmatrix}
W_{21}(\zeta)&W_{22}(\zeta)
\end{bmatrix}=\begin{bmatrix}a_{1}{\rm Ai}(\eta)+b_1{\rm Ai}\left(\eta\ee^{-\frac{2\ii\pi}{3}}\right)&a_{2}{\rm Ai}(\eta)+b_2{\rm Ai}\left(\eta\ee^{-\frac{2\ii\pi}{3}}\right)
\end{bmatrix},
\end{equation}
where $a_1, b_1, a_2, b_2$ are some constants to be determined. Expanding Eq.\eqref{eq:airy-W} at $\zeta\to\infty$ and comparing it with the asymptotics of $W_{21}(\zeta)$ and $W_{22}(\zeta)$, we can calculate the following relations,
\begin{equation}\label{eq:rela-b}
\begin{aligned}
&\frac{b_1}{2\left(\frac{4}{3}\right)^{-\frac{1}{6}}\zeta^{\frac{1}{4}}\ee^{-\frac{\ii\pi}{6}}\sqrt{\pi}}\ee^{\frac{\zeta^{\frac{3}{2}}}{2}}=-\frac{\sqrt{2}\ee^{\frac{\ii\pi}{4}}}{2\zeta^{\frac{1}{4}}}\ee^{\frac{\zeta^{\frac{3}{2}}}{2}},\quad
\frac{a_2}{2\left(\frac{4}{3}\right)^{-\frac{1}{6}}\zeta^{\frac{1}{4}}\sqrt{\pi}}\ee^{-\frac{\zeta^{\frac{3}{2}}}{2}}=\frac{\sqrt{2}\ee^{-\frac{\ii\pi}{4}}}{2\zeta^{\frac{1}{4}}}\ee^{-\frac{\zeta^{\frac{3}{2}}}{2}},\quad \\&c=0,\quad b_{21}=0, \quad b_{22}=0.
\end{aligned}
\end{equation}
Thus we have
\begin{equation}
\begin{bmatrix}
W_{21}(\zeta)&W_{22}(\zeta)
\end{bmatrix}=\begin{bmatrix}\sqrt{2\pi}\left(\frac{3}{4}\right)^{\frac{1}{6}}\ee^{-\frac{11}{12}\ii\pi}{\rm Ai}\left(\left(\frac{3}{4}\right)^{\frac{2}{3}}\zeta\ee^{-\frac{2}{3}\ii\pi}\right)&\sqrt{2\pi}\left(\frac{3}{4}\right)^{\frac{1}{6}}\ee^{-\frac{1}{4}\ii\pi}{\rm Ai}\left(\left(\frac{3}{4}\right)^{\frac{2}{3}}\zeta\right)
\end{bmatrix}.
\end{equation}
Furthermore, from the differential equation $\frac{d\mathbf{W}(\zeta)}{d\zeta}=\mathbf{Y}(\zeta)\mathbf{W}(\zeta) $ and the relation in Eq.\eqref{eq:rela-b}, we get
\begin{equation}
\begin{aligned}
W_{11}(\zeta)&=-\frac{4}{3}W_{21,\zeta}=\sqrt{2\pi}\left(\frac{4}{3}\right)^{\frac{1}{6}}\ee^{-\frac{7}{12}\ii\pi}{\rm Ai}'\left(\left(\frac{3}{4}\right)^{\frac{2}{3}}\zeta\ee^{-\frac{2}{3}\ii\pi}\right),\\
W_{12}(\zeta)&=-\frac{4}{3}W_{22,\zeta}=\sqrt{2\pi}\left(\frac{4}{3}\right)^{\frac{1}{6}}\ee^{\frac{3}{4}\ii\pi}{\rm Ai}'\left(\left(\frac{3}{4}\right)^{\frac{2}{3}}\zeta\right).
\end{aligned}
\end{equation}
Thus the solution of $\mathbf{W}(\zeta)$ in the sector $0<\arg(\zeta)<\frac{2}{3}\pi$ can be given as
\begin{equation}
\mathbf{W}(\zeta)=\begin{bmatrix}\sqrt{2\pi}\left(\frac{4}{3}\right)^{\frac{1}{6}}\ee^{-\frac{7}{12}\ii\pi}{\rm Ai}'\left(\left(\frac{3}{4}\right)^{\frac{2}{3}}\zeta\ee^{-\frac{2}{3}\ii\pi}\right)&\sqrt{2\pi}\left(\frac{4}{3}\right)^{\frac{1}{6}}\ee^{\frac{3}{4}\ii\pi}{\rm Ai}'\left(\left(\frac{3}{4}\right)^{\frac{2}{3}}\zeta\right)\\
\sqrt{2\pi}\left(\frac{3}{4}\right)^{\frac{1}{6}}\ee^{-\frac{11}{12}\ii\pi}{\rm Ai}\left(\left(\frac{3}{4}\right)^{\frac{2}{3}}\zeta\ee^{-\frac{2}{3}\ii\pi}\right)&\sqrt{2\pi}\left(\frac{3}{4}\right)^{\frac{1}{6}}\ee^{-\frac{1}{4}\ii\pi}{\rm Ai}\left(\left(\frac{3}{4}\right)^{\frac{2}{3}}\zeta\right)
\end{bmatrix}.
\end{equation}
For the other sectors, we can get the solution of $\mathbf{W}(\zeta)$ with the jump conditions Eq.\eqref{eq:jump-W} and Eq.\eqref{eq:jump-W-1},
\begin{equation}
\begin{aligned}
&\mathbf{W}(\zeta)\\
&{=}\left\{\begin{aligned}&\begin{bmatrix}\sqrt{2\pi}\left(\frac{4}{3}\right)^{\frac{1}{6}}\ee^{-\frac{7}{12}\ii\pi}{\rm Ai}'\left(\left(\frac{3}{4}\right)^{\frac{2}{3}}\zeta\ee^{-\frac{2}{3}\ii\pi}\right)&\sqrt{2\pi}\left(\frac{4}{3}\right)^{\frac{1}{6}}\ee^{-\frac{11}{12}\ii\pi}{\rm Ai}'\left(\left(\frac{3}{4}\right)^{\frac{2}{3}}\zeta\ee^{-\frac{4}{3}\ii\pi}\right)\\
\sqrt{2\pi}\left(\frac{3}{4}\right)^{\frac{1}{6}}\ee^{-\frac{11}{12}\ii\pi}{\rm Ai}\left(\left(\frac{3}{4}\right)^{\frac{2}{3}}\zeta\ee^{-\frac{2}{3}\ii\pi}\right)&\sqrt{2\pi}\left(\frac{3}{4}\right)^{\frac{1}{6}}\ee^{-\frac{7}{12}\ii\pi}{\rm Ai}\left(\left(\frac{3}{4}\right)^{\frac{2}{3}}\zeta\ee^{-\frac{4}{3}\ii\pi}\right)
\end{bmatrix},\,\, \frac{2}{3}\pi<\arg(\zeta)<\pi,\\
&\begin{bmatrix}\sqrt{2\pi}\left(\frac{4}{3}\right)^{\frac{1}{6}}\ee^{-\frac{11}{12}\ii\pi}{\rm Ai}'\left(\left(\frac{3}{4}\right)^{\frac{2}{3}}\zeta\ee^{\frac{2}{3}\ii\pi}\right)&\sqrt{2\pi}\left(\frac{4}{3}\right)^{\frac{1}{6}}\ee^{\frac{5}{12}\ii\pi}{\rm Ai}'\left(\left(\frac{3}{4}\right)^{\frac{2}{3}}\zeta\ee^{\frac{4}{3}\ii\pi}\right)\\
\sqrt{2\pi}\left(\frac{3}{4}\right)^{\frac{1}{6}}\ee^{-\frac{7}{12}\ii\pi}{\rm Ai}\left(\left(\frac{3}{4}\right)^{\frac{2}{3}}\zeta\ee^{\frac{2}{3}\ii\pi}\right)&\sqrt{2\pi}\left(\frac{3}{4}\right)^{\frac{1}{6}}\ee^{\frac{1}{12}\ii\pi}{\rm Ai}\left(\left(\frac{3}{4}\right)^{\frac{2}{3}}\zeta\ee^{\frac{4}{3}\ii\pi}\right)
\end{bmatrix},\,\, -\pi<\arg(\zeta)<-\frac{2}{3}\pi,\\
&\begin{bmatrix}\sqrt{2\pi}\left(\frac{4}{3}\right)^{\frac{1}{6}}\ee^{-\frac{11}{12}\ii\pi}{\rm Ai}'\left(\left(\frac{3}{4}\right)^{\frac{2}{3}}\zeta\ee^{\frac{2}{3}\ii\pi}\right)&\sqrt{2\pi}\left(\frac{4}{3}\right)^{\frac{1}{6}}\ee^{\frac{3}{4}\ii\pi}{\rm Ai}'\left(\left(\frac{3}{4}\right)^{\frac{2}{3}}\zeta\right)\\
\sqrt{2\pi}\left(\frac{3}{4}\right)^{\frac{1}{6}}\ee^{-\frac{7}{12}\ii\pi}{\rm Ai}\left(\left(\frac{3}{4}\right)^{\frac{2}{3}}\zeta\ee^{\frac{2}{3}\ii\pi}\right)&\sqrt{2\pi}\left(\frac{3}{4}\right)^{\frac{1}{6}}\ee^{-\frac{1}{4}\ii\pi}{\rm Ai}\left(\left(\frac{3}{4}\right)^{\frac{2}{3}}\zeta\right)
\end{bmatrix},\,\, -\frac{2}{3}\pi<\arg(\zeta)<0.
\end{aligned}
\right.
\end{aligned}
\end{equation}
We can easily check that the solution $\mathbf{W}(\zeta)$ can match the normalization condition Eq.\eqref{eq:nor-W} in all directions as $\zeta\to\infty$, then we get the solution of $\mathbf{W}(\zeta)$ in the whole $\zeta$ plane.
\section*{Acknowledgements}

Liming Ling is supported by the National Natural Science Foundation of China (Grant Nos. 12122105); Xiaoen Zhang is supported by the National Natural Science Foundation of China (Grant No.12101246), the China Postdoctoral Science Foundation (Grant No. 2020M682692), the Guangzhou Science and Technology Program of China(Grant No. 202102020783).
\bibliographystyle{siam}
\bibliography{reference}

\begin{thebibliography}{10}

\bibitem{armitage2006elliptic}
{\sc J.~V. Armitage and W.~F. Eberlein}, {\em Elliptic functions}, vol.~67,
  Cambridge University Press, 2006.

\bibitem{belokolos1994algebro}
{\sc E.~D. Belokolos, A.~I. Bobenko, V.~Z. Enolskii, A.~R. Its, and V.~B.
  Matveev}, {\em Algebro-geometric approach to nonlinear integrable equations},
  vol.~550, Springer, 1994.

\bibitem{bertola2021soliton}
{\sc M.~Bertola, T.~Grava, and G.~Orsatti}, {\em Soliton shielding of the
  focusing nonlinear schr\"{o}dinger equation}, arXiv preprint
  arXiv:2112.05985,  (2021).

\bibitem{Bilman-JNS-2019}
{\sc D.~Bilman and R.~Buckingham}, {\em Large-order asymptotics for
  multiple-pole solitons of the focusing nonlinear schr{\"o}dinger equation},
  J. Nonlinear Sci., 29 (2019), pp.~2185--2229.

\bibitem{Bilman-JDE-2021}
{\sc D.~Bilman, R.~Buckingham, and D.~S. Wang}, {\em Far-field asymptotics for
  multiple-pole solitons in the large-order limit}, J. Differential Equations,
  297 (2021), pp.~320--369.

\bibitem{Bilman-Duke-2019}
{\sc D.~Bilman, L.~Ling, and P.~D. Miller}, {\em Extreme superposition: Rogue
  waves of infinite order and the painlev{\'e}-iii hierarchy}, Duke Math. J.,
  169 (2020), pp.~671--760.

\bibitem{Bilman-2019-CPAM}
{\sc D.~Bilman and P.~D. Miller}, {\em A robust inverse scattering transform
  for the focusing nonlinear schr\"{o}dinger equation}, Comm. Pure Appl. Math.,
  72 (2019), pp.~1722--1805.

\bibitem{BILMAN2022133289}
\leavevmode\vrule height 2pt depth -1.6pt width 23pt, {\em Broader universality
  of rogue waves of infinite order}, Phys. D, 435 (2022), p.~133289.

\bibitem{Bilman-arxiv-2021}
\leavevmode\vrule height 2pt depth -1.6pt width 23pt, {\em Extreme
  superposition: high-order fundamental rogue waves in the far-field regime},
  To appear in Mem. Amer. Math. Soc.,  (2022).

\bibitem{biondini2017long}
{\sc G.~Biondini and D.~Mantzavinos}, {\em Long-time asymptotics for the
  focusing nonlinear schr{\"o}dinger equation with nonzero boundary conditions
  at infinity and asymptotic stage of modulational instability}, Comm. Pure
  Appl. Math., 70 (2017), pp.~2300--2365.

\bibitem{Biondini-CPAM-2017}
\leavevmode\vrule height 2pt depth -1.6pt width 23pt, {\em Long-time
  asymptotics for the focusing nonlinear schr{\"o}dinger equation with nonzero
  boundary conditions at infinity and asymptotic stage of modulational
  instability}, Comm. Pure Appl. Math., 70 (2017), pp.~2300--2365.

\bibitem{monvel2022focusing}
{\sc A.~Boutet~de Monvel, J.~Lenells, and D.~Shepelsky}, {\em The focusing
  {NLS} equation with step-like oscillating background: the genus 3 sector},
  Comm. Math. Phys., 390 (2022), pp.~1081--1148.

\bibitem{buckingham2007long}
{\sc R.~Buckingham and S.~Venakides}, {\em Long-time asymptotics of the
  nonlinear schr{\"o}dinger equation shock problem}, Comm. Pure Appl. Math., 60
  (2007), pp.~1349--1414.

\bibitem{chen1976solitons}
{\sc H.~H. Chen and C.~S. Liu}, {\em Solitons in nonuniform media}, Phys. Rev.
  Lett., 37 (1976), p.~693.

\bibitem{deift1997new}
{\sc P.~Deift, S.~Venakides, and X.~Zhou}, {\em New results in small dispersion
  kdv by an extension of the steepest descent method for riemann-hilbert
  problems}, Int. Math. Res. Not. IMRN, 1997 (1997), pp.~285--299.

\bibitem{deift1993steepest}
{\sc P.~Deift and X.~Zhou}, {\em A steepest descent method for oscillatory
  riemann--hilbert problems. asymptotics for the mkdv equation}, Ann. of Math.,
  137 (1993), pp.~295--368.

\bibitem{doktorov2007dressing}
{\sc E.~V. Doktorov and S.~B. Leble}, {\em A dressing method in mathematical
  physics}, vol.~28, Springer Science \& Business Media, 2007.

\bibitem{PRL-2019}
{\sc A.~Gelash, D.~Agafontsev, V.~Zakharov, G.~El, S.~Randoux, and P.~Suret},
  {\em Bound state soliton gas dynamics underlying the spontaneous modulational
  instability}, Phys. Rev. Lett., 123 (2019), p.~234102.

\bibitem{Fritz-1992}
{\sc F.~Gesztesy, W.~Karwowski, and Z.~Zhao}, {\em Limits of soliton
  solutions}, Duke Math. J., 68 (1992), pp.~101--150.

\bibitem{girotti2021rigorous}
{\sc M.~Girotti, T.~Grava, R.~Jenkins, and K.-R. McLaughlin}, {\em Rigorous
  asymptotics of a kdv soliton gas}, Comm. Math. Phys., 384 (2021),
  pp.~733--784.

\bibitem{guo2012nonlinear}
{\sc B.~Guo, L.~Ling, and Q.~P. Liu}, {\em Nonlinear schr{\"o}dinger equation:
  generalized darboux transformation and rogue wave solutions}, Phys. Rev. E,
  85 (2012), p.~026607.

\bibitem{JMP-1995}
{\sc S.~Kamvissis}, {\em Focusing nonlinear schr\"{o}dinger equation with
  infinitely many solitons}, J. Math. Phys., 36 (1995), pp.~4175--4180.

\bibitem{kamvissis2003semiclassical}
{\sc S.~Kamvissis, K.-R. McLaughlin, and P.~D. Miller}, {\em Semiclassical
  soliton ensembles for the focusing nonlinear {S}chr\"{o}dinger equation},
  vol.~154 of Annals of Mathematics Studies, Princeton University Press,
  Princeton, NJ, 2003.

\bibitem{kotlyarov2017planar}
{\sc V.~Kotlyarov and D.~Shepelsky}, {\em Planar unimodular baker-akhiezer
  function for the nonlinear schr{\"o}dinger equation}, Ann. Math. Sci. Appl.,
  2 (2017), pp.~343--384.

\bibitem{kuznetsov1977solitons}
{\sc E.~A. Kuznetsov}, {\em Solitons in a parametrically unstable plasma}, Sov.
  Phys.-Dokl. (Engl. Transl.); (United States), 236 (1977), pp.~575--577.

\bibitem{ling2016multi}
{\sc L.~Ling, B.~Feng, and Z.~Zhu}, {\em Multi-soliton, multi-breather and
  higher order rogue wave solutions to the complex short pulse equation}, Phys.
  D, 327 (2016), pp.~13--29.

\bibitem{Ling-arxiv-2021}
{\sc L.~Ling and X.~Zhang}, {\em Large and infinite order solitons of the
  coupled nonlinear schr\" odinger equation}, arXiv preprint arXiv:2103.15373,
  (2021).

\bibitem{lyng2007n}
{\sc G.~D. Lyng and P.~D. Miller}, {\em The n-soliton of the focusing nonlinear
  schr{\"o}dinger equation for n large}, Comm. Pure Appl. Math., 60 (2007),
  pp.~951--1026.

\bibitem{ma1979perturbed}
{\sc Y.~C. Ma}, {\em The perturbed plane-wave solutions of the cubic
  schr{\"o}dinger equation}, Stud. Appl. Math., 60 (1979), pp.~43--58.

\bibitem{CPAM-1975}
{\sc H.~P. McKean and E.~Trubowitz}, {\em Hill's operator and hyperelliptic
  function theory in the presence of infinitely many branch points}, Comm. Pure
  Appl. Math., 29 (1976), pp.~143--226.

\bibitem{miller2008Riemann}
{\sc P.~D. Miller}, {\em Riemann-{H}ilbert problems with lots of discrete
  spectrum}, in Integrable systems and random matrices, vol.~458 of Contemp.
  Math., Amer. Math. Soc., Providence, RI, 2008, pp.~163--181.

\bibitem{OLMEDILLA1987330}
{\sc E.~Olmedilla}, {\em Multiple pole solutions of the non-linear
  schr{\"o}dinger equation}, Phys. D, 25 (1987), pp.~330--346.

\bibitem{satsuma1974b}
{\sc J.~Satsuma and N.~Yajima}, {\em B. initial value problems of
  one-dimensional self-modulation of nonlinear waves in dispersive media},
  Prog. Thero. Phys. Supp., 55 (1974), pp.~284--306.

\bibitem{Schiebold_2010}
{\sc C.~Schiebold}, {\em The noncommutative akns system: projection to matrix
  systems, countable superposition and soliton-like solutions}, J. Phys. A, 43
  (2010), p.~434030.

\bibitem{schiebold2017asymptotics}
\leavevmode\vrule height 2pt depth -1.6pt width 23pt, {\em Asymptotics for the
  multiple pole solutions of the nonlinear schr{\"o}dinger equation},
  Nonlinearity, 30 (2017), p.~2930.

\bibitem{Shabat_1992}
{\sc A.~Shabat}, {\em The infinite-dimensional dressing dynamical system},
  Inverse Prob., 8 (1992), p.~303.

\bibitem{shabat1972exact}
{\sc A.~Shabat and V.~Zakharov}, {\em Exact theory of two-dimensional
  self-focusing and one-dimensional self-modulation of waves in nonlinear
  media}, Sov. Phys. JETP, 34 (1972), p.~62.

\bibitem{terng2000backlund}
{\sc C.~L. Terng and K.~Uhlenbeck}, {\em B{\"a}cklund transformations and loop
  group actions}, Comm. Pure Appl. Math., 53 (2000), pp.~1--75.

\bibitem{wang1989special}
{\sc Z.~Wang and D.~Guo}, {\em Special functions}, world scientific, 1989.

\bibitem{yang2010nonlinear}
{\sc J.~Yang}, {\em Nonlinear waves in integrable and nonintegrable systems},
  SIAM, 2010.

\bibitem{zhang2021asymptotic}
{\sc X.~Zhang and L.~Ling}, {\em Asymptotic analysis of high-order solitons for
  the hirota equation}, Phys. D, 426 (2021), p.~132982.

\bibitem{zhang2020study}
{\sc Y.~Zhang, X.~Tao, T.~Yao, and J.~He}, {\em The regularity of the multiple
  higher-order poles solitons of the {NLS} equation}, Stud. Appl. Math., 145
  (2020), pp.~812--827.

\end{thebibliography}

\end{document}